\pdfoutput=1
\RequirePackage{ifpdf}
\ifpdf 
\documentclass[pdftex]{sigma}
\else
\documentclass{sigma}
\fi

\numberwithin{equation}{section}

\newtheorem{Theorem}{Theorem}[section]

\newtheorem{Lemma}[Theorem]{Lemma}
\newtheorem{Proposition}[Theorem]{Proposition}

 { \theoremstyle{definition}
\newtheorem{Definition}[Theorem]{Definition}

\newtheorem{Example}[Theorem]{Example}
\newtheorem{Remark}[Theorem]{Remark} }

\DeclareMathOperator{\Res}{Res}

\begin{document}
\allowdisplaybreaks

\newcommand{\arXivNumber}{2012.02622}

\renewcommand{\thefootnote}{}

\renewcommand{\PaperNumber}{085}

\FirstPageHeading

\ShortArticleName{Perturbative and Geometric Analysis of the Quartic Kontsevich Model}

\ArticleName{Perturbative and Geometric Analysis\\ of the Quartic Kontsevich Model\footnote{This paper is a~contribution to the Special Issue on Algebraic Structures in Perturbative Quantum Field Theory in honor of Dirk Kreimer for his 60th birthday. The~full collection is available at \href{https://www.emis.de/journals/SIGMA/Kreimer.html}{https://www.emis.de/journals/SIGMA/Kreimer.html}}}

\Author{Johannes BRANAHL~$^{\rm a}$, Alexander HOCK~$^{\rm b}$ and Raimar WULKENHAAR~$^{\rm a}$}

\AuthorNameForHeading{J.~Branahl, A.~Hock and R.~Wulkenhaar}

\Address{$^{\rm a)}$~Mathematisches Institut, Westf\"alische Wilhelms-Universit\"at M\"unster,\\
\hphantom{$^{\rm a)}$}~Einsteinstr.~62, 48149 M\"unster, Germany}
\EmailD{\href{mailto:j_bran33@uni-muenster.de}{j\_bran33@uni-muenster.de}, \href{mailto:raimar@math.uni-muenster.de}{raimar@math.uni-muenster.de}}
\URLaddressD{\url{https://www.uni-muenster.de/MathPhys/}}

\Address{$^{\rm b)}$~Mathematical Institute, University of Oxford,\\
\hphantom{$^{\rm b)}$}~Andrew Wiles Building, Woodstock Road, OX2 6GG, Oxford, UK}
\EmailD{\href{mailto:alexander.hock@maths.ox.ac.uk}{alexander.hock@maths.ox.ac.uk}}

\ArticleDates{Received February 26, 2021, in final form September 10, 2021; Published online September 16, 2021}

\Abstract{The analogue of Kontsevich's matrix Airy function, with the cubic potential~$\operatorname{Tr}\big(\Phi^3\big)$ replaced by a quartic term $\operatorname{Tr}\big(\Phi^4\big)$ with the same covariance, provides a toy model for quantum field theory in which all correlation functions can be computed exactly and explicitly. In this paper we show that distinguished polynomials of correlation functions, themselves given by quickly growing series of Feynman ribbon graphs, sum up to much simpler and highly structured expressions. These expressions are deeply connected with meromorphic forms conjectured to obey blobbed topological recursion. Moreover, we show how the exact solutions permit to explore critical phenomena in the quartic Kontsevich model.}


\Classification{81T18; 81T16; 14H81; 32A20}

\begin{flushright}
\begin{minipage}{60mm}
\it Dedicated to Dirk Kreimer\\ on the occasion of his 60th birthday
\end{minipage}
\end{flushright}

\renewcommand{\thefootnote}{\arabic{footnote}}
\setcounter{footnote}{0}

\section{Introduction}

Quantum field theory has often been a source of inspiration for
mathematics. In the previous 25 years, many of these inspirations
came from Dirk Kreimer. We~mention the vision \cite{Kreimer:1996js} of
a~deep relation between Feynman graphs and knots which led to
impressive progress on multiple zeta values
\cite{Broadhurst:1996kc}. The discovery that renormalisation in
quantum field theory is encoded in a~Hopf algebra
\cite{Kreimer:1997dp} led to the insight that renormalisation is
another example for the Birkhoff decomposition to solve a
Riemann--Hilbert problem \cite{Connes:1999yr}. There is much more to
say, but we confine ourselves to highlighting just one point: Although
the Hopf algebra was originally defined with Feynman graphs, it was
emphasised very soon \cite{Broadhurst:2000dq} that Dyson--Schwinger
equations will eventually provide a non-perturbative formulation.

One may ask whether multiple zeta values and other connections between
quantum field theory and number theory also find a non-perturbative
explanation. We~are working on a~programme which achieves and
investigates the exact solution of a quantum field theory toy model,
namely of a matrix model with quartic interaction and non-trivial
covariance~\cite{Grosse:2012uv}. It is already known that for
particular choices of parameters the exact solution of the planar
sector expands into number-theoretic objects such as Nielsen
polylogarithms \cite{Panzer:2018tvy} and hyperlogarithms
\cite{Grosse:2019qps}, respectively.

It is highly desirable to extend this construction to richer
topological sectors, which can be seen as analogy to knots. This
contribution provides the first steps in this direction. We~give a~low-order perturbative expansion of exact correlation
functions, derived in~\cite{Branahl:2020yru}, and compare the result
with a Feynman graph evaluation. We~perform this investigation in a
finite-dimensional case where no renormalisation is needed. We~show
that even this simple case has rich features, for instance an
enormous simplification in particular polynomials of correlation functions
(or Feynman graphs) compared with individual functions or graphs. We~expect that these simplifications will extend to
an infinite-dimensional limit where renormalisation is necessary,
although considerable work is still ahead.

\section{The model}

We sketch the main ideas about the model under consideration and refer
to \cite{Branahl:2020yru, Schurmann:2019mzu-v3} for more details. We~follow the paragon of the $\lambda\phi^4$ model, but defined on a
noncommutative space instead of on a Riemannian or Lorentzian
manifold. Apart from physical reasons, choosing a noncommutative
geometry has the advantage of a simple finite-dimensional
approximation. Let $H_N$ be the real vector space of self-adjoint
$(N\times N)$-matrices, $H_N'$ be its dual and $(e_{kl})$ be the
standard matrix basis in the complexification of $H_N$. Our quantum
scalar fields are noncommutative random variables $\Phi$ on $H_N'$
distributed according to a measure
\begin{gather}
{\rm d}\mu_{E,\lambda}(\Phi)=\frac{1}{\mathcal{Z}}\exp\bigg({-}\frac{\lambda N}{4}\operatorname{Tr}\big(\Phi^4\big)\bigg)
{\rm d}\mu_{E,0}(\Phi),\nonumber
\\
\mathcal{Z}:=\int_{H_N'} \exp\bigg({-}\frac{\lambda N}{4}\operatorname{Tr}\big(\Phi^4\big)\bigg)
{\rm d}\mu_{E,0}(\Phi),
\label{measure}
\end{gather}
where ${\rm d}\mu_{E,0}(\Phi)$ is a Gau\ss{}ian measure with covariance
$\big[\int_{H_N'}\! {\rm d}\mu_{E,0}(\Phi)\;
\Phi(e_{jk})\Phi(e_{lm})\big]_c=\frac{\delta_{jm}\delta_{kl}}{N(E_j+E_l)}$
for some $0<E_1<\dots <E_N$. We~call the Euclidean quantum field
theory defined via (\ref{measure}) the \emph{quartic Kontsevich model}
because of its formal analogy with the Kontsevich model
\cite{Kontsevich:1992ti} in~which $\frac{\lambda}{4} \operatorname{Tr}\big(\Phi^4\big)$ in
(\ref{measure}) is replaced by\footnote{In the Kontsevich model
with potential $\operatorname{Tr}\big(\Phi^3\big)$, a purely imaginary coupling constant
is necessary for convergence of $\mathcal{Z}$. Of course,
${\rm d}\mu_{E,\lambda}$ is then only a signed measure. Choosing
$\lambda>0$ in the quartic model
with potential $\operatorname{Tr}\big(\Phi^4\big)$ gives both a convergent
partition function $\mathcal{Z}$ and a true measure
${\rm d}\mu_{E,\lambda}$.}
$\frac{\mathrm{i}}{6} \operatorname{Tr}\big(\Phi^3\big)$.
The Gau\ss{}ian measure ${\rm d}\mu_{E,0}(\Phi)$ is the same.
Kontsevich proved in~\cite{Kontsevich:1992ti} that (\ref{measure}) with
$\operatorname{Tr}\big(\Phi^3\big)$-term, viewed as function of the $E_k$, is the
generating function for intersection numbers of tautological
characteristic classes on the moduli space
$\overline{\mathcal{M}}_{g,n}$ of stable complex curves.

Derivatives of the Fourier transform
$\mathcal{Z}(M):=\int_{H_N'} {\rm d}\mu_{E,\lambda}(\Phi) {\rm e}^{\mathrm{i}\Phi(M)}$ with respect to
matrix entries $M_{kl}$ and parameters $E_k$ of the free theory give rise to
\emph{Dyson--Schwinger equations} between the cumulants
\begin{gather}
\langle e_{k_1l_1}\cdots e_{k_nl_n}\rangle_c
= \frac{1}{\mathrm{i}^{n}}
\frac{\partial^n\log \mathcal{Z}(M)}{\partial
M_{k_1l_1}\cdots \partial M_{k_nl_n}} \bigg|_{M=0}.
\label{cumulants}
\end{gather}
Of particular interest are cumulants of the form
\begin{gather}
N^{n_1+\cdots +n_b}
\big\langle (e_{k_1^1k_2^1}
e_{k_2^1k_3^1} \cdots
e_{k_{n_1}^1k_1^1}) \cdots
(e_{k_1^bk_2^b} e_{k_2^bk_3^b} \cdots
e_{k_{n_b}^bk_1^b}) \big\rangle_c\nonumber
\\ \qquad
{}=: N^{2-b} \, G_{|k_1^1\cdots k_{n_1}^1|\cdots
|k_1^b\cdots k_{n_b}^b|} ,
\label{eq:Gbn}
\end{gather}
called $(n_1+\dots +n_b)$-point functions. To define these functions properly it is
necessary that the $k^j_i$ are pairwise different. After their identification a natural extension
to any diagonal is possible. The corresponding
derivatives in (\ref{cumulants}) then decompose into linear combinations of such
functions. One has, for example,
$-N^2 \frac{\partial^2 \log \mathcal{Z}(M)}{\partial M_{kk} \partial M_{kk}}
= NG_{|kk|}+G_{|k|k|}$.

After $1/N$-expansion
$G_{|k_1^1\cdots k_{n_1}^1|\cdots |k_1^b\cdots k_{n_b}^b|} =:
\sum_{g=0}^\infty N^{-2g} G^{(g)}_{|k_1^1\cdots k_{n_1}^1|\cdots
 |k_1^b\cdots k_{n_b}^b|} $ of the correlation functions (\ref{eq:Gbn})
one obtains a non-linear equation for the planar 2-point function
$G_{|kl|}^{(0)}$ alone \cite{Grosse:2009pa} and a hierarchy of affine
equations for all other functions. The arduous solution process for
$G^{(0)}_{|kl|}$ was recently completed in~\cite{Grosse:2019jnv,Panzer:2018tvy}.

Then things accelerated: During the attempt of finding an elegant
algorithm to cover any correlation function, we recognised that we
were somehow looking for the wrong quantities: A~non-trivial
rearrangement \cite{Branahl:2020yru} of those gives birth to
meromorphic differential forms $\omega_{g,n}$ label\-led by genus $g$
and number $n$ of marked points of a Riemann surface. The solution of
the complicated Dyson--Schwinger equations for $\omega_{g,n}$ at small
$2g+n-2$ in~\cite{Branahl:2020yru} provided strong
evidence\footnote{Two of us have proved the algebraic structure
 of blobbed topological recursion for the genus $g=0$ case
 in~\cite{Hock:2021tbl}.} for a~remarkable algebraic structure behind the model under consideration:
({\it blobbed})~\cite{Borot:2015hna} {\it topological recursion}
\cite{Eynard:2007kz}. As a consequence, the $\omega_{g,n}$ with
$2g+n-2<0$ are recursively built from $\omega_{0,1}$ and
$\omega_{0,2}$ by a relatively simple evaluation of residues, much
faster than solving the Dyson--Schwinger equations. Topological
recursion has been identified in numerous areas of mathematics and
physics including one- and two-matrix models \cite{Chekhov:2006vd},
Hurwitz theory \cite{Bouchard:2007hi} and Gromov--Witten theory
\cite{Bouchard:2007ys}. Topological recursion also governs the
combinatorics of the Kontsevich model \cite{Kontsevich:1992ti} (see,
e.g., \cite[Chapter~6]{Eynard:2016yaa} for details) and organises the
Weil--Petersson volumes of moduli spaces of hyperbolic Riemann surfaces
\cite{Mirzakhani:2006fta}.

We discuss in Section~\ref{perturb} the perturbative expansion of
correlation functions~(\ref{eq:Gbn}) into weighted labelled ribbon
graphs. Section~\ref{sec:aux} shows that two families of auxiliary
functions $T^{(g)}$ and $\Omega^{(g)}$ introduced in~\cite{Branahl:2020yru} are representable as polynomials in correlation
functions. Section~\ref{sec:BTR} compares the Taylor series of exact
results for $\Omega^{(g)}$ with the ribbon graph expansion of the
correlation functions. It is impressive to see how contributions of a
huge number of ribbon graphs almost cancel up to a tiny and structured
remnant which is conjectured to obey blobbed topological recursion.
In Section~\ref{sec:critical} we start a (partly numerical)
investigation of critical phenomena in the quartic Kontsevich model.
The number of branch cuts and the order of ramification points changes
at critical values of the coupling constant. Interestingly, the
correlation functions cross analytically into the other phases. We~conclude in Section~\ref{sec:conclusion} with possible lessons for
more realistic quantum field theories.

\subsection[Differences between the quartic Kontsevich model and the generalised Kontsevich model]{Differences between the quartic Kontsevich model\\ and the generalised Kontsevich model}

The notation of the generalised Kontsevich model (GKM) is well-defined in the literature. The~ori\-ginal motivation came from Witten \cite{Witten:1993mgm} in terms of the $r$-spin classes on the moduli space of complex curves, $r\in \mathbb{N}_{>1}$. Also the corresponding formal matrix model representation is well-known \cite{Adler:1992tj}. Its partition function is defined by \cite{Belliard:2021jtj}
\begin{gather}\label{rspin}
	\mathcal{ Z}^{r{\rm spin}}:=\int_{H_N} {\rm d}M {\rm e}^{-N\alpha^{r+1}\operatorname{Tr}(V(M)-V(\Lambda)-(M-\Lambda)V'(\Lambda))},
\end{gather}
where $V(M)=\frac{M^{r+1}}{r+1}$ is the potential, $\alpha$ a formal parameter and $\Lambda$ the diagonal matrix with positive eigenvalues $\Lambda_1,\dots,\Lambda_N$.

In case of $r=2$, this model becomes exactly the Kontsevich model \cite{Kontsevich:1992ti}
\begin{gather*}
\mathcal{ Z}^{2{\rm spin}}=\int_{H_N} {\rm d}M {\rm e}^{-N\alpha^{3}\operatorname{Tr} \left(\frac{M^{3}}{3}-M\Lambda^{2}+\frac{2\Lambda^3}{3}\right)}
=\int_{H_N} {\rm d}\tilde{M} {\rm e}^{-N\alpha^{3}\operatorname{Tr} \left(\frac{\tilde{M}^{3}}{3}+\tilde{M}^2\Lambda\right)}
\end{gather*}
with the transformation $M=\tilde{M}+\Lambda$ in the last equality. The last representation in terms of $\tilde{M}$ leads to the combinatorics of weighted ribbon graphs with 3-valent vertices (see Section~\ref{perturb}).

However, avoiding for $r>2$ the linear term in $M$ in the exponential in \eqref{rspin} by a transformation $M=\tilde{M}+\Lambda$ one gets a very restricted matrix model. For instance, the $r=3$ case~is
\begin{gather}\label{3spin}
\mathcal{ Z}^{3{\rm spin}}=\int_{H_N} {\rm d}M {\rm e}^{-N\alpha^{4}\operatorname{Tr} \left(\frac{M^{4}}{4}-M\Lambda^{3}+\frac{3\Lambda^4}{4}\right)}
=\int_{H_N} {\rm d}\tilde{M} {\rm e}^{-N\alpha^{4}\operatorname{Tr} \left(\frac{\tilde{M}^{4}}{4}+\tilde{M}^3\Lambda+\frac{3}{2}\tilde{M}^2\Lambda^2\right)}.
\end{gather}
We emphasise that the $r$-spin matrix model is proved to satisfy
topological recursion for the combinatorics of the resolvents~\cite{Eynard:2008we}. Also for the expectation values of disjoint
cycles it is proved in~\cite{Belliard:2021jtj} that it is governed by
higher topological recursion of Bouchard and Eynard~\cite{Bouchard:2012yg}. These two results are connected by exchanging
the rôle of $x$ and $y$ in topological recursion.

Now, the model considered in this paper, defined via~\eqref{measure} and the same covariance as the Kontsevich model,
has the following
matrix model representation, achieved by the canonical duality between
the vector space $H_N$ and its dual $H_N'$:
\begin{gather*}
 \mathcal{ Z}=\int_{H_N} {\rm d}\tilde{M} \;
 {\rm e}^{-N\operatorname{Tr}\left(\frac{\lambda \tilde{M}^{4}}{4}+\frac{1}{2}\tilde{M}^2E\right)},
\end{gather*}
where the diagonal matrix $E$ has the entries $0<E_1<\dots <E_N$.

We see that the quartic Kontsevich model does not fit into the class
of $r$-spin models, even for arbitrary polynomial potentials in
\eqref{rspin}. For instance, the cubic term of $\tilde{M}$ in
\eqref{3spin} is also weighted by the external matrix $\Lambda$ and
its appearance is indispensable. It turns out that the more
complicated enumerative structure of $r$-spin models has the easier
algebraic structure in terms of (higher-order) topological recursion,
whereas the easier enumerative structure of the quartic Kontsevich
model is governed by the enriched structure of blobbed topological
recursion.

In the subsequent section, the combinatorics of the quartic Kontsevich
model is discussed. For the combinatorics of the $r$-spin model we
refer to the recent work \cite{Belliard:2021jtj}.

\section{Perturbation theory}\label{perturb}

\subsection{Weighted labelled ribbon graphs}

The expansion of $\exp\big({-}\frac{\lambda N}{4}\operatorname{Tr}\big(\Phi^4\big)\big)$ inside the
measure ${\rm d}\mu_{E,\lambda}(\Phi)$ defined in (\ref{measure}) represents the cumulants
(\ref{cumulants}) as a series
\begin{gather}
\langle e_{p_1q_1}\cdots e_{p_nq_n}\rangle_c\nonumber
\\ \qquad
{}=\! \sum_{v=0}^\infty \frac{N^v(-\lambda)^v}{4^vv!}
\bigg[\int_{H_N'}\!\!\!\! {\rm d}\mu_{E,0}(\Phi)
\Phi_{p_1q_1} \cdots \Phi_{p_nq_n} \!\!\!\!\!
\sum_{j_1,\dots,m_v=1}^N
\prod_{i=1}^v\big(\Phi_{j_ik_i}\Phi_{k_il_i}\Phi_{l_im_i}\Phi_{m_ij_i}\big)\bigg]_c\!,\!\!
\label{cumulants-1}
\end{gather}
where $\Phi_{kl}:=\Phi(e_{kl})$ and $[~]_c$ means taking the connected
part. We~fix the order $v$ and rest\-rict our attention to the case that
$p_1,\dots,p_n$ are pairwise different. By the definition of the
Gau\ss{}ian measure ${\rm d}\mu_{E,0}(\Phi)$, this integral is zero for
$n$ odd, whereas for $n$ even it evaluates into a sum over all
partitions of
$\Phi_{p_1q_1} \cdots \Phi_{p_nq_n} \sum_{j_1,\dots,m_v=1}^N
\prod_{i=1}^v\big(\Phi_{j_ik_i}\Phi_{k_il_i}\Phi_{l_im_i}\Phi_{m_ij_i}\big)$
into products of pairs with a pair $(\Phi_{jk}\Phi_{lm})$ replaced by
$\frac{\delta_{jm}\delta_{kl}}{N(E_j+E_l)}$.

Every pairing contributing to (\ref{cumulants-1}) has a convenient
visualisation as a ribbon graph. Its~buil\-ding blocks are $n$
one-valent vertices representing $\Phi_{p_1q_1},\dots,\Phi_{p_nq_n}$
and $v$ four-valent vertices representing
$\Phi_{j_ik_i}\Phi_{k_il_i}\Phi_{l_im_i}\Phi_{m_ij_i}$ for
$i=1,\dots,v$ as well as $r$ ribbons which connect the vertices. A~ribbon represents a pair $(\Phi_{jk}\Phi_{lm})$; it is drawn as a double line
\begin{picture}(20,15)
 \put(0,5){\line(1,0){20}}
 \put(0,9){\line(1,0){20}}
 \put(8,-1){\mbox{\scriptsize$j$}}
 \put(8,11){\mbox{\scriptsize$l$}}
\end{picture}
between the vertices (can be the same) at which $\Phi_{jk}$ and
$\Phi_{lm}$ are located. The two strands of this double line are labelled~$j$ and~$l$, respectively. A~strand is left open at a one-valent
vertex, whereas at a four-valent vertex we connect it
with the strand of the neighboured ribbon. A~four-valent vertex with
its attached ribbons thus looks as
\begin{picture}(24,19)
 \put(0,6){\line(1,0){10}}
 \put(0,10){\line(1,0){10}}
 \put(24,6){\line(-1,0){10}}
 \put(24,10){\line(-1,0){10}}
 \put(10,-4){\line(0,1){10}}
 \put(14,-4){\line(0,1){10}}
 \put(10,20){\line(0,-1){10}}
 \put(14,20){\line(0,-1){10}}
 \put(5,-1.5){\mbox{\scriptsize$j$}}
 \put(15,-1.5){\mbox{\scriptsize$k$}}
 \put(15,12.5){\mbox{\scriptsize$l$}}
 \put(2,12.5){\mbox{\scriptsize$m$}}
\end{picture}. A~ribbon graph is connected when any two vertices (one- or four-valent)
are connected by a chain of ribbons. We~only retain the connected
ribbon graphs in~(\ref{cumulants-1}).

\begin{figure}[h!t]
	\centering
	\def\svgwidth{500pt}
	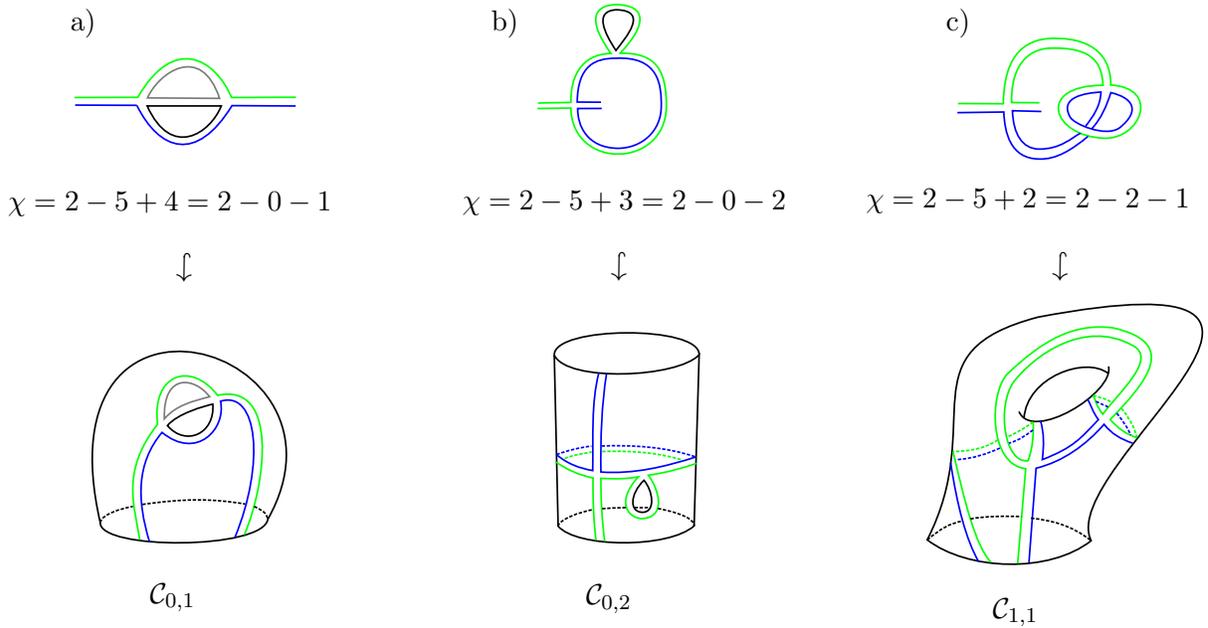\vspace*{-0.8cm}
	\caption{Three different ribbon graphs together with the associated
 Riemann surfaces $\mathcal{C}_{g,b}$. The external strands are coloured
 in green and blue. The topology of $\mathcal{C}_{g,b}$ is computed
 by $\chi=v-r+(n+s)=2-2g-b$.}	\label{fig:RibbonEx1}
\end{figure}

The upper row of Figure \ref{fig:RibbonEx1} shows three examples of ribbon graphs with $n=2$ one-valent vertices,
$v=2$ four-valent vertices and $r=5$ ribbons. In general, this construction lets the
strands connect to $n$ open lines and a certain number $s$ loops. Due
to the Kronecker-$\delta$s from pairs and vertices, every one of the
$n+s$ lines or loops carries a unique label. Every loop is labelled by
a summation index which is the remnant of the summation
$\sum_{j_1,\dots,m_v=1}^N$ in (\ref{cumulants-1}) after taking all
Kronecker-$\delta$s into account. The open lines are labelled by the
first matrix indices $p_1,\dots,p_n$ of the product
$\Phi_{p_1q_1} \cdots \Phi_{p_nq_n}$ in (\ref{cumulants-1}), and the
integral (\ref{cumulants-1}) is zero unless there is a permutation
$\pi \in \mathcal{S}_n$ with $q_i=\pi(p_i)$ for all $i=1,\dots,n$. This
permutation $\pi$ is uniquely\footnote{Uniqueness of $\pi$
 is the reason for choosing $p_1,\dots,p_n$ pairwise different.}
determined by the ribbon graph, simply
by looking at the strand labels of the ribbon at every one-valent
vertex.

Next, observe that due to the symmetry of the product
$\prod_{i=1}^v\big(\Phi_{j_ik_i}\Phi_{k_il_i}\Phi_{l_im_i}\Phi_{m_ij_i}\big)$
in (\ref{cumulants-1}) under cyclic permutation
$j_i\to k_i\to l_i\to m_i \to j_i$ of order $4$ at every vertex and
the $v!$ permutations of the vertices there are $4^v v!$ pairings which
give the same labelled ribbon graph. We~can thus omit the factor
$\frac{1}{4^v v!}$ and sum only over the different labelled ribbon
graphs. This expresses the integral (\ref{cumulants-1}) as
follows:
\begin{Proposition}\label{prop:weight}
 Let $p_1,\dots,p_n$ be pairwise different and
 $\mathfrak{G}^{v,\pi}_{p_1,\dots,p_n}$ be the set of labelled
 connec\-ted ribbon graphs with $v$ four-valent vertices and $n$
 one-valent vertices labelled $(p_1,\pi(p_1)), \dots,\allowbreak
 (p_n,\pi(p_n))$. Then for $n$ even the integral \eqref{cumulants-1} evaluates to
\begin{align}
N^n \langle e_{p_1\pi(p_1)}\dots e_{p_n\pi(p_n)}\rangle_c
 =\sum_{v=0}^\infty \sum_{\Gamma \in \mathfrak{G}^{v,\pi}_{p_1,\dots,p_n}}
 N^{v-r+n+s(\Gamma)} \varpi(\Gamma),
\label{cumulants-2}
\end{align}
where $r=2v+n/2$ is the number of ribbons,
$s(\Gamma)$ the number of loops in $\Gamma$ and
the weight $\varpi(\Gamma)$ is derived from the following Feynman rules:
\begin{itemize}\itemsep=0pt
\item label the $s=s(\Gamma)$ loops by $k_1,\dots,k_s$;
\item associate a factor $-\lambda$ to a $4$-valent ribbon-vertex;
\item associate the factor $\frac{1}{E_p+E_q}$ to a ribbon with
 strands labelled by $p$, $q$;
\item multiply all factors and apply the summation
 operator $\frac{1}{N^s}\sum_{k_1,\dots,k_s=1}^N$.
\end{itemize}
\end{Proposition}

The exponent $\chi:=v-r+n+s(\Gamma)$ of $N$ in (\ref{cumulants-2}) has
a topological interpretation. Let~$b$ be the number of cycles in
$\pi$. We~take $b$ auxiliary faces, called boundary
components, and attach the one-valent vertices in cyclic order and
according to the cycle they belong to the circumference of the
boundary components. The edge between two neighboured vertices on a~boundary component closes an open line of $\Gamma$. In total this produces
$n$ additional loops. We~thus get a simplicial 2-complex of $v+n$ vertices,
$r+n$ edges and $b+n+s$ faces, hence of Euler characteristic
$2-2g=(v+n)-(r+n)+(b+n+s)$. This identifies the exponent
$\chi=2-2g-b=v-r+n+s$ of $N$ as the Euler characteristic of a bordered
Riemann surface~$\mathcal{C}_{g,b}(\Gamma)$. It is, up to homotopy,
uniquely defined\footnote{The dual graph of a connected ribbon graph
 $\Gamma$ associated to $\mathcal{C}_{g,b}$ is a quadrangulation (a
 map) of $\mathcal{C}_{g,b}$. Our definition of the correlation functions by
 disjoint cycles is equivalent to the definition used in~\cite{Borot:2017agy} for
 fully simple maps. Fully simple maps are a subset of ordinary maps which are usually
 studied in matrix models
 (see \cite{Borot:2017agy} for more details).} by the simplicial
2-complex encoded in $\Gamma$. The lower row of Figure
\ref{fig:RibbonEx1} sketches the bordered Riemann surfaces defined by
the corresponding ribbon graphs of the first row.

To compare with the solution via Dyson--Schwinger equations we give an
equivalent formulation of (\ref{cumulants-2}). A~ribbon graph
$\Gamma\in \mathfrak{G}^{v,\pi}_{p_1,\dots,p_n}$ contains full
information about the cycle type of~$\pi \in \mathcal{S}_n$ and about
which of the $p_1,\dots,p_n$ label a chosen cycle of $\pi$. We~rearrange these labels as follows. Set $p^1_1:=p_1$ and then
recursively $p^1_k:=\pi^{1-k}\big(p_1^1\big)$ for $k=2,\dots,n_1$ if
$\pi^{-n_1}\big(p_1^1\big)=p_1^1$. Next relabel any of the not yet assigned
$p_k$ as $p^2_1$ and continue to set $p^2_k:=\pi^{1-k}\big(p_2^1\big)$ for
$k=2,\dots,n_2$ if $\pi^{-n_2}\big(p_2^1\big)=p_2^2$. Proceed until the
relabelling is complete. We~denote by
$\mathfrak{G}^{v}_{|p_1^1\cdots p_{n_1}^1|\cdots |p_1^b\cdots p_{n_b}^b|}$ the set
of relabelled ribbon graphs in $\mathfrak{G}^{v,\pi}_{p_1,\dots,p_n}$;
both sets are in one-to-one correspondence. We~further partition this
set as
$\mathfrak{G}^{v}_{|p_1^1\cdots p_{n_1}^1|\cdots |p_1^b\cdots p_{n_b}^b|}
=\bigcup_{g=0}^\infty
\mathfrak{G}^{g,v}_{|p_1^1\cdots p_{n_1}^1|\cdots |p_1^b\cdots p_{n_b}^b|}$ into
subsets of graphs of the same genus $g$. For fixed $v$, this union is
actually finite. With these preparations we can represent the series coefficients of
the genus expansion
\[
G_{|p_1^1\cdots p_{n_1}^1|\cdots |p_1^b\cdots p_{n_b}^b|}=\sum_{g=0}^\infty N^{-2g}
G^{(g)}_{|p_1^1\cdots p_{n_1}^1|\cdots |p_1^b\cdots p_{n_b}^b|}
\]
 of~(\ref{eq:Gbn})
for pairwise different $p_i^j\in \{1,\dots,N\}$
as
\begin{gather*}
G^{(g)}_{|p_1^1\cdots p_{n_1}^1|\cdots |p_1^b\cdots p_{n_b}^b|}
=\sum_{v=0}^\infty\, \sum_{\Gamma \in
 \mathfrak{G}^{g,v}_{|p_1^1\cdots p_{n_1}^1|\cdots |p_1^b\cdots p_{n_b}^b|}}
 \!\!\!\!\!\varpi(\Gamma).
\end{gather*}

\begin{Remark}
 One can define similar structures for the logarithm of the partition function itself,
\[
\log \mathcal{Z}=\log \int_{H_N'} {\rm d}\mu_{E,0}(\Phi)\,{\rm e}^{-\frac{\lambda
 N}{4} \operatorname{Tr}(\Phi^4)}.
 \]
 Let $\mathfrak{G}^{g,v}_\varnothing$ be the set of connected
 vacuum ribbon graphs of genus $g$ made of $v$ four-valent vertices
 and no one-valent vertices. Then the analogue of~\eqref{cumulants-2}
 is
\begin{align*}
 \log \mathcal{ Z}
 =\sum_{g=0}^\infty\sum_{v=0}^\infty
 \sum_{\Gamma_0\in\mathfrak{G}^{g,v}_\varnothing}
 \frac{N^{2-g}\varpi(\Gamma_0)}{|\mathrm{Aut}(\Gamma_0)|},
\end{align*}
where $\varpi(\Gamma_0)$ is given by the same Feynman rules as in
Proposition~\ref{prop:weight} and
$|\mathrm{Aut}(\Gamma_0)|$ is the order of the automorphism
group\footnote{The automorphism group of any ribbon graph
 $\Gamma$ with at
 least one boundary $b\geq 1$ is trivial, i.e.,
 $|\mathrm{Aut}(\Gamma)|=1$.} of the vacuum ribbon graph $\Gamma_0$.

Later in Definition~\ref{def:freeenergy} we will introduce the free energy. We~can perturbatively
establish
\begin{align}
 \mathcal{F}^{(g)}= -\frac{\delta_{g,0}}{2N^2} \sum_{k,l=1}^N \log(E_k+E_l)
 +\sum_{v=1}^\infty \sum_{\Gamma_0\in\mathfrak{G}^{g,v}_\varnothing}\frac{\varpi(\Gamma_0)}{
 |\mathrm{Aut}(\Gamma_0)|} .
 \label{freeenergy}
\end{align}
\end{Remark}

\subsection{Examples}

\begin{Example}
 For the ribbon graphs of Figure \ref{fig:RibbonEx1}, we label the
 green open line by $p_1$ and the blue open line by
 $p_2$. Consequently, the graphs become elements of
 $\mathfrak{G}^{0,2}_{|p_1p_2|}$, $\mathfrak{G}^{0,2}_{|p_1|p_2|}$,
 $\mathfrak{G}^{1,2}_{|p_1p_2|}$
 respectively. The weights $\varpi(\Gamma)$ associated to these
 ribbon graphs are
\begin{gather*}
a)\quad \frac{(-\lambda)^2}{N^2(E_{p_1}+E_{p_2})^2}
\sum_{k_1,k_2=1}^N\frac{1}{(E_{p_1}+E_{k_1})(E_{p_2}+E_{k_2})(E_{k_1}+E_{k_2})}
,
\\
b)\quad \frac{(-\lambda)^2}{N(E_{p_1}+E_{p_2})^2(2E_{p_1})(2E_{p_2})}
\sum_{k_1=1}^N \frac{1}{E_{p_1}+E_{k_1}},
\\
c)\quad \frac{(-\lambda)^2}{(E_{p_1}+E_{p_2})^3(2E_{p_1})(2E_{p_2})}.
\end{gather*}
\end{Example}

\begin{Example}\label{ExF0}
 The free energy $\mathcal{F}^{(0)}$ of genus $g=0$ consists of
 the empty ribbon graph with weight given by the first term in (\ref{freeenergy}) and
 4 ribbon graphs up to order $\lambda^2$ (see Figure \ref{fig:freeenergy}).
Taking weights and order of the automorphism groups into account, we have perturbatively
\begin{align*}
\mathcal{F}^{(0)}={}& \frac{-1}{2N^2}\sum_{k,l=1}^N \log(E_k+E_l)+
\frac{(-\lambda)}{2N^3}\sum_{k,l,m=1}^N\frac{1}{(E_k+E_l)(E_k+E_m)}
\\
&+\frac{(-\lambda)^2}{2N^4}\sum_{j,k,l,m=1}^N\frac{1}{
 (E_j+E_m)(E_j+E_k)^2}\Big(\frac{1}{E_j+E_l}+\frac{1}{E_k+E_l}\Big)
\\
&+\frac{(-\lambda)^2}{8N^4}\sum_{j,k,l,m=1}^N\frac{1}{(E_j+E_k)(E_k+E_l)(E_l+E_m)(E_m+E_j)}
+\mathcal{O}\big(\lambda^3\big).
\end{align*}
\end{Example}

\begin{figure}[h!]
	\centering
	\def\svgwidth{400pt}
\begingroup%
  \makeatletter%
  \providecommand\color[2][]{%
    \errmessage{(Inkscape) Color is used for the text in Inkscape, but the package 'color.sty' is not loaded}%
    \renewcommand\color[2][]{}%
  }%
  \providecommand\transparent[1]{%
    \errmessage{(Inkscape) Transparency is used (non-zero) for the text in Inkscape, but the package 'transparent.sty' is not loaded}%
    \renewcommand\transparent[1]{}%
  }%
  \providecommand\rotatebox[2]{#2}%
  \newcommand*\fsize{\dimexpr\f@size pt\relax}%
  \newcommand*\lineheight[1]{\fontsize{\fsize}{#1\fsize}\selectfont}%
  \ifx\svgwidth\undefined%
    \setlength{\unitlength}{513.86557695bp}%
    \ifx\svgscale\undefined%
      \relax%
    \else%
      \setlength{\unitlength}{\unitlength * \real{\svgscale}}%
    \fi%
  \else%
    \setlength{\unitlength}{\svgwidth}%
  \fi%
  \global\let\svgwidth\undefined%
  \global\let\svgscale\undefined%
  \makeatother%
  \begin{picture}(1,0.30156411)%
    \lineheight{1}%
    \setlength\tabcolsep{0pt}%
    \put(0.28905551,0.28161671){\makebox(0,0)[lt]{\lineheight{1.25}\smash{\begin{tabular}[t]{l}$\lambda^1$\end{tabular}}}}%
    \put(0.7372662,0.28199349){\makebox(0,0)[lt]{\lineheight{1.25}\smash{\begin{tabular}[t]{l}$\lambda^2$\end{tabular}}}}%
    \put(0,0){\includegraphics[width=\unitlength,page=1]{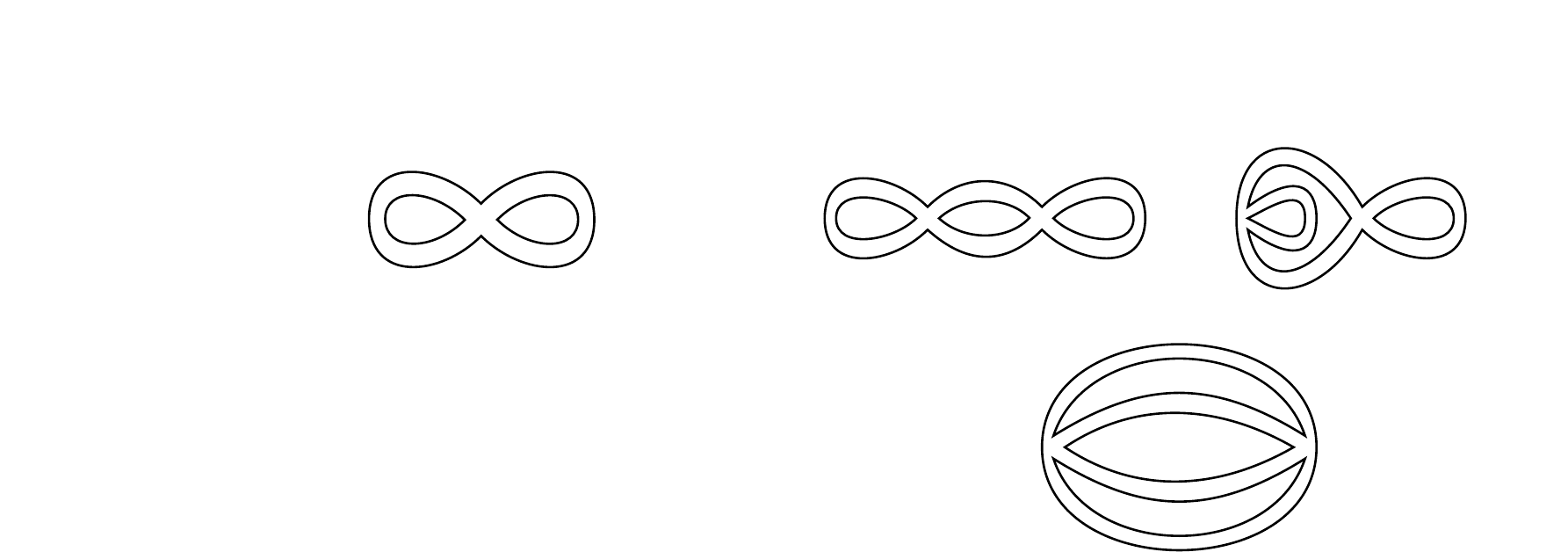}}%
    \put(0.04445554,0.28110333){\makebox(0,0)[lt]{\lineheight{1.25}\smash{\begin{tabular}[t]{l}$\lambda^0$\end{tabular}}}}%
    \put(0,0){\includegraphics[width=\unitlength,page=2]{freeenergy.pdf}}%
  \end{picture}%
\endgroup%

	\caption{The graph at order $\lambda^0$ is added as the empty
 ribbon graph. All these graphs contribute to the free energy
 of genus 0 up to order $\lambda^2$. These graphs are
 elements of $\mathfrak{G}^{0,v}_\varnothing$. The melon graph $\Gamma_M$ (in~the second row) has
 $ |\mathrm{Aut}(\Gamma_M)|=8$ and the other four graphs
 $ |\mathrm{Aut}(\Gamma)|=2$. }
	\label{fig:freeenergy}
 \end{figure}

\begin{Example}\label{Ex2}
 The first example with one boundary component is the 2-point function to which
 12 ribbon graphs contribute up to order $\lambda^2$ (see Figure \ref{fig:2PGraphs}).
 Taking the weights into account, we~have perturbatively
\begin{gather*}
G^{(0)}_{|ab|}=\frac{1}{E_a+E_b}+\frac{(-\lambda)}{(E_a+E_b)^2}
\frac{1}{N}\sum_{k=1}^N\bigg(\frac{1}{E_a+E_k}+\frac{1}{E_b+E_k}\bigg)
\\ \hphantom{G^{(0)}_{|ab|}=}
{}+\frac{(-\lambda)^2}{(E_a+E_b)^2}\frac{1}{N^2}\sum_{k,l=1}^N\bigg(
\frac{1}{(E_a+E_{k})^2(E_a+E_{l})}+\frac{2}{(E_a+E_{k})(E_b+E_{l})(E_a+E_b)}
\\ \hphantom{G^{(0)}_{|ab|}=+}
{} +\frac{1}{(E_b+E_{k})^2(E_b+E_{l})}+\frac{1}{(E_a+E_{k})^2(E_{k}+E_{l})}
+\frac{1}{(E_b+E_{k})^2(E_{k}+E_{l})}
\\ \hphantom{G^{(0)}_{|ab|}=+}
{} +\frac{1}{(E_a+E_{k})(E_a+E_{l})(E_a+E_b)}+\frac{1}{(E_b+E_{k})(E_b+E_{l})(E_a+E_b)}
\\ \hphantom{G^{(0)}_{|ab|}=+}
{} +\frac{1}{(E_a+E_{k})(E_b+E_{l})(E_{k}+E_{l})}\bigg)
+\mathcal{O}\big(\lambda^3\big).
\end{gather*}

\begin{figure}[h!]
	\centering
	\def\svgwidth{400pt}
\begingroup%
  \makeatletter%
  \providecommand\color[2][]{%
    \errmessage{(Inkscape) Color is used for the text in Inkscape, but the package 'color.sty' is not loaded}%
    \renewcommand\color[2][]{}%
  }%
  \providecommand\transparent[1]{%
    \errmessage{(Inkscape) Transparency is used (non-zero) for the text in Inkscape, but the package 'transparent.sty' is not loaded}%
    \renewcommand\transparent[1]{}%
  }%
  \providecommand\rotatebox[2]{#2}%
  \newcommand*\fsize{\dimexpr\f@size pt\relax}%
  \newcommand*\lineheight[1]{\fontsize{\fsize}{#1\fsize}\selectfont}%
  \ifx\svgwidth\undefined%
    \setlength{\unitlength}{446.93579482bp}%
    \ifx\svgscale\undefined%
      \relax%
    \else%
      \setlength{\unitlength}{\unitlength * \real{\svgscale}}%
    \fi%
  \else%
    \setlength{\unitlength}{\svgwidth}%
  \fi%
  \global\let\svgwidth\undefined%
  \global\let\svgscale\undefined%
  \makeatother%
  \begin{picture}(1,0.40734917)%
    \lineheight{1}%
    \setlength\tabcolsep{0pt}%
    \put(0,0){\includegraphics[width=\unitlength,page=1]{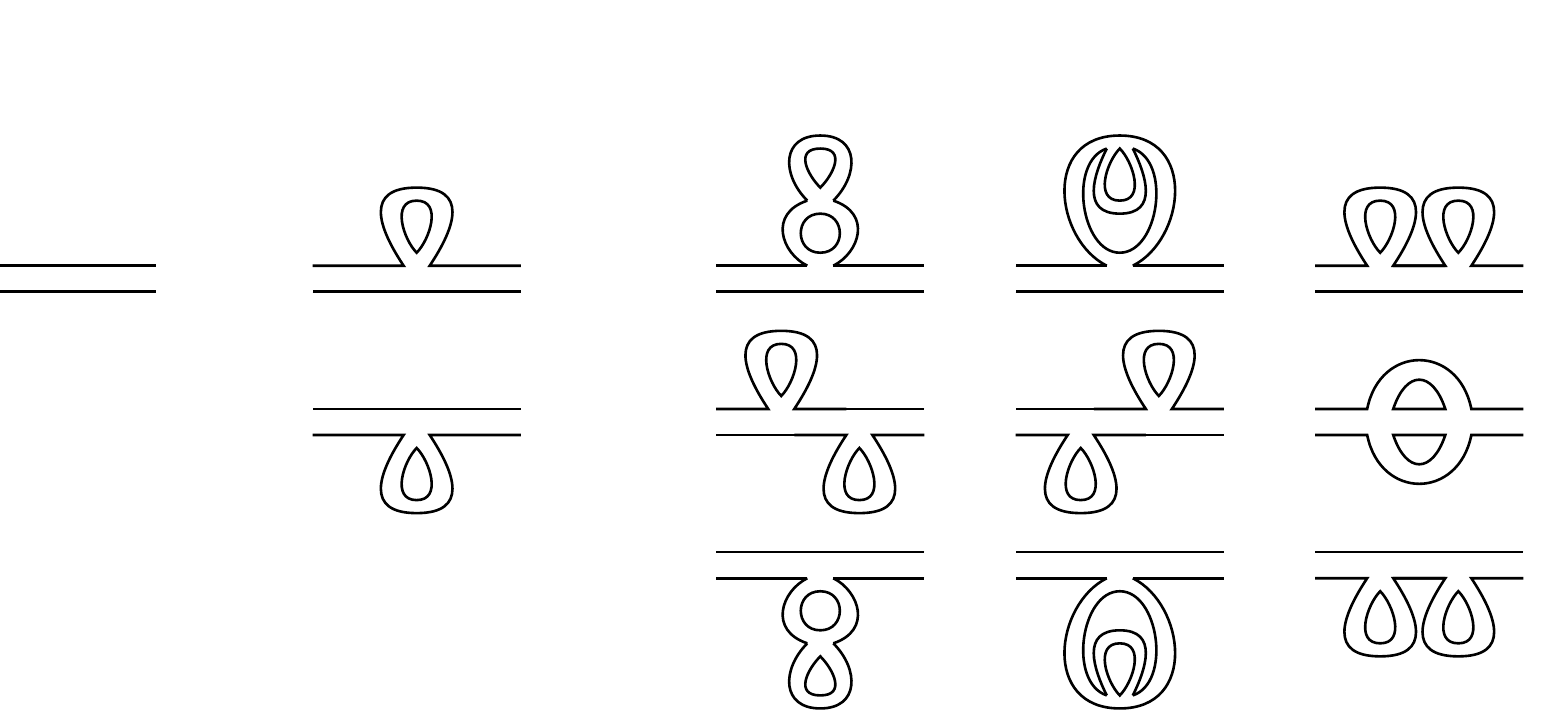}}%
    \put(0.02308892,0.39271376){\makebox(0,0)[lt]{\lineheight{1.25}\smash{\begin{tabular}[t]{l}$\lambda^0$\end{tabular}}}}%
    \put(0.25979371,0.39171405){\makebox(0,0)[lt]{\lineheight{1.25}\smash{\begin{tabular}[t]{l}$\lambda^1$\end{tabular}}}}%
    \put(0.69792114,0.39335028){\makebox(0,0)[lt]{\lineheight{1.25}\smash{\begin{tabular}[t]{l}$\lambda^2$\end{tabular}}}}%
  \end{picture}%
\endgroup%

	\caption{All graphs contributing to the 2-point function $G^{(0)}_{|ab|}$ up to order $\lambda^2$, where the upper strand is labelled by $a$ and the lower by $b$ for each graph. Topologically,
 some graphs are the same but different elements of $\mathfrak{G}^{0,v}_{|ab|}$
 due to different labellings.}
	\label{fig:2PGraphs}
\end{figure}
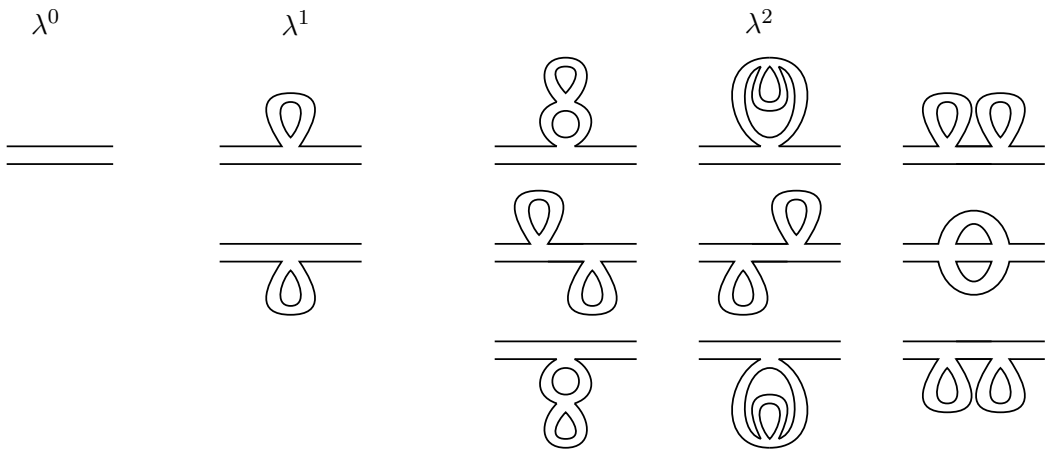
\end{Example}

\begin{Example}\label{Ex4}
 The second example will be the 4-point function to which 11 graphs contribute up to order
 $\lambda^2$ (see Figure \ref{fig:4PGraphs}).
Taking the weights into account,
 we have perturbatively
\begin{gather*}
G^{(0)}_{|abcd|}=\frac{(-\lambda)}{(E_a{+}E_b)(E_b{+}E_c)(E_c{+}E_d)(E_d{+}E_a)}
+\frac{(-\lambda)^2}{(E_a{+}E_b)(E_b{+}E_c)(E_c{+}E_d)(E_d{+}E_a)}
\\ \hphantom{G^{(0)}_{|abcd|}=}
{}\times\frac{1}{N} \sum_{k=1}^N\Big(\frac{1}{(E_a{+}E_{k})(E_a{+}E_b)}
+\frac{1}{(E_a{+}E_{k})(E_a{+}E_d)}+\frac{1}{(E_b{+}E_{k})(E_b{+}E_c)}
\\ \hphantom{G^{(0)}_{|abcd|}=\times}
{}+\frac{1}{(E_b{+}E_{k})(E_b{+}E_a)}+\frac{1}{(E_c{+}E_{k})(E_c{+}E_d)}+\frac{1}{(E_c{+}E_{k})(E_c{+}E_b)}
\\ \hphantom{G^{(0)}_{|abcd|}=\times}
{}+\frac{1}{(E_d{+}E_{k})(E_d{+}E_a)}+\frac{1}{(E_d{+}E_{k})(E_d{+}E_c)}+\frac{1}{(E_b{+}E_{k})(E_d{+}E_k)}
\\ \hphantom{G^{(0)}_{|abcd|}=\times}
{}+\frac{1}{(E_a{+}E_{k})(E_c{+}E_k)}\Big)+\mathcal{O}\big(\lambda^3\big).
\end{gather*}
 \begin{figure}[h!]
	\centering
	\def\svgwidth{300pt}
\begingroup%
  \makeatletter%
  \providecommand\color[2][]{%
    \errmessage{(Inkscape) Color is used for the text in Inkscape, but the package 'color.sty' is not loaded}%
    \renewcommand\color[2][]{}%
  }%
  \providecommand\transparent[1]{%
    \errmessage{(Inkscape) Transparency is used (non-zero) for the text in Inkscape, but the package 'transparent.sty' is not loaded}%
    \renewcommand\transparent[1]{}%
  }%
  \providecommand\rotatebox[2]{#2}%
  \newcommand*\fsize{\dimexpr\f@size pt\relax}%
  \newcommand*\lineheight[1]{\fontsize{\fsize}{#1\fsize}\selectfont}%
  \ifx\svgwidth\undefined%
    \setlength{\unitlength}{437.51516517bp}%
    \ifx\svgscale\undefined%
      \relax%
    \else%
      \setlength{\unitlength}{\unitlength * \real{\svgscale}}%
    \fi%
  \else%
    \setlength{\unitlength}{\svgwidth}%
  \fi%
  \global\let\svgwidth\undefined%
  \global\let\svgscale\undefined%
  \makeatother%
  \begin{picture}(1,0.25866307)%
    \lineheight{1}%
    \setlength\tabcolsep{0pt}%
    \put(0,0){\includegraphics[width=\unitlength,page=1]{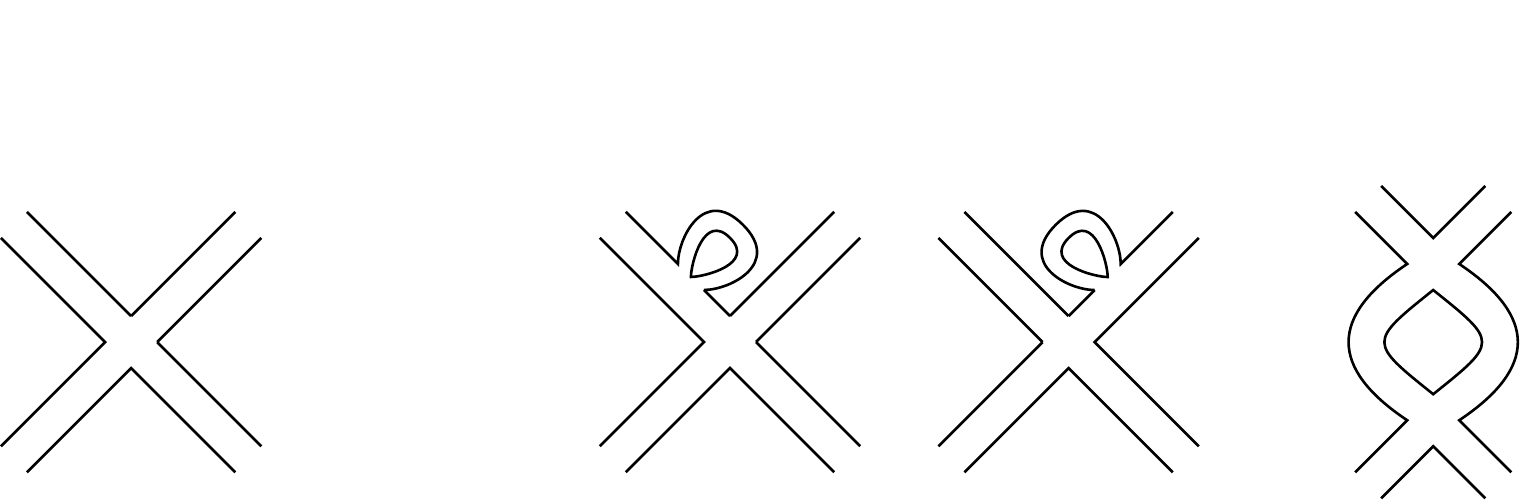}}%
    \put(0.06917512,0.24202455){\makebox(0,0)[lt]{\lineheight{1.25}\smash{\begin{tabular}[t]{l}$\lambda^1$\end{tabular}}}}%
    \put(0.66178395,0.24393211){\makebox(0,0)[lt]{\lineheight{1.25}\smash{\begin{tabular}[t]{l}$\lambda^2$\end{tabular}}}}%
  \end{picture}%
\endgroup%
 \vspace*{-2ex}
\caption{All graphs contributing to the 4-point function $G^{(0)}_{|abcd|}$ up to
 order $\lambda^2$, where the first two graphs of order $\lambda^2$ contribute with
 4 different labellings and the last graph with 2 different labellings. In~total this gives 10 different
 labelled ribbon graphs which
 contribute at order $\lambda^2$. These 11 labelled graphs are elements of
 $\mathfrak{G}^{0,v}_{|abcd|}$.}
	\label{fig:4PGraphs}
 \end{figure}
\end{Example}
\begin{Example}\label{Ex2+2}
The third example will be the $(2+2)$-point function to which 2
unlabelled ribbon graphs contribute up to order $\lambda^2$ (see Figure
\ref{fig:2+2PGraphs}).
Taking the weights into account, these split into 6 labelled
ribbon graphs, leading to a perturbative expansion
\begin{align*}
G^{(0)}_{|ab|cd|}={}&\frac{(-\lambda)^2}{(E_a+E_b)^2(E_c+E_d)^2}
\Big(\frac{1}{(E_a+E_c)^2}+\frac{1}{(E_a+E_d)^2}+\frac{1}{(E_b+E_c)^2}+\frac{1}{(E_b+E_d)^2}
\\
&+\frac{1}{(E_a+E_c)(E_b+E_d)}+\frac{1}{(E_a+E_d)(E_b+E_c)}\Big)+\mathcal{O}\big(\lambda^3\big).
\end{align*}
\begin{figure}[h!]
	\centering
	\def\svgwidth{200pt}
\begingroup%
  \makeatletter%
  \providecommand\color[2][]{%
    \errmessage{(Inkscape) Color is used for the text in Inkscape, but the package 'color.sty' is not loaded}%
    \renewcommand\color[2][]{}%
  }%
  \providecommand\transparent[1]{%
    \errmessage{(Inkscape) Transparency is used (non-zero) for the text in Inkscape, but the package 'transparent.sty' is not loaded}%
    \renewcommand\transparent[1]{}%
  }%
  \providecommand\rotatebox[2]{#2}%
  \newcommand*\fsize{\dimexpr\f@size pt\relax}%
  \newcommand*\lineheight[1]{\fontsize{\fsize}{#1\fsize}\selectfont}%
  \ifx\svgwidth\undefined%
    \setlength{\unitlength}{253.57041302bp}%
    \ifx\svgscale\undefined%
      \relax%
    \else%
      \setlength{\unitlength}{\unitlength * \real{\svgscale}}%
    \fi%
  \else%
    \setlength{\unitlength}{\svgwidth}%
  \fi%
  \global\let\svgwidth\undefined%
  \global\let\svgscale\undefined%
  \makeatother%
  \begin{picture}(1,0.32144105)%
    \lineheight{1}%
    \setlength\tabcolsep{0pt}%
    \put(0.46756463,0.2815156){\makebox(0,0)[lt]{\lineheight{1.25}\smash{\begin{tabular}[t]{l}$\lambda^2$\end{tabular}}}}%
    \put(0,0){\includegraphics[width=\unitlength,page=1]{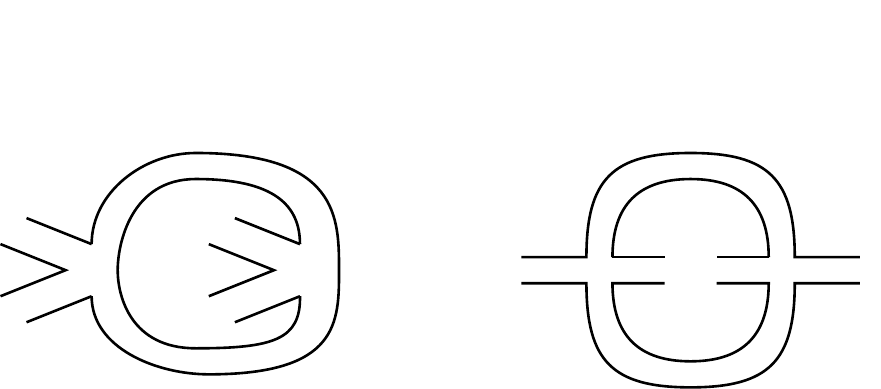}}%
  \end{picture}%
\endgroup%

 \caption{All graphs contributing to the $(2{+}2)$-point function $G^{(0)}_{|ab|cd|}$ up to
 order $\lambda^2$, where the left graph contributes with 4 different versions
 for the labelling and the right graph with 2 different versions for the labelling
 of the strands. This means that $\mathfrak{G}^{0,2}_{|ab|cd|}$ consists of 6 elements.
 } 
	\label{fig:2+2PGraphs}
 \end{figure}
\end{Example}

It is clear from the Feynman rules of Proposition~\ref{prop:weight}
that at each order the correlation functions are rational functions of $E_{p_i^j}$ and $E_{k_i}$.
Consequently the limit to coinciding indices $p_i^j\to p_{i'}^{j'}$ is well-defined
at any order in $\lambda$.

For any $n_i>2$, recursive algebraic relations between correlation functions are known;
we~refer to \cite{Hock:2020rje} for the general formula. In case of $(g,b)=(0,1)$, the
algebraic relation for the $n_1$-point function of genus zero is\vspace{-1ex}
\begin{gather*}
 G^{(0)}_{|p_1p_2\cdots p_{n_1}|}=-\lambda \sum_{k=1}^{\frac{n_1-2}{2}}
 \frac{G^{(0)}_{|p_{2k+2}\cdots p_{n_1}p_1|} G^{(0)}_{|p_2\cdots p_{2k+1}|}
 -G^{(0)}_{|p_{2k+1}\cdots p_{n_1}|} G^{(0)}_{|p_1\cdots p_{2k}|}}{(E_{p_{2k+1}}-E_{p_1})(E_{p_2}-E_{p_{n_1}})}.
\end{gather*}
The explicit combinatorial structure of this recursive equation was understood
in~\cite{DeJong} in form of two nested combinatorial problems each governed by Catalan numbers.
\begin{Example}
The 4-point function is algebraically expressed in terms of the 2-point func\-tion~by
\begin{gather*}
G^{(0)}_{|abcd|}=-\lambda \frac{G^{(0)}_{|ad|}G^{(0)}_{|bc|}-G^{(0)}_{|ab|}G^{(0)}_{|cd|}}{(E_c-E_a)(E_b-E_d)}.
\end{gather*}
The reader can check that this equation holds at the first two orders by
inserting Example~\ref{Ex2} into the r.h.s.\ to recover Example~\ref{Ex4}.
\end{Example}

\section{Auxiliary functions of topological significance}\label{sec:aux}

\subsection{Creation operator}

The derivative\vspace{-1ex}
\begin{gather*}
\hat{T}_q:=-N\frac{\partial}{\partial E_q}
\end{gather*}
with respect to the parameters of the free theory, which we call
\emph{boundary creation operator}, plays a central r\^ole in~\cite{Branahl:2020yru}. It is used to define the auxiliary functions\vspace{-1ex}
\begin{gather}
T_{q_1,q_2,\dots,q_m\|p_1^1\cdots p_{n_1}^1|p_1^2\cdots p_{n_2}^2|\cdots |p_1^b\cdots p_{n_b}^b|}
:=\hat{T}_{q_1}\cdots\hat{T}_{q_m}G_{|p_1^1\cdots p_{n_1}^1|p_1^2\cdots p_{n_2}^2|\cdots |p_1^b\cdots p_{n_b}^b|}
\nonumber
\\
\Omega_{q_1,\dots,q_m} :=\hat{T}_{q_2}\cdots\hat{T}_{q_m}\Omega_{q_1}
+\frac{\delta_{m,2}}{(E_{q_1}-E_{q_2})^2},\qquad
m\geq 2, \label{def:Om}
\end{gather}
where\vspace{-1ex}
\begin{gather*}
\Omega_{q} := \frac{1}{N}\sum_{k=1}^N G_{|qk|}+\frac{1}{N^2}G_{|q|q|}.
\end{gather*}
To define these functions properly it is necessary that all
$q_i$, $p_i^j$ are pairwise different.
As before we introduce genus expansions
$T_{q_1,q_2,\dots,q_m\|p_1^1\cdots p_{n_1}^1|\cdots |p_1^b\cdots p_{n_b}^b|}
=\sum_{g=0}^\infty N^{-2g}T^{(g)}_{q_1,q_2,\dots,q_m\|p_1^1\cdots p_{n_1}^1|\cdots |p_1^b\cdots p_{n_b}^b|}$
and
$\Omega_{q_1,\dots,q_m} =\sum_{g=0}^\infty N^{-2g}\Omega^{(g)}_{q_1,\dots,q_m}$.

\begin{Definition}\label{def:freeenergy}
 The free energy $\mathcal{F}$ is defined to be a primitive of $\Omega_q$ under the
 creation operator, i.e.,
$\Omega_q^{(g)}=:\hat{T}_q \mathcal{F}^{(g)}$.
\end{Definition}
Main tool to evaluate applications of $\hat{T}_q$ is the
equation of motion \cite[Lemma 2]{Schurmann:2019mzu-v3} which can be reformulated as
\begin{gather}
\frac{1}{N} \frac{\partial \log \mathcal{Z}(M)}{\partial E_{q}}
= \sum_{k=1}^N\bigg(
\frac{\partial^2 \log \mathcal{Z}(M)}{\partial M_{qk} \partial M_{kq}}\nonumber
+
\frac{\partial \log \mathcal{Z}(M)}{\partial M_{qk}}
\frac{\partial \log \mathcal{Z}(M)}{\partial M_{kq}}\bigg)
\\ \hphantom{\frac{1}{N} \frac{\partial \log \mathcal{Z}(M)}{\partial E_{q}}=}
{}+ \frac{1}{N}\sum_{k=1}^N G_{|qk|}+\frac{1}{N^2}G_{|q|q|}.
\label{eom-2}
\end{gather}
The following proposition gives
the exact result for a single $\hat{T}_q$-operation. In its proof the
assumption that the $p^j_i$ are pairwise different and different from $q$ is essential.
For application of another $\hat{T}_{q'}$ on the result such an assumption does not hold.
The calculation of several $\hat{T}$-operations must be carefully repeated.
\begin{Proposition}\label{PropT}
 For $\mathcal{J}=p_1^1\cdots p_{n_1}^1\big|p_1^2\cdots p_{n_2}^2\big|\cdots \big|p_1^b\cdots p_{n_b}^b\equiv
 \big\{J^1,\dots,J^b\big\}$ and $J^i=[p_1^i,\allowbreak\dots,p_{n_i}^i]$ with all $p^j_l$ pairwise different and
different from $q$ one has
\begin{gather*}
 \hat{T}_qG^{(g)}_{|\mathcal{J}|} =
 \frac{1}{N}\sum_{k=1}^N G^{(g)}_{|\mathcal{J}|qk|} +G^{(g-1)}_{|\mathcal{J}|q|q|}
+\sum_{j=1}^b\sum_{l=1}^{n_j} G^{(g)}_{|[q,p_l^j]\triangleright_l J^j| \mathcal{J}{\setminus}J^j|}
 +\sum_{\substack{g_1+g_2=g\\ \mathcal{J}_1\uplus \mathcal{J}_2=\mathcal{J}}}
G^{(g_1)}_{|\mathcal{ J}_1|q|}G^{(g_2)}_{|\mathcal{ J}_2|q|},
\end{gather*}
where
 $[p_1,p_2,\dots,p_i]\triangleright_l [q_1,q_2,\dots,q_j]:=[q_1,\dots,q_l,p_1,\dots,p_i,q_{l+1},\dots,q_j]$
 denotes the insertion of the first tuple
after the $l^\text{th}$ position of the second
tuple, $l=0,\dots,j$.
\begin{proof}
As in~\cite{Branahl:2020yru} we introduce derivative operators
$\frac{D^{|\mathcal{J}|}}{D M^{\mathcal{J}}}=
\frac{D^{|J^1|}}{DM^{J^1}}\cdots \frac{D^{|J^b|}}{D^{J^b}}$
with $\frac{D^{n}}{DM^{[p_1,\dots,p_n]}}:= \frac{(-\mathrm{i}N)^{n} \partial^n}{
 \partial M_{p_1p_2}\cdots \partial M_{p_{n-1}p_n}\partial M_{p_{n}p_1}}$. This gives a representation
$G_{|\mathcal{J}|}
=N^{b-2} \frac{D^{|\mathcal{J}|}}{D M^{\mathcal{J}}} \log \mathcal{Z}(M)\big|_{M=0}
$
to which we apply $\hat{T}_q$ via (\ref{eom-2}):
\begin{gather*}
\hat{T}_qG_{|\mathcal{J}|}
= N^{b-2} \sum_{k=1}^N \frac{D^{|\mathcal{J}|}}{D M^{\mathcal{J}}} \bigg(
{-}N^2 \frac{\partial^2 \log \mathcal{Z}(M)}{\partial M_{qk} \partial M_{kq}}
 -N^2 \frac{\partial \log \mathcal{Z}(M)}{\partial M_{qk}}
\frac{\partial \log \mathcal{Z}(M)}{\partial M_{kq}}
\bigg)\bigg|_{M=0}
\\ \hphantom{\hat{T}_qG_{|\mathcal{J}|}}
=\frac{1}{N}\sum_{k=1}^N G_{|\mathcal{J}|qk|} +\frac{1}{N^2} G_{|\mathcal{J}|q|q|}
+\sum_{j=1}^b\sum_{l=1}^{n_j} G_{|[q,p_l^j]\triangleright_l J^j| \mathcal{J}{\setminus}J^j|}
 +\sum_{\mathcal{J}_1\uplus \mathcal{J}_2=\mathcal{J}}
 G_{|\mathcal{ J}_1|q|}G_{|\mathcal{ J}_2|q|}.
 \end{gather*}
The second line results from the first line by the following considerations.
The first term
$N^{b-2}\frac{D^{|\mathcal{J}|}}{DM^{\mathcal{J}}}
\frac{D^2 (\log \mathcal{Z}(M))}{DM^{[q,k]}}$ contributes in three
different ways:
\begin{enumerate}\itemsep -3pt
\item[$(a)$] For generic $k$ it produces $\frac{1}{N}G_{|\mathcal{J}|qk|}$.
\item[$(b)$] For $k=q$ it produces, besides $\frac{1}{N}G_{|\mathcal{J}|qq|}$ included
 in (a), also $\frac{1}{N^2}G_{|\mathcal{J}|q|q|}$ when interpreting
$\frac{D^2}{DM^{[q,q]}}=\frac{D}{DM^{[q]}}\frac{D}{DM^{[q]}}$.

\item[$(c)$] For $k=p^j_l$ it produces,
besides $\frac{1}{N}G_{|\mathcal{J}|qp^j_l|}$ included in (a), also
$G_{|[q,p_l^j]\triangleright_l J^j|
\mathcal{J}{\setminus} J^j|}$ when taking
$\frac{D^{n^j}}{DM^{[p^j_1,\dots,p^j_{n_j}]}}
\frac{D^2}{DM^{[q,p^j_l]}}=\frac{D^{n_j+2}}{D
M^{[p^j_1,\dots,p^j_l,q,p^j_l,p^j_{l+1},p^j_{n_j}]}}$ into account.
\end{enumerate}
The second term of the first line only contributes for $k=q$ and for partitions of
$\frac{D^{|\mathcal{J}|}}{DM^{\mathcal{J}}}$
into two blocks $\mathcal{J}=\mathcal{J}'\uplus\mathcal{J}''$ which preserve the $J^j$
individually. Indeed, any splitting of the
$\frac{D^{n}}{DM^{[p_1,\dots,p_n]}}$ applied
to $\mathcal {Z}(M)$ gives zero when setting $M=0$.

Inserting $G_{\dots}=\sum_{g=0}^\infty N^{-2g}G_{\dots}^{(g)}$ in the second line and extracting the coefficient of
$N^{-2g}$ gives the assertion.
\end{proof}
\end{Proposition}

\begin{Example}\label{ex2pg}
	The action of the creation operator on the 2-point function reads
	\begin{gather*}
 \hat{T}_qG^{(g)}_{|p_1p_2|}=\frac{1}{N}\sum_{k=1}^N G^{(g)}_{|p_1p_2|qk|}
 +G^{(g-1)}_{|p_1p_2|q|q|}+G^{(g)}_{|p_1qp_1p_2|}+G^{(g)}_{|p_1p_2qp_2|}.
	\end{gather*}
\end{Example}

\begin{Example}
	The action of the creation operator on the $(1+1)$-point function reads
	\begin{gather*}
 \hat{T}_qG^{(g)}_{|p^1|p^2|}=\frac{1}{N}\sum_{k=1}^N G^{(g)}_{|p^1|p^2|qk|}
 +G^{(g-1)}_{|p^1|p^2|q|q|}+G^{(g)}_{|p^1qp^1|p^2|}+G^{(g)}_{|p^1|p^2qp^2|}
 +\sum_{g_1+g_2=g}G^{(g_1)}_{|p^1|q|}G^{(g_2)}_{|p^2|q|}.
	\end{gather*}
\end{Example}

We also give a perturbative proof of Proposition~\ref{PropT}.
The creation operator $\hat{T}_q$
takes the derivative with respect to $E_q$ of a rational function arising from the
Feynman rules in Proposition~\ref{prop:weight}. Since all external indices $p_i^j$ are by assumption
different from $q$, the derivative hits only the sums of the internal
strands (loops) if the summation index coincides with $q$. Being a~derivative,
it is the sum over all strands of all internal loops. Isolating every such target as
\begin{gather*}
\frac{1}{N}\sum_{k=1}^N \frac{1}{E_k+E_m}f(E_k,E_m,\dots ),
\end{gather*}
where $f(E_k,E_m,\dots )$ is a rational function in $E_k$, $E_m$ and further $E_{j}$, then the creation operator
generates
\begin{gather*}
 \hat{T}_q \frac{1}{N}\sum_{k=1}^N \frac{1}{E_k+E_m} f(E_k,E_m,\dots )
 \rightarrow \frac{1}{(E_q+E_m)^2}f(E_q,E_m,\dots ).
\end{gather*}
Graphically, a ribbon with internal strand labelled by $k$ is hit by
the creation operator $\hat{T}_q$. Its~ribbon is cut into two ribbons
each with weight $\frac{1}{E_q+E_m}$, where the previous loop label $k$ is now fixed to $q$.
Depending on the type of the other index $m$ and the
topology of the graph, four classes of ribbon graphs can be produced by
action of $\hat{T}_q$:
\begin{enumerate}
\item \label{itemA} The creation operator $\hat{T}_q$ acts
 on a ribbon
 in which both strands are internal strands, but different from each
 other. This means that $m$ above is another summation
 index to which a summation operator $\frac{1}{N}\sum_{m=1}^N$ is
 assigned. The ribbon graph resulting from application of
 $\hat{T}_q$ thus receives an additional boundary component with $2$
 one-valent vertices with one strand fixed to $q$ and the other to
 $m$ with a summation over $m$. All other boundary components of the
 previous ribbon graphs are untouched. The resulting graph
 contributes to
 $\frac{1}{N}\sum_{m=1}^N
 G^{(g)}_{|p^1_1\cdots p^1_{n_1}|\cdots |p^b_1\cdots p^b_{n_b}|qm|}$.
\begin{figure}[h!]
	\centering
	\def\svgwidth{300pt}
\begingroup%
  \makeatletter%
  \providecommand\color[2][]{%
    \errmessage{(Inkscape) Color is used for the text in Inkscape, but the package 'color.sty' is not loaded}%
    \renewcommand\color[2][]{}%
  }%
  \providecommand\transparent[1]{%
    \errmessage{(Inkscape) Transparency is used (non-zero) for the text in Inkscape, but the package 'transparent.sty' is not loaded}%
    \renewcommand\transparent[1]{}%
  }%
  \providecommand\rotatebox[2]{#2}%
  \newcommand*\fsize{\dimexpr\f@size pt\relax}%
  \newcommand*\lineheight[1]{\fontsize{\fsize}{#1\fsize}\selectfont}%
  \ifx\svgwidth\undefined%
    \setlength{\unitlength}{293.25000754bp}%
    \ifx\svgscale\undefined%
      \relax%
    \else%
      \setlength{\unitlength}{\unitlength * \real{\svgscale}}%
    \fi%
  \else%
    \setlength{\unitlength}{\svgwidth}%
  \fi%
  \global\let\svgwidth\undefined%
  \global\let\svgscale\undefined%
  \makeatother%
  \begin{picture}(1,0.4117647)%
    \lineheight{1}%
    \setlength\tabcolsep{0pt}%
    \put(0,0){\includegraphics[width=\unitlength,page=1]{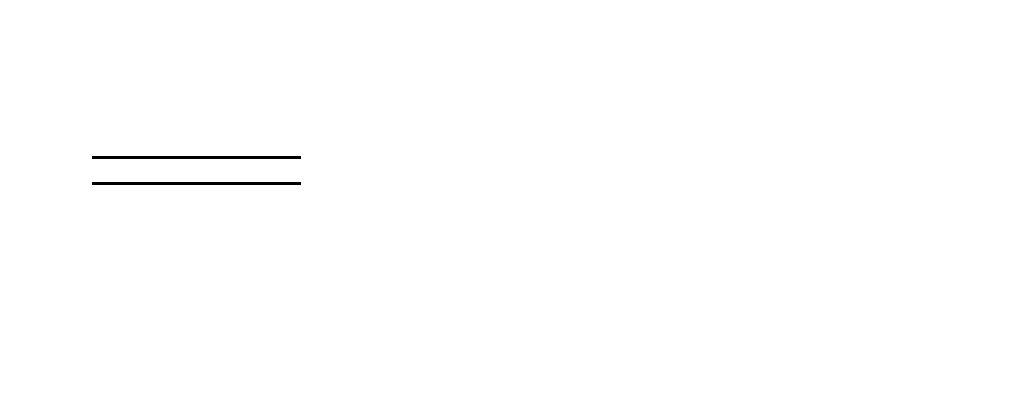}}%
    \put(0.15513402,0.27608659){\makebox(0,0)[lt]{\lineheight{1.25}\smash{\begin{tabular}[t]{l}$m$\end{tabular}}}}%
    \put(0.15103749,0.18269442){\makebox(0,0)[lt]{\lineheight{1.25}\smash{\begin{tabular}[t]{l}$k$\end{tabular}}}}%
    \put(0.39811776,0.23588309){\makebox(0,0)[lt]{\lineheight{1.25}\smash{\begin{tabular}[t]{l}$\rightarrow~\frac{1}{N}\sum_{m=1}^N $\end{tabular}}}}%
    \put(0,0){\includegraphics[width=\unitlength,page=2]{insertion2.pdf}}%
    \put(0.72475953,0.26768853){\makebox(0,0)[lt]{\lineheight{1.25}\smash{\begin{tabular}[t]{l}$m$\end{tabular}}}}%
    \put(0.72742552,0.17658642){\makebox(0,0)[lt]{\lineheight{1.25}\smash{\begin{tabular}[t]{l}$q$\end{tabular}}}}%
    \put(0,0){\includegraphics[width=\unitlength,page=3]{insertion2.pdf}}%
    \put(0.83919991,0.26845462){\makebox(0,0)[lt]{\lineheight{1.25}\smash{\begin{tabular}[t]{l}$m$\end{tabular}}}}%
    \put(0.83859576,0.17734407){\makebox(0,0)[lt]{\lineheight{1.25}\smash{\begin{tabular}[t]{l}$q$\end{tabular}}}}%
  \end{picture}%
\endgroup%

\end{figure}

\item \label{itemB} The creation operator $\hat{T}_q$ acts on a
 ribbon, where both strands are internal strands and the same, i.e.,
 the index $m$ is also set to $k$. Here we consider the case that after cutting the ribbon, the
 ribbon graph is still connected. Cutting the selected ribbon then decreases the genus by 1;
 otherwise it is not possible that both strands have the same label.
 Acting with~$\hat{T}_q$ on the other strand of the same ribbon leads to the same result,
 thus a total factor of $2$.
The resulting graph has two additional boundary components
 each with 1 one-valent vertex with strands fixed to $q$. All other
 boundary components of the previous ribbon graphs are untouched. The
 resulting graph with its factor of $2$ contributes to
 $G^{(g-1)}_{|p^1_1\cdots p^1_{n_1}|\cdots |p^b_1\cdots p^b_{n_b}|q|q|}$. The factor of 2 accounts for the
 difference between labelled and unlabelled ribbon graphs. To see this consider
 $G^{(g-1)}_{|p^1_1\cdots p^1_{n_1}|\cdots |p^b_1\cdots p^b_{n_b}|q'|q''|}$ with $q'\neq q''$ in which every topological
 ribbon graph occurs twice, namely first with labels $q'$, $q''$ on an ordered pair of open lines and
 second with labels $q''$, $q'$ on that pair. When setting $q'=q''=q$ we get twice the same labelled
 ribbon graph.
\begin{figure}[h!]
\centering
\def\svgwidth{350pt}
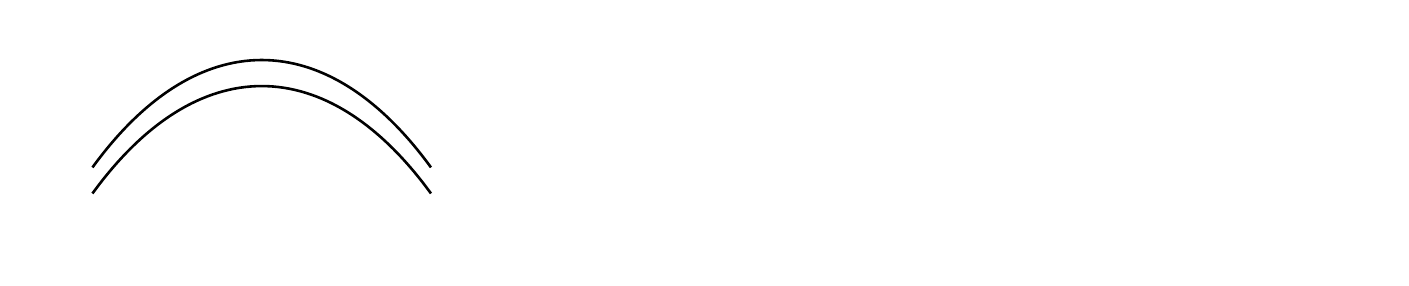
\end{figure}

\item \label{itemC} The creation operator $\hat{T}_q$ acts on a
 ribbon, where one strand is internal and the other external,
 i.e., the index $m$ above is some $p_l^j$. After cutting the
 ribbon, the previously inter\-nal strand becomes part of the
 $j^{\mathrm{th}}$ boundary component. The resulting ribbon graph recei\-ves~2 additional
 one-valent vertices next to each other, with attached ribbons labelled $p^j_lq$ and~$qp^j_l$,
 at the $j^{\mathrm{th}}$ boundary component. All other $b-1$ boundary components
 are untou\-ched. The resulting ribbon graph thus contributes to
 $G^{(g)}_{|p^1_1\cdots p^1_{n_1}|p^j_1\cdots p^j_{l-1}p^j_{l}qp^j_l p^j_{l+1}\cdots p^j_{n_j}|p^b_1\cdots p^b_{n_b}|}$.\vspace{-2ex}
\begin{figure}[h!]
	\centering
	\def\svgwidth{350pt}
	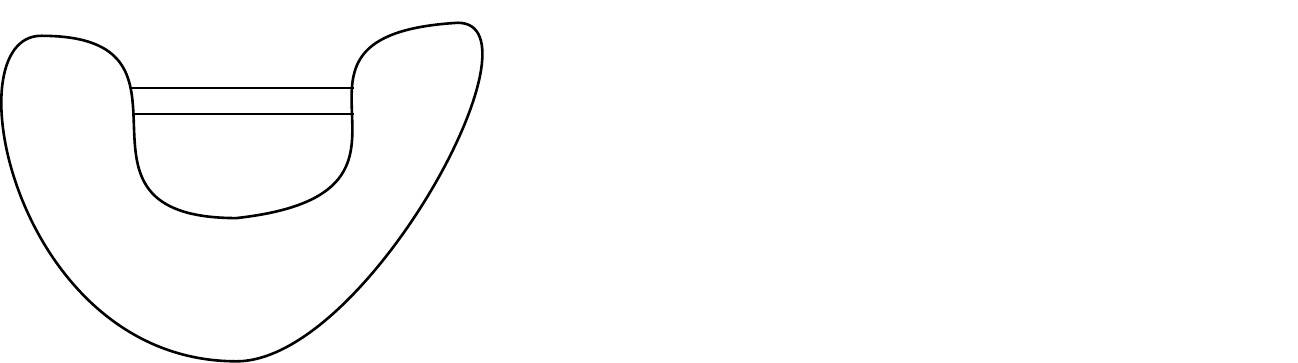\vspace{-1ex}
\end{figure}
\item\label{itemD} The creation operator $\hat{T}_q$ acts on a ribbon,
 where both strands are internal strands of the same label $k$, but
 in contrast to Case~\ref{itemB} the ribbon graph disconnects after
 cutting the ribbon. Acting with $\hat{T}_q$ on the other strand also
 labelled $k$ gives the same splitting, thus a~total factor of $2$.
 Each of the two resulting connected ribbon graphs receives an~additional boundary component with a single one-valent vertex whose
 ribbon is labelled~$qq$. All~pre\-vious boundaries with labels
 $\mathcal{J}$ as well as the total genus $g$ are untouched, but
 split into each of the ribbon graphs. This splitting accounts for
 the additional factor of $2$ because for given assignment
 $\mathcal{ J}'$, $\mathcal{ J}''$ the decompositions
 $\mathcal{ J}' \uplus \mathcal{ J}''=\mathcal{J}$ and
 $\mathcal{ J}'' \uplus \mathcal{ J}'=\mathcal{J}$ are considered as
 different. The resulting ribbon graph thus contributes to
 $G^{(g_1)}_{|\mathcal{ J}_1|q|}G^{(g_2)}_{|\mathcal{ J}_2|q|}$
 with sum over $g_1+g_2=g$ and over splittings $\mathcal{J}_1\uplus\mathcal{J}_2=\mathcal{J}$.
 \begin{figure}[h!]
	\centering
	\def\svgwidth{350pt}
\begingroup%
  \makeatletter%
  \providecommand\color[2][]{%
    \errmessage{(Inkscape) Color is used for the text in Inkscape, but the package 'color.sty' is not loaded}%
    \renewcommand\color[2][]{}%
  }%
  \providecommand\transparent[1]{%
    \errmessage{(Inkscape) Transparency is used (non-zero) for the text in Inkscape, but the package 'transparent.sty' is not loaded}%
    \renewcommand\transparent[1]{}%
  }%
  \providecommand\rotatebox[2]{#2}%
  \newcommand*\fsize{\dimexpr\f@size pt\relax}%
  \newcommand*\lineheight[1]{\fontsize{\fsize}{#1\fsize}\selectfont}%
  \ifx\svgwidth\undefined%
    \setlength{\unitlength}{401.22304637bp}%
    \ifx\svgscale\undefined%
      \relax%
    \else%
      \setlength{\unitlength}{\unitlength * \real{\svgscale}}%
    \fi%
  \else%
    \setlength{\unitlength}{\svgwidth}%
  \fi%
  \global\let\svgwidth\undefined%
  \global\let\svgscale\undefined%
  \makeatother%
  \begin{picture}(1,0.11305817)%
    \lineheight{1}%
    \setlength\tabcolsep{0pt}%
    \put(0.14855232,0.07410796){\makebox(0,0)[lt]{\lineheight{1.25}\smash{\begin{tabular}[t]{l}$k$\end{tabular}}}}%
    \put(0.4902379,0.04305858){\makebox(0,0)[lt]{\lineheight{1.25}\smash{\begin{tabular}[t]{l}$\rightarrow\qquad 2\times$\end{tabular}}}}%
    \put(0,0){\includegraphics[width=\unitlength,page=1]{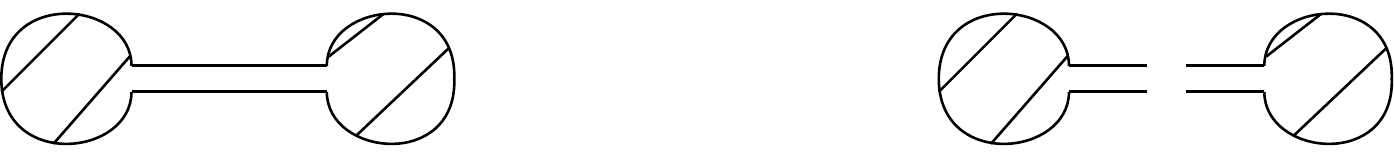}}%
    \put(0.15069729,0.00756968){\makebox(0,0)[lt]{\lineheight{1.25}\smash{\begin{tabular}[t]{l}$k$\end{tabular}}}}%
    \put(0.78095842,0.07501556){\makebox(0,0)[lt]{\lineheight{1.25}\smash{\begin{tabular}[t]{l}$q$\end{tabular}}}}%
    \put(0.77695432,0.00732624){\makebox(0,0)[lt]{\lineheight{1.25}\smash{\begin{tabular}[t]{l}$q$\end{tabular}}}}%
    \put(0.86857363,0.07518566){\makebox(0,0)[lt]{\lineheight{1.25}\smash{\begin{tabular}[t]{l}$q$\end{tabular}}}}%
    \put(0.86968487,0.00849679){\makebox(0,0)[lt]{\lineheight{1.25}\smash{\begin{tabular}[t]{l}$q$\end{tabular}}}}%
  \end{picture}%
\endgroup%
\vspace{-1ex}
\end{figure}
\end{enumerate}

Notice that for the vacuum ribbon graphs of the free energy $\mathcal{F}^{(g)}$ only
the two Cases~\ref{itemA} and~\ref{itemB} contribute.
Case~\ref{itemC} contributes only if a ribbon graph has an external strand,
which is not the case for a vacuum ribbon graph, and Case~\ref{itemD} does not
contribute because any vacuum ribbon graph is a 1PI (one particle irreducible)
due to four-valent vertices, i.e., after cutting a ribbon the graph
stays connected.

\begin{Example}
 Take Example~\ref{ex2pg} for $g=0$ with its ribbon graph expansion. The
 first orders of the expansion of the 2-point function are given in
 Example~\ref{Ex2} with ribbon graphs drawn in Figure~\ref{fig:2PGraphs}. The perturbative action of the creation operator
 described above generates the corresponding contributions of the 4-point function and
 the $(2+2)$-point function, which can be taken from the Examples
 \ref{Ex4} and \ref{Ex2+2}, where the graphs are drawn in Figures
 \ref{fig:4PGraphs} and \ref{fig:2+2PGraphs}. It is left to the
 reader to check the explicit formulae.
\end{Example}

\begin{Example}
The action of the creation operator on the free energy of genus $g=0$ is
\begin{align*}
\hat{T}_q \mathcal{F}^{(0)}=\frac{1}{N}\sum_{k=1}^N G^ {(0)}_{|qk|}.
\end{align*}
Take the perturbative expansion of the free energy from Example \ref{ExF0}
with ribbon graphs drawn in Figure \ref{fig:freeenergy}. The symmetry factor
of the automorphism group of each ribbon graph ensures that the ribbon graphs
generated by the creation operator have correct factors in agreement with
Example \ref{Ex2}. Consequently, this gives a way to compute the
order of the automorphism group. For instance, the contribution
of the sunrise graph to $\frac{1}{N}\sum_{k=1}^N G^ {(0)}_{|qk|}$ is
generated in 8 different ways by acting with the creation operator at
the melon graph $\Gamma_M$ of $\mathcal{F}^{(0)}$, which hence has $|\mathrm{Aut}(\Gamma_M)|=8$.
\end{Example}
We conclude that the action of the creation operator can be
represented in two different but equivalent ways, first directly on the
correlation functions by using manipulations of the partition function
and second perturbatively by using the action on the weighted graphs.

\subsection[Representation of Omega(g) in terms of correlation functions]
{Representation of $\boldsymbol{\Omega^{(g)}}$ in terms of correlation functions}

We have shown in~\cite{Branahl:2020yru} that the
$\Omega^{(g)}_{q_1,\dots,q_m}$ defined in (\ref{def:Om}) extend to
meromorphic differential forms $\omega_{g,m}$ on $\hat{\mathbb{C}}^m$
for which we provided strong evidence \cite{Borot:2015hna}
and a proof for $g=0$ \cite{Hock:2021tbl} that they obey blobbed
topological recursion. If true the $\omega_{g,m}$
are relatively easy to obtain via evaluation of residues. The
translation back to $\Omega^{(g)}_{q_1,\dots,q_m}$ is simple. In this
section we give another representation of the
$\Omega^{(g)}_{q_1,\dots,q_m}$ as special polynomials in the correlation
functions $G^{(g')}_{\dots}$. The purpose is twofold. First,
comparison of a perturbative expansion of the $G^{(g')}_{\dots}$ with
a~Taylor expansion of the exact formulae provides an important
consistency check. Second, we understand our observation as a message
for quantum field theory in general: Also in realistic QFT it might be
worthwhile to investigate whether certain polynomials of correlation
functions, which themselves are Feynman graph series, reveal a deeper
structure than individual functions or graphs.

The definition already gives $\Omega^{(g)}_{q_1}=\frac{1}{N}\sum_{k=1}^N G^{(g)}_{|q_1k|}
+G^{(g-1)}_{|q_1|q_1|}$. The next two propositions provide $\Omega^{(g)}_{q_1,q_2}$ and
$\Omega^{(g)}_{q_1,q_2,q_3}$. It would be straightforward but increasingly lengthy to continue.
\begin{Proposition}
 \label{prop:Om02G}
 For $q_1\neq q_2$ one has
 \begin{align*}
\Omega^{(g)}_{q_1,q_2}=&{}
\frac{\delta_{g,0}}{(E_{q_1}-E_{q_2})^2}
+\sum_{g_1+g_2=g} G^{(g_1)}_{|q_1q_2|} G^{(g_2)}_{|q_1q_2|}
+ \frac{1}{N^2}\sum_{k,l=1}^N G^{(g)}_{|q_1k|q_2l|}
\\
&+\frac{1}{N}\sum_{k=1}^N \Big(G^{(g)}_{|q_1kq_1q_2|}+G^{(g)}_{|q_2kq_2q_1|}+G^{(g)}_{|q_1kq_2k|}
\Big)+\frac{1}{N} \sum_{k=1}^N \Big(G^{(g-1)}_{|q_1k|q_2|q_2|}+G^{(g-1)}_{|q_2k|q_1|q_1|}\Big)
\\
&+G^{(g-1)}_{|q_1q_2q_2|q_2|}+G^{(g-1)}_{|q_2q_1q_1|q_1|}
+\sum_{g_1+g_2=g-1} G^{(g_1)}_{|q_1|q_2|} G^{(g_2)}_{|q_1|q_2|}
+G^{(g-2)}_{|q_1|q_1|q_2|q_2|}.
\end{align*}
\end{Proposition}

\begin{proof}
 Using $\Omega_{q_1}=\frac{1}{N^2}\sum_{k=1}^N \frac{D^2}{DM^{[q_1,k]}} \log \mathcal{Z}(M)\big|_{M=0}$
 and (\ref{eom-2}) we have
\begin{align}
\hat{T}_{q_2}\Omega_{q_1}
={}&-\sum_{k,l=1}^N \frac{D^2}{DM^{[q_1,k]}} \bigg(
\frac{\partial^2\log \mathcal{Z}(M)}{\partial M_{q_2l}\partial M_{lq_2}}
+\frac{\partial \log \mathcal{Z}(M)}{\partial M_{q_2l}}
\frac{\partial \log \mathcal{Z}(M)}{\partial M_{lq_2}}\bigg)\bigg|_{M=0}
\nonumber
\\
={}&\frac{1}{N^2}\sum_{k,l=1}^N G_{|q_1k|q_2l|}
+\frac{1}{N}\sum_{k=1}^N \big(G_{|q_1kq_1q_2|}+G_{|q_2kq_2q_1|}+G_{|q_1kq_2k|}
\big)+G_{|q_1q_2|} G_{|q_1q_2|}
\nonumber
\\
& +\frac{1}{N^3} \sum_{k=1}^N \big(G_{|q_1k|q_2|q_2|}+G_{|q_2k|q_1|q_1|}\big)
+\frac{1}{N^2} \big(G_{|q_1q_2q_2|q_2|}+G_{|q_2q_1q_1|q_1|}\big)
\nonumber
\\
&+\frac{1}{N^2}G_{|q_1|q_2|} G_{|q_1|q_2|}
+\frac{1}{N^4}G_{|q_1|q_1|q_2|q_2|}.
\label{TqOmq}
\end{align}
The last three lines result from the first line of (\ref{TqOmq}) as follows.
The first term
$\frac{1}{N^2} \frac{D^{2}}{DM^{[q_1,k]}}\times
\frac{D^2 (\log \mathcal{Z}(M))}{DM^{[q_2,l]}}$ contributes in nine
different ways:
\begin{enumerate}\itemsep=0pt
\item[$(a)$] For generic $k$, $l$ it produces $\frac{1}{N^2}G_{|q_1k|q_2l|}$ summed over $k$, $l$.
\item[$(b)$] For $l=q_1$ it produces, besides $\frac{1}{N^2}G_{|q_1k|q_2q_1|}$ included
 in $(a)$, also $\frac{1}{N}G_{|q_1kq_1q_2|}$ when interpreting
 $ \frac{D^{2}}{DM^{[q_1,k]}}\frac{D^2}{DM^{[q_2,q_1]}}
 =\frac{D^{4}}{DM^{[q_1,k,q_1,q_2]}}$.

\item[$(c)$] For $k=q_2$ it produces, besides $\frac{1}{N^2}G_{|q_1q_2|q_2l|}$ included
 in $(a)$, also $\frac{1}{N}G_{|q_2lq_2q_1|}$ when interpreting
 $ \frac{D^{2}}{DM^{[q_1,q_2]}}\frac{D^2}{DM^{[q_2,l]}}
 =\frac{D^{4}}{DM^{[q_2,l,q_2,q_1]}}$. We~replace $l\mapsto k$.

\item[$(d)$] For $l=k$ it produces, besides $\frac{1}{N^2}G_{|q_1k|q_2k|}$ included
 in $(a)$, also $\frac{1}{N}G_{|q_1kq_2k|}$ when interpreting
 $ \frac{D^{2}}{DM^{[q_1,k]}}\frac{D^2}{DM^{[q_2,k]}}
 =\frac{D^{4}}{DM^{[q_1,k,q_2,k]}}$.

 \item[$(e)$] For $l=q_2$ it produces, besides $\frac{1}{N^2}G_{|q_1k|q_2q_2|}$ included
 in $(a)$, also $\frac{1}{N^3}G_{|q_1k|q_2|q_2|}$ when interpreting
 $ \frac{D^{2}}{DM^{[q_1,k]}}\frac{D^2}{DM^{[q_2,q_2]}}
 =\frac{D^{2}}{DM^{[q_1,k]}}\frac{D^{1}}{DM^{[q_2]}}\frac{D^{1}}{DM^{[q_2]}}
 $.

 \item[$(f)$] For $k=q_1$ it produces, besides $\frac{1}{N^2}G_{|q_1q_1|q_2l|}$ included
 in $(a)$, also $\frac{1}{N^3}G_{|q_2l|q_1|q_1|}$ when interpreting
 $ \frac{D^2}{DM^{[q_1,q_1]}}\frac{D^{2}}{DM^{[q_2,l]}}
 =\frac{D^{2}}{DM^{[q_2,l]}}\frac{D^{1}}{DM^{[q_1]}}\frac{D^{1}}{DM^{[q_1]}}
 $. We~replace $l\mapsto k$.

\item[$(g)$] For $k\!=l\!=q_2$ it produces, besides the three cases
 $\frac{1}{N^2}G_{|q_1q_2|q_2q_2|}$ included
 in $(a)$, $\frac{1}{N}G_{|q_2q_2q_2q_1|}$ included in $(c)$ and
 $\frac{1}{N^3}G_{|q_1q_2|q_2|q_2|}$ included in $(e)$,
 also $\frac{1}{N^2}G_{|q_1q_2q_2|q_2|}$ when interpreting
 $ \frac{D^2}{DM^{[q_1,q_2]}}\frac{D^{2}}{DM^{[q_2,q_2]}}
 =\frac{D^{3}}{DM^{[q_1,q_2,q_2]}}\frac{D^{1}}{DM^{[q_2]}} $.

\item[$(h)$] For $k\!=l\!=q_1$ it produces, besides the three cases
 $\frac{1}{N^2}G_{|q_1q_1|q_2q_1|}$ included
 in $(a)$, $\frac{1}{N}G_{|q_1q_1q_1q_2|}$ included in $(b)$ and
 $\frac{1}{N^3}G_{|q_2q_1|q_1|q_1|}$ included in $(f)$,
 also $\frac{1}{N^2}G_{|q_2q_1q_1|q_1|}$ when interpreting
 $ \frac{D^2}{DM^{[q_1,q_1]}}\frac{D^{2}}{DM^{[q_2,q_1]}}
 =\frac{D^{3}}{DM^{[q_2,q_1,q_1]}}\frac{D^{1}}{DM^{[q_1]}} $.

\item[(i)] For $k=q_1$ and $l=q_2$ it produces, besides the three cases
 $\frac{1}{N^2}G_{|q_1q_1|q_2q_2|}$ included
 in $(a)$, $\frac{1}{N^3}G_{|q_1q_1|q_2|q_2|}$ included in $(e)$ and
 $\frac{1}{N^3}G_{|q_1|q_1|q_2q_2|}$ included in $(f)$,
 also $\frac{1}{N^4}G_{|q_1|q_1|q_2|q_2|}$ when interpreting
 $ \frac{D^2}{DM^{[q_1,q_1]}}\frac{D^{2}}{DM^{[q_2,q_2]}}
 =\frac{D^{1}}{DM^{[q_1]}}\frac{D^{1}}{DM^{[q_1]}}
 \frac{D^{1}}{DM^{[q_2]}}\frac{D^{1}}{DM^{[q_2]}}$.
\end{enumerate}
The second term
$-\frac{1}{N^2} \frac{D^{2}}{DM^{[q_1,k]}}
\frac{\partial \log \mathcal{Z}(M)}{\partial M_{q_2l}}
\frac{\partial \log \mathcal{Z}(M)}{\partial M_{lq_2}}$
contributes in two different ways:
\begin{enumerate}\itemsep=0pt
\item[$(a)$] Either $k=q_2$ and $l=q_1$ and derivatives distributed into
$\frac{1}{N^2}\frac{D^{2}\log \mathcal{Z}(M)}{DM^{[q_1,q_2]}}
\frac{D^{2}\log \mathcal{Z}(M)}{DM^{[q_1,q_2]}}=G_{|q_1q_2|}G_{|q_1q_2|}$,

\item[$(b)$] or $k=q_1$ and $l=q_2$ and derivatives distributed into
 $\frac{1}{N^2}\frac{D^{2}\log \mathcal{Z}(M)}{DM^{[q_1]}DM^{[q_2]}}
 \frac{D^{2}\log \mathcal{Z}(M)}{DM^{[q_1]}DM^{[q_2]}}
=\frac {1}{N^2}G_{|q_1|q_2|}G_{|q_1|q_2|}$.
\end{enumerate}
Including the special term $\frac{1}{(E_{q_1}-E_{q_2})^2}$ and
extracting the coefficient of $N^{-2g}$ gives the as\-ser\-tion.
\end{proof}

\begin{Proposition} \label{prop:Om03G}
For pairwise different $q_1$, $q_2$, $q_3$ one has
\begin{gather*}
\Omega_{q_1,q_2,q_3}^{(g)}
=\frac{1}{N^3}\sum_{j,k,l=1}^N G^{(g)}_{|q_1j|q_2k|q_3l|}
\\ \hphantom{\Omega_{q_1,q_2,q_3}^{(g)}=}
{}+\frac{1}{N^2}\sum_{k,l=1}^N
\Big(G^{(g)}_{|q_1kq_2k|q_3l|}+G^{(g)}_{|q_2kq_3k|q_1l|}
+G^{(g)}_{|q_3kq_1k|q_2l|}
+G^{(g)}_{|q_1kq_1q_2|q_3l|}
\\ \hphantom{\Omega_{q_1,q_2,q_3}^{(g)}=+}
{}+G^{(g)}_{|q_1kq_1q_3|q_2l|}
+G^{(g)}_{|q_2kq_2q_3|q_1l|}+G^{(g)}_{|q_2kq_2q_1|q_3l|}
+G^{(g)}_{|q_3kq_3q_1|q_2l|}+G^{(g)}_{|q_3kq_3q_2|q_1l|}\Big)
\\ \hphantom{\Omega_{q_1,q_2,q_3}^{(g)}=}
{} +\frac{1}{N}\sum_{k=1}^N
\Big(G^{(g)}_{|q_1kq_2kq_3k|}+G^{(g)}_{|q_1kq_3kq_2k|}
+G^{(g)}_{|q_1kq_1q_2q_1q_3|}+G^{(g)}_{|q_1kq_1q_3q_1q_2|}
\\ \hphantom{\Omega_{q_1,q_2,q_3}^{(g)}=+}
{}+G^{(g)}_{|q_2kq_2q_3q_2q_1|}
+G^{(g)}_{|q_2kq_2q_1q_2q_3|}
+G^{(g)}_{|q_3kq_3q_1q_3q_2|}+G^{(g)}_{|q_3kq_3q_2q_3q_1|}
\\ \hphantom{\Omega_{q_1,q_2,q_3}^{(g)}=+}
{}+G^{(g)}_{|q_1kq_1q_2q_3q_2|}
+G^{(g)}_{|q_1kq_1q_3q_2q_3|}
+G^{(g)}_{|q_2kq_2q_3q_1q_3|}+G^{(g)}_{|q_2kq_2q_1q_3q_1|}
\\ \hphantom{\Omega_{q_1,q_2,q_3}^{(g)}=+}
{}+G^{(g)}_{|q_3kq_3q_1q_2q_1|}
+G^{(g)}_{|q_3kq_3q_2q_1q_2|}
+G^{(g)}_{|kq_1kq_2q_3q_2|}+G^{(g)}_{|kq_1kq_3q_2q_3|}
\\ \hphantom{\Omega_{q_1,q_2,q_3}^{(g)}=+}
{}+G^{(g)}_{|kq_2kq_3q_1q_3|}
+G^{(g)}_{|kq_2kq_1q_3q_1|}
+G^{(g)}_{|kq_3kq_1q_2q_1|}+G^{(g)}_{|kq_3kq_2q_1q_2|}\Big)
\\ \hphantom{\Omega_{q_1,q_2,q_3}^{(g)}=}
{}+ 2\sum_{g_1+g_2=g}\frac{1}{N}\sum_{k=1}^N \Big(G^{(g_1)}_{|q_1q_2|} G^{(g_2)}_{|q_3k|q_1q_2|}
+G^{(g_1)}_{|q_1q_2|}G^{(g_2)}_{|q_1q_2q_1q_3|}+G^{(g_1)}_{|q_1q_2|} G^{(g_2)}_{|q_2q_1q_2q_3|}
\\ \hphantom{\Omega_{q_1,q_2,q_3}^{(g)}=+}
{} + G^{(g_1)}_{|q_2q_3|}G^{(g_2)}_{|q_1k|q_2q_3|}
+G^{(g_1)}_{|q_2q_3|}G^{(g_2)}_{|q_2q_3q_2q_1|}+G^{(g_1)}_{|q_2q_3|}G^{(g_2)}_{|q_3q_2q_3q_1|}
\\ \hphantom{\Omega_{q_1,q_2,q_3}^{(g)}=+}
{} +G^{(g_1)}_{|q_3q_1|} G^{(g_2)}_{|q_2k|q_3q_1|}
+G^{(g_1)}_{|q_3q_1|}G^{(g_2)}_{|q_3q_1q_3q_2|}+G^{(g_1)}_{|q_3q_1|}G^{(g_2)}_{|q_1q_3q_1q_2|}\Big)
\\ \hphantom{\Omega_{q_1,q_2,q_3}^{(g)}=}
{}+\frac{1}{N^2}\sum_{k,l=1}^N \Big(G^{(g-1)}_{|q_1k|q_2l|q_3|q_3|}+G^{(g-1)}_{|q_2k|q_3l|q_1|q_1|}
+G^{(g-1)}_{|q_3k|q_1l|q_2|q_2|}\Big)
\\ \hphantom{\Omega_{q_1,q_2,q_3}^{(g)}=}
{}+\frac{1}{N}\sum_{l=1}^N \Big(G^{(g-1)}_{|q_1q_1q_2|q_3l|q_1|}+G^{(g-1)}_{|q_1q_1q_3|q_2l|q_1|}
+G^{(g-1)}_{|q_2q_2q_3|q_1l|q_2|}+G^{(g-1)}_{|q_2q_2q_1|q_3l|q_2|}
\\ \hphantom{\Omega_{q_1,q_2,q_3}^{(g)}=+}
{}+G^{(g-1)}_{|q_3q_3q_1|q_2l|q_3|}+G^{(g-1)}_{|q_3q_3q_2|q_1l|q_3|}\Big)
\\ \hphantom{\Omega_{q_1,q_2,q_3}^{(g)}=}
{}+G^{(g-1)}_{|q_1q_1q_2q_3q_2|q_1|}+G^{(g-1)}_{|q_1q_1q_3q_2q_3|q_1|}
+G^{(g-1)}_{|q_2q_2q_3q_1q_3|q_2|}+G^{(g-1)}_{|q_2q_2q_1q_3q_1|q_2|}
\\[.3ex] \hphantom{\Omega_{q_1,q_2,q_3}^{(g)}=}
{}+G^{(g-1)}_{|q_3q_3q_1q_2q_1|q_3|}+G^{(g-1)}_{|q_3q_3q_2q_1q_2|q_3|}
+ G^{(g-1)}_{|q_1q_1q_2q_1q_3|q_1|}+G^{(g-1)}_{|q_1q_1q_3q_1q_2|q_1|}
\\[.3ex] \hphantom{\Omega_{q_1,q_2,q_3}^{(g)}=}
{}+G^{(g-1)}_{|q_2q_2q_3q_2q_1|q_2|}+G^{(g-1)}_{|q_2q_2q_1q_2q_3|q_2|}
+G^{(g-1)}_{|q_3q_3q_1q_3q_2|q_3|}+G^{(g-1)}_{|q_3q_3q_2q_3q_1|q_3|}
\\[.3ex] \hphantom{\Omega_{q_1,q_2,q_3}^{(g)}=}
{}+ G^{(g-1)}_{|q_1q_1q_2|q_1q_1q_3|}
+G^{(g-1)}_{|q_2q_2q_3|q_2q_2q_1|}
+G^{(g-1)}_{|q_3q_3q_1|q_3q_3q_2|}
\\[.3ex] \hphantom{\Omega_{q_1,q_2,q_3}^{(g)}=}
+\frac{1}{N}\sum_{k=1}^N \Big(
G^{(g-1)}_{|q_1kq_2k|q_3|q_3|}+G^{(g-1)}_{|q_2kq_3k|q_1|q_1|}
+G^{(g-1)}_{|q_3kq_1k|q_2|q_2|}
\\[.3ex] \hphantom{\Omega_{q_1,q_2,q_3}^{(g)}=+}
{}+G^{(g-1)}_{|q_1kq_1q_2|q_3|q_3|}+G^{(g-1)}_{|q_1kq_1q_3|q_2|q_2|}
+G^{(g-1)}_{|q_2kq_2q_3|q_1|q_1|}+G^{(g-1)}_{|q_2kq_2q_1|q_3|q_3|}
\\[.3ex] \hphantom{\Omega_{q_1,q_2,q_3}^{(g)}=+}
{}+G^{(g-1)}_{|q_3kq_3q_1|q_2|q_2|}+G^{(g-1)}_{|q_3kq_3q_2|q_1|q_1|}\Big)
\\ \hphantom{\Omega_{q_1,q_2,q_3}^{(g)}=}
{}+2\sum_{g_1+g_2=g-1} \Big(G^{(g_1)}_{|q_1q_2|}G^{(g_2)}_{|q_1q_2|q_3|q_3|}
+ G^{(g_1)}_{|q_2q_3|}G^{(g_2)}_{|q_2q_3|q_1|q_1|}
+G^{(g_1)}_{|q_3q_1|}G^{(g_2)}_{|q_3q_1|q_2|q_2|}\Big)
\\ \hphantom{\Omega_{q_1,q_2,q_3}^{(g)}=}
{}+4 \sum_{g_1+g_2=g-1}\frac{1}{N} \sum_{k=1}^N\Big(
G^{(g_1)}_{|q_1|q_2|} G^{(g_2)}_{|q_3k|q_1|q_2|}
+G^{(g_1)}_{|q_2|q_3|} G^{(g_2)}_{|q_1k|q_2|q_3|}
+G^{(g_1)}_{|q_3|q_1|} G^{(g_2)}_{|q_2k|q_3|q_1|} \Big)
\\ \hphantom{\Omega_{q_1,q_2,q_3}^{(g)}=}
{}+4 \sum_{g_1+g_2=g-1}\Big(
G^{(g_1)}_{|q_1|q_2|} G^{(g_2)}_{|q_2|q_1q_1q_3|}+G^{(g_1)}_{|q_1|q_2|} G^{(g_2)}_{|q_1|q_2q_2q_3|}
+G^{(g_1)}_{|q_2|q_3|} G^{(g_2)}_{|q_3|q_2q_2q_1|}
\\ \hphantom{\Omega_{q_1,q_2,q_3}^{(g)}=+}
{}+G^{(g_1)}_{|q_2|q_3|} G^{(g_2)}_{|q_2|q_3q_3q_1|}
+G^{(g_1)}_{|q_3|q_1|} G^{(g_2)}_{|q_1|q_3q_3q_2|}+G^{(g_1)}_{|q_3|q_1|} G^{(g_2)}_{|q_3|q_1q_1q_2|}
\Big)
\\ \hphantom{\Omega_{q_1,q_2,q_3}^{(g)}=}
{}+ 8\sum_{g_1+g_2+g_3=g-1} G^{(g_1)}_{|q_1|q_2|}G^{(g_2)}_{|q_2|q_3|} G^{(g_3)}_{|q_3|q_1|}
+G^{(g-2)}_{|q_1q_1q_2|q_3|q_3|q_1|}+G^{(g-2)}_{|q_1q_1q_3|q_2|q_2|q_1|}
\\ \hphantom{\Omega_{q_1,q_2,q_3}^{(g)}=}
+G^{(g-2)}_{|q_2q_2q_3|q_1|q_1|q_2|}+G^{(g-2)}_{|q_2q_2q_1|q_3|q_3|q_2|}
+G^{(g-2)}_{|q_3q_3q_1|q_2|q_2|q_3|} +G^{(g-2)}_{|q_3q_3q_2|q_1|q_1|q_3|}
\\ \hphantom{\Omega_{q_1,q_2,q_3}^{(g)}=}
{}+ \frac{1}{N}
\sum_{k=1}^N \Big(G^{(g-2)}_{|q_1k|q_2|q_2|q_3|q_3|}+G^{(g-2)}_{|q_2k|q_3|q_3|q_1|q_1|}
+G^{(g-2)}_{|q_3k|q_1|q_1|q_2|q_2|}\Big)
\\ \hphantom{\Omega_{q_1,q_2,q_3}^{(g)}=}
{}+4 \sum_{g_1+g_2=g-2}\Big( G^{(g_1)}_{|q_1|q_2|} G^{(g_2)}_{|q_1|q_2|q_3|q_3|}
+G^{(g_1)}_{|q_2|q_3|} G^{(g_2)}_{|q_2|q_3|q_1|q_1|}
+G^{(g_1)}_{|q_3|q_1|} G^{(g_2)}_{|q_3|q_1|q_2|q_2|} \Big)
\\ \hphantom{\Omega_{q_1,q_2,q_3}^{(g)}=+}
{}+ G^{(g-3)}_{|q_1|q_1|q_2|q_2|q_3|q_3|}.
\end{gather*}
\end{Proposition}

\begin{proof}
Applications of $\hat{T}_{q_3}$ to the first line of (\ref{TqOmq}) gives with (\ref{eom-2})
\begin{gather*}
\hat{T}_{q_3}\hat{T}_{q_2}\Omega_{q_1}
=N^2\sum_{j,k,l=1}^N \frac{D^2}{DM^{[q_1,k]}}
\bigg\{\frac{\partial^2}{\partial M_{q_2l}\partial M_{lq_2}}
\bigg(\frac{\partial^2\log \mathcal{Z}(M)}{\partial M_{q_3j}\partial M_{jq_3}}
+\frac{\partial \log \mathcal{Z}(M)}{\partial M_{q_3j}}
\frac{\partial \log \mathcal{Z}(M)}{\partial M_{jq_3}}\bigg)
\\ \hphantom{\hat{T}_{q_3}\hat{T}_{q_2}\Omega_{q_1}=}
{}+\frac{\partial}{\partial M_{q_2l}}
\bigg(\frac{\partial^2\log \mathcal{Z}(M)}{\partial M_{q_3j}\partial M_{jq_3}}
+\frac{\partial \log \mathcal{Z}(M)}{\partial M_{q_3j}}
\frac{\partial \log \mathcal{Z}(M)}{\partial M_{jq_3}}\bigg)
\frac{\partial \log \mathcal{Z}(M)}{\partial M_{lq_2}}
\\ \hphantom{\hat{T}_{q_3}\hat{T}_{q_2}\Omega_{q_1}=}
{}+\frac{\partial}{\partial M_{lq_2}}
\bigg(\frac{\partial^2\log \mathcal{Z}(M)}{\partial M_{q_3j}\partial M_{jq_3}}
+\frac{\partial \log \mathcal{Z}(M)}{\partial M_{q_3j}}
\frac{\partial \log \mathcal{Z}(M)}{\partial M_{jq_3}}\bigg)
\frac{\partial \log \mathcal{Z}(M)}{\partial M_{q_2l}} \bigg\}\bigg|_{M=0}.
\end{gather*}
A similar discussion as before gives the assertion.
\end{proof}

Using the previous Examples \ref{Ex2}, \ref{Ex4} and \ref{Ex2+2}, it
is easy to write $\Omega_{q_1}^{(0)}$ and $\Omega_{q_1,q_2}^{(0)}$ up to order~$\lambda^2$ and $\Omega_{q_1,q_2,q_3}^{(0)}$ up to order~$\lambda^1$.

\section{Results connected with blobbed topological recursion}\label{sec:BTR}

In the next subsection we briefly recall the main construction of
\cite{Branahl:2020yru} how \textit{explicit and exact} results for
$\Omega^{(g)}_{b_1,\dots,b_m}$ are obtained. Afterwards, the first few
examples are expanded in $\lambda$ and shown to reproduce the
perturbative results.

\subsection{Summary of previous work}

Our main tool is the usage of \textit{Dyson--Schwinger equations}. They are first
derived for the correlation functions $G^{(g)}_{\dots}$ introduced
before, then complexified to functions $G^{(g)}$ of several complex
variables which satisfy
$G^{(g)}\big(E_{p^1_1},\dots,E_{p^1_{n_1}}|\cdots |
E_{p^b_1},\dots,E_{p^b_{n_b}}\big)=G^{(g)}_{|p^1_1\cdots p^1_{n_1}|\cdots |p^b_1\cdots p^b_{n_b}|}$. After
complexification one can admit multiplicities of the $E_k$, i.e., we
assume that $(e_1,\dots,e_d)$ are the pairwise different values in
$(E_1,\dots,E_N)$ which arise with multiplicities $(r_1,\dots,r_d)$,
respectively, with $r_1+\dots +r_d=N$. It is also straightforward to take
a limit, where the $e_k$ continuously fill an interval with a certain spectral
measure. As mentioned before, there
is a closed non-linear equation \cite{Grosse:2009pa,Grosse:2012uv} for the planar 2-point function alone
and an infinite hierarchy of affine equations for all other functions. A~continuum variant of
the non-linear equation was solved in~\cite{Panzer:2018tvy} for the 2-dimensional Moyal case and later
in~\cite{Grosse:2019jnv} in full generality. It suggested an ansatz in which an implicitly
defined function $R\colon \hat{\mathbb{C}}\to\hat{\mathbb{C}}$, where
$\hat{\mathbb{C}}=\mathbb{C}\cup \{\infty\}$, is crucial:
\begin{Theorem}[\cite{Schurmann:2019mzu-v3}]\label{thmOm1}
Let $(E_1,\dots,E_N)$ be partitioned into
pairwise different $e_1,\dots,e_d>0$ which arise with multiplicities $(r_1,\dots,r_d)$, respectively.
Assume that the complexification $\hat{\Omega}^{(0)}(E_q)\allowbreak=\Omega^{(0)}_q$
can be expressed as\vspace{-1ex}
\begin{gather}
\hat{\Omega}^{(0)}(R(z))=:\Omega^{(0)}_1(z)
=-\frac{R(-z)+R(z)}{\lambda}-\frac{1}{N}\sum_{k=1}^d \frac{r_k}{R(\varepsilon_k)-R(z)}
 \label{eq:om01}
\end{gather}
for some meromorphic function $R$ of degee $d+1$ with $R(\varepsilon_k)=e_k$ and
$R(\infty)=\infty$. Then $R$ is, for generic values of $(e_k)$, uniquely determined by the
non-linear Dyson--Schwinger equation to\vspace{-1ex}
\begin{gather}
 R(z):=z-\frac{\lambda}{N}\sum_{k=1}^d\frac{\varrho_k}{\varepsilon_k+z},\qquad
 \varrho_k=\frac{r_k}{R'(\varepsilon_k)}.
 \label{eq:R}
\end{gather}
Choosing $\lim\limits_{\lambda\to 0}\varepsilon_k=e_k$, then any $(e_k)$ is generic for $\lambda$ in an open
(real or complex) neighbourhood of $0$.
\end{Theorem}

The implicitly defined function $R$ provides a ramified covering
$R\colon \hat{\mathbb{C}}\to\hat{\mathbb{C}}$ of
Riemann spheres, see Figure~\ref{fig:complexification1}. The important
observation is that $R$ pulls $\hat{\Omega}^{(0)}(\zeta)$ back to
a \emph{rational} function~$\Omega_1^{(0)}(z)$. This rationality on the
(right) $z$-plane of Figure~\ref{fig:complexification1}
extends to all other correlation functions. In contrast, when expressing these functions in terms of
the original variables~$(e_k)$ we need to invert $R$ which in closed form is not possible beside $d=1$.
\vspace{-1ex}
\begin{figure}[h!]
\centering
	\includegraphics[width= 0.99\textwidth]{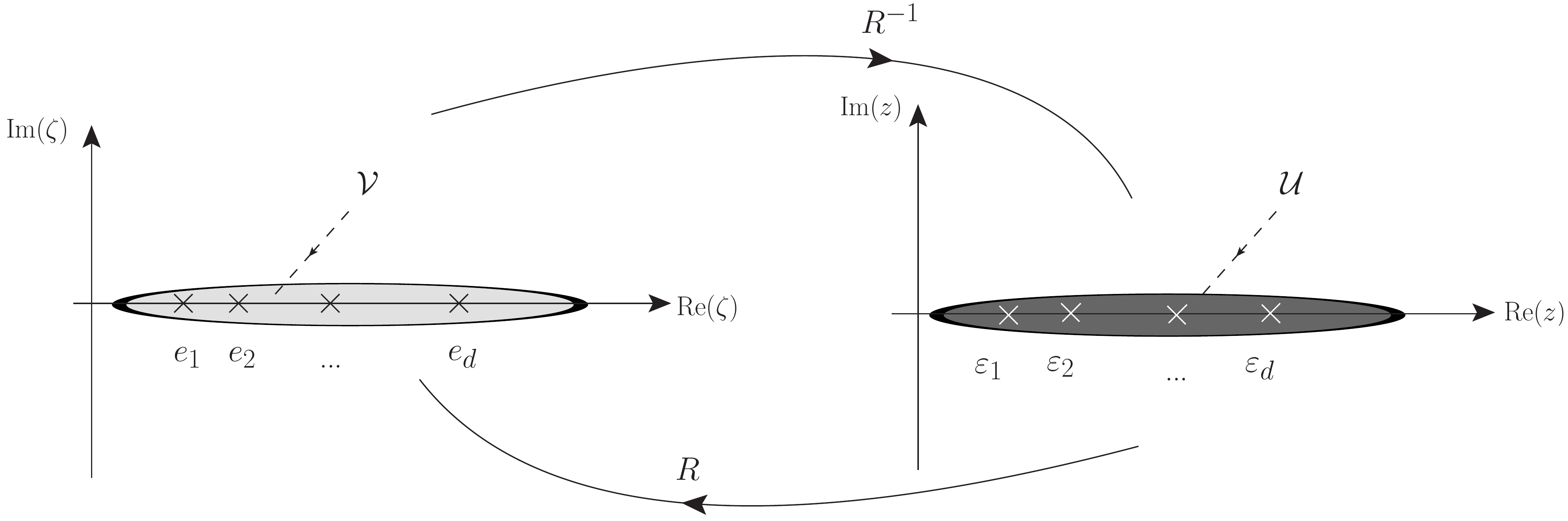}
	\caption{Illustration of the ramified covering map $R\colon \hat{\mathbb{C}}\to \hat{\mathbb{C}}$
		satisfying $R(\varepsilon_k)=e_k$. The map $R$ is biholomorphic between the
		neighbourhoods $\mathcal{V}$ and $\mathcal{U}$.
		\label{fig:complexification1}}\vspace{-1ex}
\end{figure}

It was understood in~\cite{Branahl:2020yru} that, although all other
complexified functions $G^{(g)}$ satisfy affine Dyson--Schwinger
equations (see \cite{Hock:2020rje}), an explicit solution must first
be achieved for the auxiliary functions $\Omega^{(g)}_m(z_1,\dots,z_m)$
with
$\Omega^{(g)}_m(\varepsilon_{q_1},\dots,\varepsilon_{q_m})=\Omega^{(g)}_{q_1,\dots,q_m}$. We~refer to \cite{Branahl:2020yru} for details about the solution
strategy for $\Omega^{(g)}_m(z_1,\dots,z_m)$. Here we only quote the
remarkably simple result:

\begin{Proposition}[\cite{Branahl:2020yru}]\label{propOm2}
 Let $R(z)$ be as in Theorem~$\ref{thmOm1}$ and $\beta_i$ for $i\in\{1,\dots,2d\}$ be the $2d$
 solutions of $R'(z)=0$. We~have the solutions
\begin{gather*}
\Omega^{(0)}_{2}(u,z)=\frac{1}{R'(u)R'(z)}\bigg(\frac{1}{(u-z)^2}+\frac{1}{(u+z)^2}\bigg),
\\
\Omega^{(0)}_{3}(u,v,z)=\frac{1}{R'(u)R'(v)R'(z)} \frac{\partial^3}{\partial u \partial v \partial z}
\Bigg[\frac{\lambda \big(\frac{1}{v+u}+\frac{1}{v-u}\big)}{R'(u)R'(-u)(z+u)}
+\frac{\lambda \big(\frac{1}{u+v}+\frac{1}{u-v}\big)}{R'(v)R'(-v)(z+v)}
\\ \hphantom{\Omega^{(0)}_{2}(u,z)=}
{}+\sum_{i=1}^{2d}
\frac{\lambda \big(\frac{1}{v+\beta_i}+\frac{1}{v-\beta_i}\big)
 \big(\frac{1}{u+\beta_i}+\frac{1}{u-\beta_i}\big)}{R'(-\beta_i)R''(\beta_i) (z-\beta_i)}
		\Bigg].\vspace{-1ex}
\end{gather*}
The function $\Omega^{(0)}_{3}(u,v,z)$ is completely symmetric in its arguments.
\end{Proposition}

In \cite{Branahl:2020yru} also the solutions of $\Omega^{(g)}_{m}$ with
$(g,m)=(0,4),(1,1)$ are derived, leading to the
conjecture that meromorphic forms $\omega_{g,m}$ defined by
$\omega_{0,1}(z)=-R(-z)R'(z){\rm d}z$ and for $2g+m\geq 2$ by
$\omega_{g,m}(z_1,\dots,z_m)=\lambda^{2-2g-m}\Omega^{(g)}_m(z_1,\dots,z_m)\prod_{j=1}^m R'(z_j){\rm d}z_j$
follow blobbed topological recursion for the spectral curve
\begin{gather*}
 (x\colon \Sigma\to \Sigma_0,\omega_{0,1}(z),B(u,z))
 =\bigg(R\colon \hat{\mathbb{C}}\to\hat{\mathbb{C}},\,-R(-z){\rm d}R(z),\,\frac{{\rm d}u\,{\rm d}z}{(u-z)^2}\bigg).
\end{gather*}
This means that the poles of $\omega_{g,m}$ at the ramification points
$\beta_i$ of $R$ are given by the universal formula of topological
recursion. These are enriched by further contributions which are
holomorphic at $\beta_i$ and have poles at $z_i=0$ and $z_i+z_j=0$ for
the quartic Kontsevich model, starting with the appearance of some
additional initial data in $\omega_{0,2}$, namely
$\frac{{\rm d}u\,{\rm d}z}{(u+z)^2}=-B(u,-z)$ (Bergman kernel with one changed
sign). Their general structure is not yet understood for $g>0$, but we
are confident that the following structures below extend to higher
genera: The symmetry of the spectral curve, $y(z)=-x(-z)$ and
$\omega_{0,2}(u,z)=B(u,z)-B(u,-z)$ gave the motivation to the deep
involution identity \cite{Hock:2021tbl}
\begin{gather}
\omega_{0,|I|+1}(I,q) +\omega_{0,|I|+1}(I,- q)\nonumber
\\ \qquad
{}=-\sum_{s=2}^{|I|} \sum_{I_1\uplus \dots \uplus I_s=I}
\frac{1}{s} \Res\displaylimits_{z\to q} \bigg(
\frac{{\rm d}R(-q) {\rm d}R(z)}{(R(-z)-R(-q))^{s}} \prod_{j=1}^s
\frac{\omega_{0,|I_j|+1}(I_j,z)}{{\rm d}R(z)}\bigg)
\label{eq:flip-om}
\end{gather}
 completely determining the meromorphic forms
$\omega_{0,m+1}$
by usual topological recursion (polar at~$\beta_i$) and something surprisingly similar giving the holomorphic parts:
\begin{Theorem}[\cite{Hock:2021tbl}]\label{thmBTRplan}
 Assume that $z\mapsto \omega_{0,n+1}(u_m,\dots,u_m,z)$ is for $m\geq 2$
 holomorphic at
 $z=-\beta_i$ and $z=u_k$ and has poles at most in points where the
 r.h.s.\ of \eqref{eq:flip-om} has poles. Then equation \eqref{eq:flip-om} is
 for $I=\{u_1,\dots,u_m\}$ with $m\geq 2$ uniquely solved by
\begin{align}
 \omega_{0,|I|+1}(I,z)= {}& \sum_{i=1}^r
\Res\displaylimits_{q\to \beta_i}K_i(z,q)
 \sum_{I_1\uplus I_2=I}^\prime \omega_{0,|I_1|+1}(I_1,q)\omega_{0,|I_2|+1}(I_2,\sigma_i(q))
\nonumber
 \\
&-\sum_{k=1}^m {\rm d}_{u_k} \bigg[\Res\displaylimits_{q\to - u_k}
\!\!\sum_{I_1\uplus I_2=I}^{\prime}\!\!\!\tilde{K}(z,q,u_k)
{\rm d}_{u_k}^{-1}\big(\omega_{0,|I_1|+1}(I_1,q)
\omega_{0,|I_2|+1}(I_2,q)\big)\bigg],\!
 \label{sol:omega}
\end{align}
where the primed sum excludes the empty sets $I_i=\varnothing$.
\end{Theorem}
Here $\sigma_i\neq \mathrm{id}$ denotes the local Galois involution
in the vicinity of $\beta_i$, i.e., $R(\sigma_i(z))=R(z)$,
$\lim\limits_{z\to \beta_i}\sigma_i(z)=\beta_i$.
By $d_{u_k}$ we denote the exterior differential in $u_k$. We~had to introduce two recursion kernels of a similar structure:
\begin{gather*}
K_i(z,q):= \frac{\frac{1}{2} \big(\frac{{\rm d}z}{z-q}-\frac{{\rm d}z}{z-\sigma_i(q)}\big)
}{{\rm d}R(\sigma_i(q))(R(-\sigma_i(q))-R(-q))},\qquad
\tilde{K}(z,q,u):=\frac{\frac{1}{2}\big(\frac{{\rm d}z}{z-q}-\frac{{\rm d}z}{z+u}\big)}{{\rm d}R(q)
 (R(u)-R(-q))}.
\end{gather*}
The linear and quadratic loop equations hold. We~emphasise that
Theorem~\ref{thmBTRplan} gives exactly the decomposition of
$\omega_{0,m}$ into
$\mathcal{P}_z\omega_{0,m}+\mathcal{H}_z\omega_{0,m}$ with projectors
$\mathcal{P}_z$ and $\mathcal{H}_z$ on poles at the ramification
points and the antidiagonal of a variable $z$. The kernels $K$ and
$\tilde{K}$ are constructed from the Bergman kernel $B(u,z)$ only, and
not by the full $\omega_{0,2}$. This is a clear difference to the
construction of Eynard and Orantin~\cite{Eynard:2007kz}, which in this
case would artificially produce (false) poles at negative ramification
points. The blobs $\varphi_{g,n}$ in the sense of Borot and Shadrin~\cite{Borot:2015hna} are defined as purely holomorphic part by the
application of $\mathcal{H}_{z_i}$ on all $m$ variables. This
procedure is straightforward but will give a cumbersome result.
Nevertheless we stress that in the quartic Kontsevich model there is a
remarkably simple rule to compute recursively the holomorphic part, which in the
general theory are only mildly constrained by the quadratic loop
equations.\looseness=1

The extension of (\ref{sol:omega}) to $g\geq 1$ is
work in progress. We~refer to
\cite{Eynard:2016yaa,Eynard:2007kz} for topological recursion in
general and to \cite{Borot:2015hna} for blobbed topological recursion.

\subsection{Comparison between exact results and weighted ribbon graphs}\label{sec:comparison}

In this subsection we compare the exact solutions of Theorem~\ref{thmOm1} and Proposition \ref{propOm2} with the perturbative
expansion via weighted ribbon graphs. First, we need the expansion of
$\varepsilon_a$ and~$R'(\varepsilon_a)$ which is easily obtained by
iterative insertion into the definition of $R(z)$. The first orders
yield:
\begin{gather*}
\varepsilon_q=e_q+\frac{\lambda}{N}\sum_{n=1}^d \frac{r_n}{e_q+e_n}
-\frac{\lambda^2}{N^2}\sum_{n,k=1}^d r_nr_k\bigg( \frac{1}{(e_q+e_n)(e_k+e_n)^2}
+\frac{1}{(e_q+e_n)^2(e_q+e_k)}
\\ \hphantom{\varepsilon_q=}
+\frac{1}{(e_q+e_n)^2(e_n+e_k)}\bigg)+\mathcal{O}\big(\lambda^3\big),
\\
R'(\varepsilon_q)=1+\frac{\lambda}{N}\sum_{n=1}^d \frac{r_n}{(e_q\!+e_n)^2}
-\frac{\lambda^2}{N^2}\sum_{n,k=1}^d r_nr_k\bigg( \frac{1}{(e_q\!+e_n)^2(e_k\!+e_n)^2}
+\frac{2}{(e_q\!+e_n)^3(e_q\!+e_k)}
\\ \hphantom{R'(\varepsilon_q)=}
+\frac{2}{(e_q+e_n)^3(e_n+e_k)}\bigg)+\mathcal{O}\big(\lambda^3\big).
\end{gather*}
Also necessary for $(g,n)=(0,3)$ and higher topologies are the zeroes $\beta_i$ of $R'$ (so-called
ramifications points). The $\lambda$-expansion yields
\begin{align}
 \beta_i&=-e_i+\mathrm{i}\sqrt{\frac{\lambda r_i}{N}}
 -\frac{\lambda}{N}\sum_{n=1}^d \frac{r_n}{e_i+e_n}
 +\mathcal{O}\big(\lambda^{\frac{3}{2}}\big),\qquad \beta_{i+d}=\bar{\beta_i}\qquad i\in\{1,\dots,d\}.
 \label{betaexpansion}
\end{align}
The expansions of $\varepsilon_q$ and $\beta_i$ are easily implemented
into a computer algebra system. Deriving perturbative results for the $\Omega^{(g)}_{n}$
is then straightforward. We~demonstrate this with the following examples:
\begin{Example}
From the expansion of the exact result, we obtain using Theorem \ref{thmOm1}
\begin{align*}
 \Omega^{(0)}_{1}(\varepsilon_q)={}&\frac{\varepsilon_q-e_q}{\lambda}
 +\frac{1}{N}\sum_{n=1}^d r_n\bigg(\frac{1}{R'(\varepsilon_n)(\varepsilon_n-\varepsilon_q)}
 -\frac{1}{e_n-e_q}\bigg)
 \\
={}&\frac{1}{N}\sum_{n=1}^d\frac{r_n}{e_n+e_q}
 -\frac{\lambda}{N^2}\sum_{n,k=1}^d r_nr_k\bigg( \frac{1}{(e_q+e_n)(e_k+e_n)^2}
 +\frac{1}{(e_q+e_n)^2(e_q+e_k)}
 \\
 &+\frac{1}{(e_q+e_n)^2(e_n+e_k)}+\frac{1}{(e_n+e_k)^2(e_n-e_q)} +\frac{\frac{1}{e_n+e_k}-\frac{1}{e_q+e_k}}{(e_n-e_q)^2}\bigg)+\mathcal{O}\big(\lambda^2\big)
 \\
={}&\frac{1}{N}\sum_{n=1}^d\frac{r_n}{e_n+e_q}-\frac{\lambda}{N^2}\!\!\sum_{n,k=1}^d\!\! \bigg(
		 \frac{r_nr_k}{(e_q\!+e_n)^2(e_q\!+e_k)}
		 +\frac{r_nr_k}{(e_q\!+e_n)^2(e_n\!+e_k)}\bigg)
+\mathcal{O}\big(\lambda^2\big).
\end{align*}
This result is in full compliance with the graph expansion of Example \ref{Ex2}
inserted into $\Omega^{(0)}_{q}=\frac{1}{N}\sum_{k=1}^N G^{(0)}_{|qk|}$. The agreement is immediate for
$r_n=1$; otherwise one collects $r_k$ identical terms, where $E_{k_1}=\dots =E_{k_{r_k}}=e_k$.
The expansion of the exact result is represented in a different partial fraction decomposition
than the graph expansion. The reader may check also the next order.
\end{Example}
\begin{Example}
We obtain from Proposition \ref{propOm2}
\begin{align*}
 \Omega^{(0)}_{2}(\varepsilon_q,\varepsilon_r)-\frac{1}{(e_q-e_r)^2}
={}&\frac{1}{R'(\varepsilon_q)R'(\varepsilon_r)}\bigg(\frac{1}{(\varepsilon_q-\varepsilon_r)^2}+\frac{1}{(\varepsilon_q+\varepsilon_r)^2}\bigg)-\frac{1}{(e_q-e_r)^2}
 \\
= {}&\frac{1}{(e_q+e_r)^2}-\frac{\lambda}{N}\sum_{n=1}^d r_n\bigg(2\frac{\frac{1}{e_n+e_q}
 -\frac{1}{e_n+e_r}}{(e_q-e_r)^3}+\frac{\frac{1}{(e_q+e_n)^2}+\frac{1}{(e_r+e_n)^2}}{(e_q-e_r)^2}
 \\
 &+2\frac{\frac{1}{e_n+e_q}+\frac{1}{e_n+e_r}}{(e_q+e_r)^3}
 +\frac{\frac{1}{(e_q+e_n)^2}+\frac{1}{(e_r+e_n)^2}}{(e_q+e_r)^2}\bigg)+\mathcal{O}\big(\lambda^2\big),
\end{align*}
which is in full compliance with the graph expansion of Examples
\ref{Ex2} and \ref{Ex4} inserted into Proposition~\ref{prop:Om02G}
(but in a different partial
fraction decomposition). The reader may check the next order, where
additionally the graphs of the $(2+2)$-point function from Example
\ref{Ex2+2} contribute.
\end{Example}

\begin{Example}
We obtain from Proposition \ref{propOm2}
\begin{align*}
\Omega^{(0)}_{3}(\varepsilon_q,\varepsilon_r,\varepsilon_s)
={}&\frac{1}{R'(\varepsilon_q)R'(\varepsilon_r)R'(\varepsilon_s)}
\Bigg[\frac{\partial}{\partial u}\frac{\lambda \big(\frac{1}{(\varepsilon_r+u)^2}
 +\frac{1}{(\varepsilon_r-u)^2}\big)}{R'(u)R'(-u)(\varepsilon_s+u)^2}\bigg\vert_{u=\varepsilon_q}
 \\
&+\frac{\partial}{ \partial v }\frac{\lambda \big(\frac{1}{(\varepsilon_q+v)^2}
 +\frac{1}{(\varepsilon_q-v)^2}\big)}{R'(v)R'(-v)(\varepsilon_s+v)^2}\bigg\vert_{v=\varepsilon_r}
\\
&-\sum_{i=1}^{2d}\frac{\lambda \big(\frac{1}{(\varepsilon_r+\beta_i)^2}+\frac{1}{(\varepsilon_r-\beta_i)^2}\big) \big(\frac{1}{(\varepsilon_q+\beta_i)^2}+\frac{1}{(\varepsilon_q-\beta_i)^2}\big)}
{R'(-\beta_i)R''(\beta_i)(\varepsilon_s-\beta_i)^2}	\Bigg]
\\
={}&-\lambda \cdot 2\bigg(\frac{\frac{1}{(e_r+e_q)^2}+\frac{1}{(e_r-e_q)^2}}{(e_s+e_q)^3} +\frac{\frac{1}{(e_r+e_q)^3} -\frac{1}{(e_r-e_q)^3}}{(e_s+e_q)^2}
\\
&+\frac{\frac{1}{(e_q+e_r)^2}+\frac{1}{(e_q-e_r)^2}} {(e_s+e_r)^3}+\frac{\frac{1}{(e_q+e_r)^3}-\frac{1}{(e_q-e_r)^3}}{(e_s+e_r)^2}\bigg)
+\mathcal{ O}\big(\lambda^2\big),
\end{align*}
where the restrictions to $u=\varepsilon_q$ and $v=\varepsilon_r$ in
the second line vanish. The only contributions come from the
$i$-summation. This result is in full compliance with the graph
expansion in Examples~\ref{Ex2} and~\ref{Ex4} inserted into
Proposition~\ref{prop:Om03G}, but again in a different partial
fraction decomposition. For the computation, we remark that the
expansion
\begin{gather*}
\frac{1}{(\varepsilon_q+\beta_q)^2}=-\frac{N}{\lambda r_q}
 +\mathcal{O}\bigg(\frac{1}{\sqrt\lambda}\bigg),\qquad
\frac{1}{(\varepsilon_q+\beta_{q+d})^2}=-\frac{N}{\lambda r_q}
 +\mathcal{O}\bigg(\frac{1}{\sqrt\lambda}\bigg),
 \\
\frac{1}{R''(\beta_i)}=-\frac{\mathrm{i}}{2}\sqrt{\frac{\lambda r_i}{N}}
 +\mathcal{O}(\lambda),\qquad
 \frac{1}{R''(\beta_{i+d})}=\frac{\mathrm{i}}{2}\sqrt{\frac{\lambda r_i}{N}}
 +\mathcal{O}(\lambda)		
	\end{gather*}
	indicates a contribution of order $\sqrt{\lambda}$ from the
 $i$-summation, which actually cancels due to the pairs
 $(\beta_i,\beta_{i+d})$ of complex conjugations
 $\bar{\beta_i}=\beta_{i+d}$. The reader may even check that
 the order~$\lambda^{\frac{3}{2}}$ cancels as well.
\end{Example}

\subsection{Combinatorics}
\label{sec:combinatorics}

A common investigation in QFT concerns the growth of the number of
Feynman graphs at a~cer\-tain order $\lambda^v$. In order to illustrate
the enormous complexity of the individual contributions to~$\Omega_{q}^{(0)}$ and $\Omega_{q_1,q_2}^{(0)}$ at a given order
$\lambda^v$, we will calculate these numbers explicitly. To~enter this
regime of enumerative geometry within the quartic Kontsevich model we
have to set $d=1$. We~will now show how to expand $G^{(0)}_{|11|}$,
$G^{(0)}_{|1111|}$ and $G^{(0)}_{|11|11|}$ (the $2$-point, $4$-point
and $(2{+}2)$-point function for a single, $(r_1{=}N)$-fold degenerate
spectral value $e_1=e$) in an exact and generic perturbative series in
$\lambda$. The prefactors of $(-\lambda)^v$ for $e=\frac{1}{2}$ then
simply \textit{count} the number of connected Feynman ribbon graphs
contributing to the graph expansion at order~$v$. As~known from the
Hermitian 1-matrix model \cite{Brezin:1977sv}, the duals of the ribbon
graphs of the quartic Kontsevich model are rooted quadrangulations.
However, due to a different definition of correlation functions, the
correspondence to \cite{Brezin:1977sv} is not one-to-one.\footnote{For
 the Hermitian 1-matrix model one usually considers resolvents which
 from a combinatorial perspective are sometimes called ordinary maps.
 Our correlation correspond for $d=1$ to the so-called fully
 simply maps, see~\cite{Borot:2017agy} for precise definitions.} To
avoid complicated redefinitions, we follow another path.

To derive the exact power series in $\lambda$, return to Theorem~\ref{thmOm1} and solve the $2d$ implicitly defined equations for $d=1$
explicitly. For $\varepsilon_1=\varepsilon$, we solve them to
\begin{gather*}
 \varepsilon= \frac{1}{6}\big(4e+\sqrt{4e^2+12\lambda}\big) , \qquad
 \varrho= \frac{N}{18 \lambda} \big(2e\sqrt{4e^2+12\lambda}-4e^2+12\lambda\big).
\end{gather*}
With the other preimage
$ \hat \varepsilon= -\frac{1}{6}\big(2e+2\sqrt{4e^2+12\lambda}\big)$ one
expresses\footnote{See \cite{Branahl:2020yru} for details about the complexification procedure from
 correlation functions $G^{(g)}_{\dots}$ to meromorphic functions $\mathcal{G}^{(g)}(\dots)$.}
the planar $2$-point func\-tion~as
\begin{gather*}
 G^{(0)}_{|11|}\equiv G^{(0)}(e,e)\equiv G^{(0)}(R(\varepsilon),R(\varepsilon))
 =:\mathcal{G}^{(0)}(\varepsilon,\varepsilon)=
-\frac{2 \hat \varepsilon}{(\varepsilon-\hat
 \varepsilon)^2}.
\end{gather*}
Admitting multiplicities $r_k$ in the definition (\ref{def:Om}) of $\Omega_q^{(0)}$ we find for $d=1$ and $r_1=N$
\begin{gather*}
 \Omega_q^{(0)}=\frac{1}{N}\sum_{k=1}^d r_k G^{(0)}_{|qk|} \longrightarrow
 \Omega_1^{(0)}= G^{(0)}_{|11|} =\mathcal{G}^{(0)}(\varepsilon,\varepsilon).
\end{gather*}
The same steps give for $\Omega^{(0)}_{q_1,q_2}$ according to Proposition~\ref{prop:Om02G}
\begin{align*}
 \Omega_{q_1,q_2}^{(0)}={}&\frac{1}{(e_{q_1}-e_{q_2})^2}+\Big(G^{(0)}_{|q_1q_2|}\Big)^2
 +\frac{1}{N}\sum_{k=1}^N r_k
 \Big(G^{(0)}_{|q_1kq_1q_2|}+G^{(0)}_{|q_2kq_2q_1|}+G^{(0)}_{|q_1kq_2k|}\Big)
\\
& +\frac{1}{N^2}\sum_{k,l=1}^N r_kr_lG^{(0)}_{|q_1k|q_2l|}
\longrightarrow \lim_{e_{q_1},e_{q_2}\to e}\bigg(\Omega_{q_1,q_2}^{(0)}-\frac{1}{(e_{q_1}-e_{q_1})^2}\bigg)
\\
={}&\mathcal{G}^{(0)}(\varepsilon,\varepsilon)^2
+3\mathcal{G}^{(0)}(\varepsilon,\varepsilon,\varepsilon,\varepsilon)
+\mathcal{G}^{(0)}(\varepsilon,\varepsilon|\varepsilon,\varepsilon).
\end{align*}
A lengthy calculation shows
\begin{gather}
 \mathcal{G}^{(0)}(\varepsilon,\varepsilon,\varepsilon,\varepsilon)
 = \frac{8e+12 \sqrt{4e^2+12\lambda}}{\big(2e+ \sqrt{4e^2+12\lambda}\big)^3}
 - 2\big(\mathcal{G}^{(0)}(\varepsilon, \varepsilon)\big)^2, \nonumber
 \\
 \mathcal{G}^{(0)}(\varepsilon,\varepsilon|\varepsilon,\varepsilon)
 = \frac{6\lambda^2}{\big(e+\sqrt{4e^2+12\lambda}/2\big)^6}.
\label{G4d1}
\end{gather}
Inserting the $\lambda$-expansion of these formulae above gives the
number of ribbon graphs contributing to $\Omega_{q}^{(0)}$ and
$\Omega_{q_1,q_2}^{(0)}$ at a given order $\lambda^v$.
In the following table we list these numbers up to order~$\lambda^5$, including also $\Omega_{q_1,q_2,q_3}^{(0)}$ for completeness.
\begin{center}\setlength{\tabcolsep}{4.5pt}
\renewcommand{\arraystretch}{1.5}
\begin{tabular}[h!b]{c|c|c|c}
\hline
Order &$\Omega_{q}^{(0)}$&$\Omega_{q_1,q_2}^{(0)}=\big(\Omega _{q_1,q_2}^{(0)}\big)^{\rm TR}+\big(\Omega_{q_1,q_2}^{(0)}\big)^{\rm BTR}$&$\Omega_{q_1,q_2,q_3}^{(0)}$
\\
\hline
$\lambda^0$ & 1 & $1=0+1$ & 0 \\
\hline
$\lambda^1$ & 2 & $7 =1+6$& 4\\
\hline
$\lambda^2$ & 9 & $58=13+45$ & 84\\
\hline
$\lambda^3$ & 54 & $522=144+378$ & 1322 \\
\hline
$\lambda^4$ & 378 & $4941=1539+3402$ &18684 \\
\hline
$\lambda^5$ & 2916 & $48411 =16335+32076$ & 249156 \\
\hline
\end{tabular}
\end{center}
Those numbers can be checked at low orders of $\lambda$ counting the
diagrams in Figures~\ref{fig:2PGraphs}, \ref{fig:4PGraphs} and~\ref{fig:2+2PGraphs}. Regarding $\Omega_{q}^{(0)}$ we encounter very
special numbers. By duality these are the same as the numbers
$m_{g=0}(n)$ of planar ($g=0$) quadrangulations with $n$ faces plus a
(marked) boundary of length $2$. We~recall:
 \begin{Theorem}[{\cite[Chapter~3.1.7]{Eynard:2016yaa}}]\label{tutte}
 The number of rooted quadrangulations of the sphere
 with~$n$ quandrangles is given by
\begin{equation*}
2\cdot 3^n\cdot \frac{(2n)!}{n!(n+2)!} = 2\cdot 3^n\cdot \frac{C_{n}}{n+2}
\end{equation*}
with the Catalan number $C_n$.
\end{Theorem}
\noindent The planar 2-point function for $d=1$ itself generates these
numbers together with weights $\frac{1}{2e}$ of the edges.

The result for the number of ribbon graphs of order $\lambda^n$
contributing to $\Omega_{1,1}^{(0)}$ (we splitted the number into the usual Bergman kernel and the blob) can be derived from the Taylor
series of~(\ref{G4d1}) whose first terms are
\begin{gather*}
 \mathcal{G}^{(0)}(\varepsilon, \varepsilon|\varepsilon, \varepsilon)=\frac{6(-\lambda)^2}{(2e)^6}
 +\frac{108(-\lambda)^3}{(2e)^8}+\frac{1458(-\lambda)^4}{(2e)^{10}}
 +\frac{17820(-\lambda)^5}{(2e)^{12}}+\cdots,
 \\
 \mathcal{G}^{(0)}(\varepsilon, \varepsilon,\varepsilon, \varepsilon) =
 \frac{(-\lambda)}{(2e)^4}+\frac{10(-\lambda)^2}{(2e)^6}+\frac{90(-\lambda)^3}{(2e)^8}
 +\frac{810(-\lambda)^4}{(2e)^{10}}+\frac{7425(-\lambda)^5}{(2e)^{12}}+\cdots.
\end{gather*}
These numbers give experimental evidence for the footnote from above that $G^{(g)}_{|p_1^1\cdots p_{n_1}^1|\cdots |p_1^b\cdots p_{n_b}^b|}$ are generating series of \textit{fully simple maps}. Building on results concerning quadrangulations from Bernardi and Fusy \cite{bernardi2017bijections}, we can express the above series in a closed form:
 \begin{gather*}
 \mathcal{G}^{(0)}(\varepsilon, \varepsilon|\varepsilon, \varepsilon)= 4 \binom{3}{1}^2 \sum_{m=0}^{\infty} \frac{3^m(6+2m-1)!}{m!(6+m)!}\frac{(- \lambda)^{m+2}}{(2e)^{2m+6}},
 \\
 \mathcal{G}^{(0)}(\varepsilon, \varepsilon,\varepsilon, \varepsilon) = \frac{6!}{3!2!} \sum_{m=0}^{\infty} \frac{3^{m-1}(3+2m)!}{m!(5+m)!} \frac{ (-\lambda)^{m+1}}{(2e)^{2m+4}}.
\end{gather*}
The separation into pure TR and BTR additions and its combinatorial interpretation is work in progress.

A more sophisticated generating function takes also the non-trivial
automorphism groups of connected closed ribbon graphs into account~-- the \textit{free energy} $\mathcal{F}^{(g)}$. As a final illustration, we~will
reproduce the power series of planar closed ribbon graphs of
Figure~\ref{fig:freeenergy} using a representation of $\mathcal{F}^{(0)}$. For
this object, we need to define a couple of quantities that are based
on the general results of~\cite{Eynard:2007kz}:
\begin{Definition}
Consider the poles $a=\{\pm \varepsilon_i, \infty \}$ of $\omega_{0,1}(z)=-R(-z)R'(z)$. Define the
\begin{itemize}\itemsep=0pt
 \item \textit{temperatures} by $t_a=\Res\displaylimits_{z \to a} \omega_{0,1}(z)$;
\item \textit{local variables} for poles of $R'(z)$ via $\xi_a(z)=\frac{1}{R(z)}$ and for poles of $R(-z)$ via $\xi_a(z)=\frac{1}{R(z)-R(a)}$;
\item \textit{potential}
\begin{gather*}
 V_a(z)=\Res\displaylimits_{q \to a}
 \omega_{0,1}(q)\log \bigg(1- \frac{\xi_a(z)}{\xi_a(q)}\bigg)
 =\sum_{k=1}^{\mathrm{deg} V_a}t_{a,k}\xi_a^k(z),
\end{gather*}
where $t_{a,k}$ defines the \textit{moduli} of the pole $a$;
\item \textit{loop annihilation operator} by the primitive
 of $\omega_{0,1}(z)$:
\begin{gather*}
 \Phi(z)=\int^{z'=z} \omega_{0,1}(z')=\frac{z^2}{2}+\frac{\lambda}{N} \sum_k \bigg[ \frac{r_kR(\varepsilon_k)}{R'(\varepsilon_k)(z+\varepsilon_k)}
 + r_k \log \bigg(\frac{z+\varepsilon_k}{z-\varepsilon_k} \bigg) \bigg].
\end{gather*}
\end{itemize}
Using these objects, we can formulate the genus zero free energy:
\begin{gather*}
 \mathcal{F}^{(0)} = \frac{1}{2} \sum_a \big [ \Res\displaylimits_{q \to a}
 \omega_{0,1}(q)V_a(q) +t_a \mu_a \big] + \mathcal{R}, \qquad
 \mu_a:= \lim_{q \to a} \big(V_a(q) -t_a\log[\xi_a(q)]-\Phi(q) \big).
\end{gather*}
It is necessary to add a compensating term
$\mathcal{R}:=-\frac{\lambda}{2N}\sum_k r_k e_k^2 +
\frac{\lambda^2}{2N^2} \sum_{k,i} r_ir_k \log(e_i-e_k)$ since
$\Omega^{(0)}_1(z)$ differs from $-R(-z)$ (recall $\omega_{0,1}(z)=-R(-z)R'(z)$). Acting with the creation operator on $\mathcal{R}$ exactly yields the additional terms in (\ref{eq:om01}):
\begin{gather*}
 \frac{\partial \mathcal{R}}{\partial e_b}= \Omega^{(0)}_1(\varepsilon_b)+R(-\varepsilon_b)
\end{gather*}
giving rise to a well-defined initial data also at $z=\varepsilon_b$, where $R(-z)$ itself becomes singular.
\end{Definition}

For our purposes, set again $d=1$ and calculate
$t_{\pm \varepsilon}=\mp \lambda$, $t_{\infty}=0$ as well as
$V_{ \varepsilon}(z)=0$, $V_{-\varepsilon}(z) =-e_1\cdot R(z)$,
$V_{\infty}(z)=\frac{R(z)^2}{2}$. Only the residue for the pole
$-\varepsilon$ gives a contribution, namely
$\frac{\lambda \rho}{16 \varepsilon^4}\big(16 \varepsilon^6 - 4
\varepsilon^4\lambda \rho+(\lambda \rho)^3\big)$. We~also calculate
\begin{gather*}
t_{ \varepsilon}\mu_{\varepsilon} +t_{- \varepsilon}\mu_{-\varepsilon}
=\lambda \bigg[\varepsilon ^2-\frac{\lambda^2\rho^2}{4 \varepsilon^2}+\lambda \log \bigg(1+\frac{4\varepsilon^2}{\lambda \rho } \bigg)\bigg].
\end{gather*}
Inserting the explicit solutions for $\epsilon$ and $\rho$ and taking
the compensation $\mathcal{R}$ into consideration, we can expand the
result in a power series:
 \begin{align*}
\mathcal{F}^{(0)}&= \sum_{n=1}^{\infty} 3^n \frac{(2n-1)!}{n!(n-2)!} \frac{(-\lambda)^n}{(2e)^{2n}}
\\
&=\frac{-\lambda}{2(2e)^2}+\frac{9(-\lambda)^2}{8(2e)^4} +\frac{9(-\lambda)^3}{2(2e)^6} +\frac{189(-\lambda)^4}{8(2e)^8} +\frac{729(-\lambda)^5}{5(2e)^{10}} +\cdots.
\end{align*}
We encounter, e.g., the automorphism groups
$\frac{1}{2}+\frac{1}{2}+\frac{1}{8}$ of the graphs at
$\mathcal{O}\big(\lambda^2\big)$ in Figure~\ref{fig:freeenergy}. Again, we can
cite a former result from the Hermitian 1-matrix model:
\begin{Proposition}[{\cite[Chapter~3.6.1]{Eynard:2016yaa}}]
The generating function of non-rooted quadrangulations of the sphere reads, including automorphism groups:
 \begin{align*}
& \frac{1}{6\big(1+\sqrt{1-12\lambda}\big)^2} -\frac{5}{6\big(1+\sqrt{1-12\lambda}\big)} +\frac{3}{8} -\frac{\log\big(1+\sqrt{1-12\lambda}\big)}{4}.
\end{align*}
\end{Proposition}
Our expression of $\mathcal{F}^{(0)}$ is at first sight far more complicated,
but can be reduced to this expression~-- the power series of both
expressions are equal (up to a sign of $\lambda$). Deriving the result with respect to $e$ proves again the correct action of the creation operator becoming a~trivial derivative for closed graphs: one recovers the numbers $2,9,54,378,\dots $ from Lemma~\ref{tutte}.
\begin{Remark}
 Despite the appearance of blobbed topological recursion in our
 model, $\mathcal{F}^{(0)}$ fits into the usual picture of
 topological recursion for obvious reasons. However, the other
 special free energy $\mathcal{F}^{(1)}$ needs additional terms
 responsible for the blob (work in progress).
\end{Remark}

\section{Critical coupling constants and geometric discussion}\label{sec:critical}

So far, we were able to show analytically the expected coincidence
between the exact solutions from (blobbed) topological recursion of
the quartic Kontsevich model and their perturbative expansion in
the coupling constant $\lambda$. Many systems of statistical physics, quantum mechanics
and quantum field theory show critical phenomena and phase transitions when parameters take
particular values. This section starts to explore such phenomena in the
quartic Kontsevich model. More precisely, we exemplify transitions between different
stratification types of~the parameter space. This includes
the appearance of higher-order ramifications in the crucial function $R$ identified in
Theorem~\ref{thmOm1} and transitions between different ramification profiles.\looseness=1

\subsection{The setup}

The investigation of special cases of the quartic Kontsevich
model already suggested certain values of $\lambda$ at which a
critical behaviour occurs. In \cite{Grosse:2019jnv} a
scaling limit $d,N \to \infty$ of~(\ref{eq:R}) to a renormalised
integral representation
$R(z)=z-\lambda (-z)^{D/2} \int_0^\infty
\frac{\varrho(t){\rm d}t }{(t+1)^{D/2}(t+1+z)}$ was established, with
$D\in \{0,2,4\}$ the smallest dimension that gives a convergent
integral. We~recall:
\begin{itemize}\itemsep=0pt
\item Let $d=1$ with an $N$-fold degenerate eigenvalue \cite{Grosse:2019jnv}. This is the
 Hermitian 1-matrix model. We~obtain
 $R(\varepsilon)=\varepsilon- \frac{\lambda}{N}
 \frac{\varrho}{2\varepsilon}=e$, where
 $\frac{N}{\varrho}=R'(\varepsilon)$, with solution
 $\varepsilon= \big(e + \sqrt{4e^2+12\lambda}\big)/6$ directly given by
 inversion of $R$. In standard conventions one should identify $e=\frac{1}{2}$ which gives
 a critical value $\lambda_c =-\frac{1}{12}$ below which $R^{-1}$ cannot be defined
 as map between real functions.

\item Let $d \to \infty$ with spectral measure $\varrho(t)=1$, the two-dimensional Moyal plane.
 After renormalisation one obtains $R(z)=z+ \lambda \log(1+z)$ \cite{Panzer:2018tvy}.
 An integral representation for the planar 2-point function is only consistent
 for $\lambda >-\frac{1}{\log(4)}$.

\item Let $d \to \infty$ with spectral measure $\varrho(t)=t$, the four-dimensional Moyal plane:
 One finds $R(z)=z\, {}_2F_1(\alpha_{\lambda},1-\alpha_{\lambda},2;-z)$, where
 $\alpha_{\lambda}= \arcsin(\lambda \pi)/\pi$ for $|\lambda|\leq \frac{1}{\pi}$ and
 $\alpha_{\lambda}= \frac{1}{2}+\mathrm{i}\,\mathrm{arcosh}(\lambda \pi)/\pi$ for $\lambda \geq
 \frac{1}{\pi}$ \cite{Grosse:2019qps}. The singular value is $\lambda_s=-1/\pi$. Its mirror
 $\lambda_{\rm crit}=+\frac1{\pi}$ is a special transition point, where $\alpha_\lambda$ is continuous
 but not differentiable. However, $R(z)$ itself crosses smoothly over $\lambda_{\rm crit}$.

\end{itemize}
Beyond these special cases, we mostly leave the realm of exact
solutions.
Existence of solutions in a real or complex neighbourhood of $\lambda=0$ is guaranteed by
the implicit function theorem which constructs $2d$ functions
$\{\varepsilon_k(\lambda), \varrho_k(\lambda)\}_{k=1,\dots,d}$ from given data
$e_k, r_k= \lim\limits_{\lambda \to 0} (\varepsilon_k,\varrho_k)$.
For~a~first discussion we simplify the situation and take
$(\varepsilon_k,\varrho_k)$ as given data independently of~$\lambda$.
This ignores the condition $r_k \in \mathbb{Z}_{>0}$ (could be arbitrarily well
approximated for $N\to \infty$). We~will mostly consider the case $d=2$.

\subsection[Behaviour of the Ramification points for d=2]
{Behaviour of the Ramification points for $\boldsymbol {d=2}$}

The case $d=2$ describes a threefold covering and four ramification
points. We~scan the running of the ramification points $\beta_{1,2}$
and their complex conjugate by a variation of the coupling
constant. Because of
$-2\sum_{k=1}^{d} \varepsilon_k = \sum_{i=1}^{2d} \beta_i$
(underpinning the perturbative expansions), which is a consequence of
Vieta's theorem, the variations of the $\beta_i$ sum up to
zero. Figure~\ref{fig:running} shows the typical situation.
\begin{figure}[h!]
 \centering
 \includegraphics[width= 0.8\textwidth]{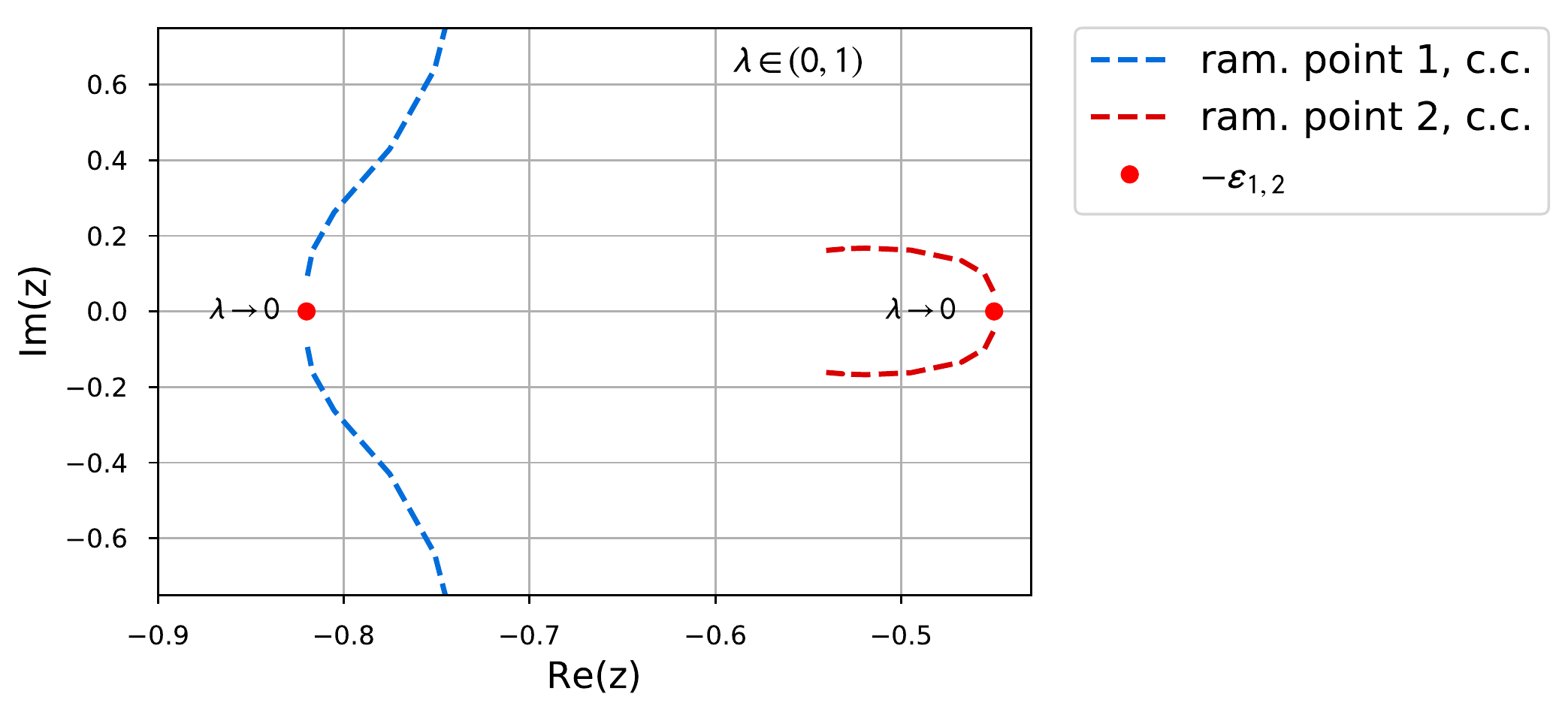} \vspace*{-2ex}
 \caption{For $\lambda \to 0$, the ramification points propagate into $-\varepsilon_1=-0.82$
 and $-\varepsilon_2=-0.45$ (these values for $\varepsilon_i$ are also chosen in Figures~\ref{fig:increasing} and
 \ref{fig:avocado},
 as well as $\varrho_1=1$, $\varrho_2=3$). In sum, the deformations average to zero. \label{fig:running}}
\end{figure}
We rediscover the square root-like behaviour (\ref{betaexpansion}) for small $\lambda$.

Taking $\varrho_1=\varrho_2$, we reenter the regime of analytically solvable
equations. Over and above, it shows a phenomenologically new behaviour:
\begin{Lemma}
 Given $d=2$ parameters $\varepsilon_1\neq \varepsilon_2$ and suppose their multiplicities
 arrange to $\varrho_1=\varrho_2=:N\varrho$. Then two ramification points merge to a single higher
 ramification point $\beta=\beta_1=\beta_2$ at the critical coupling constant
 $\lambda_{\rm crit}=\frac{(\varepsilon_1-\varepsilon_2)^2}{\varrho}$. Its real part is a
 fixed point $R(\mathrm{Re}(\beta))=\mathrm{Re}(\beta)$ of $R$.
 \begin{proof}
The value $\lambda_{\rm crit}$ is determined as follows: The four solutions of $R'(z)=0$ read
\begin{align*}
 \beta_{\pm,\pm}=
 \frac{1}{2}\Big ({-}\varepsilon_1-\varepsilon_2 \pm \sqrt{(\varepsilon_1+\varepsilon_2)^2
 -4\big [ \varepsilon_1\varepsilon_2+\lambda\varrho
 \pm \sqrt{-\lambda\varrho(\varepsilon_1-\varepsilon_2)^2+\lambda^2\varrho^2}\big ]} \Big).
\end{align*}
Then $\lambda=0$ and
$\lambda_{\rm crit} = \frac{(\varepsilon_1-\varepsilon_2)^2}{\varrho}$ are
the solution where two roots merge. Let $\beta_1$, $\beta_2$ be the
solutions in the upper half plane. Then
$\mathrm{Im}(\beta_1)=\mathrm{Im}(\beta_2)>0$ for
$\lambda <\lambda_{\rm crit}$ and
$\mathrm{Re}(\beta_1)=\mathrm{Re}(\beta_2)=-\frac{\varepsilon_1+\varepsilon_2}{2}
=:-\bar{\varepsilon}$ for $\lambda >\lambda_{\rm crit}$. For obvious reasons this value is a
fixed point of $R$, i.e., $R(-\bar{\varepsilon})=-\bar{\varepsilon}$.
This order-two ramification can be plotted as
in Figure~\ref{hrtr}.
\begin{figure}[h!t]
 \centering
 \includegraphics[width= 0.79\textwidth]{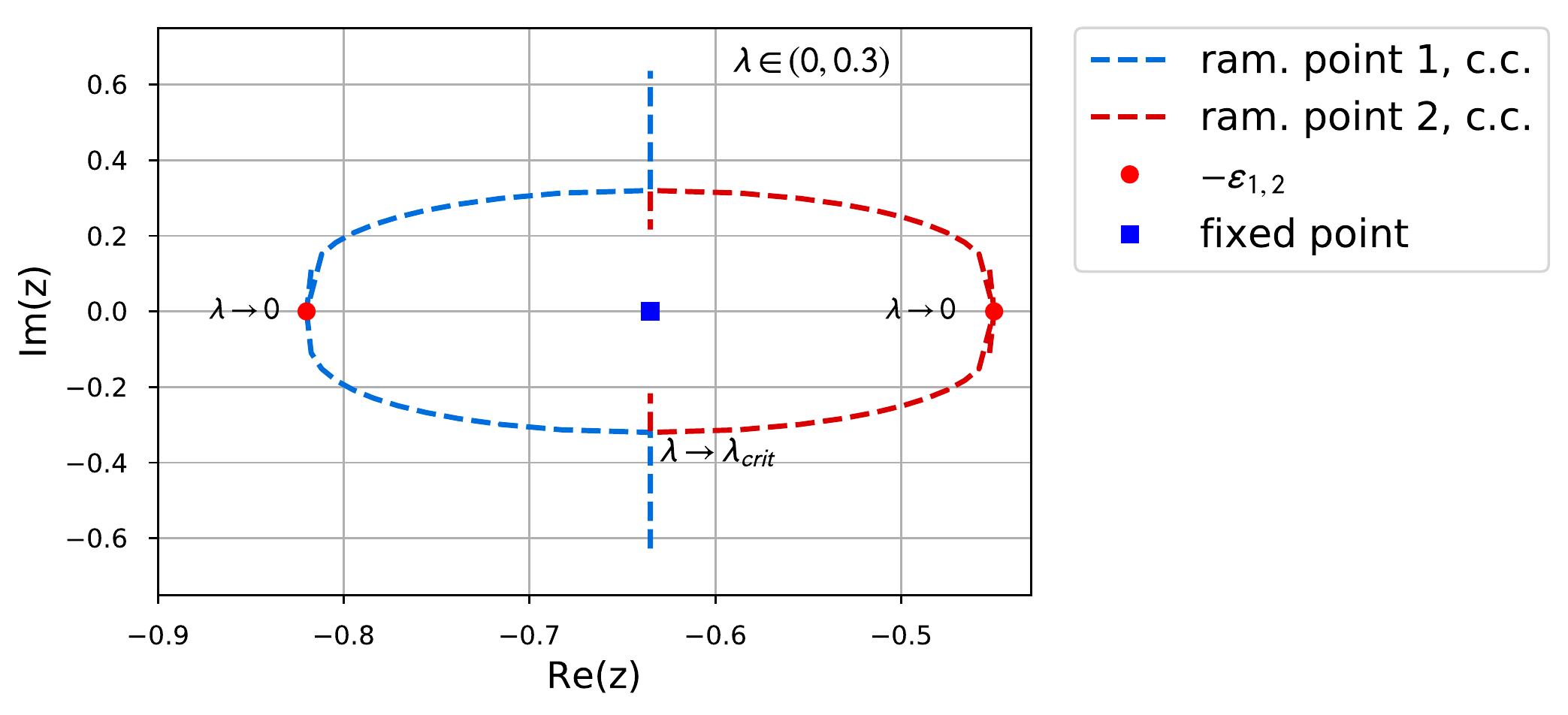}\vspace*{-2ex}
 \caption{Running of ramification points for $d=2$ and identical
 $\varrho:=\frac{\varrho_1}{N}=\frac{\varrho_2}{N}=2$. The critical coupling constant
 is $\lambda_{\rm crit}=\frac{(\varepsilon_1-\varepsilon_2)^2}{\varrho}$. For
 $\lambda>\lambda_{\rm crit}$ the ramification points have constant real
 part $-\frac{\varepsilon_1+\varepsilon_2}{2}$. At~$\lambda_{\rm crit}$ itself,
 topological recursion has to be modified to the variant of \textit{higher-order ramifications}.
 \label{hrtr}}
\end{figure}
\end{proof}
\end{Lemma}
We exemplify the possibility of smooth interpolations of $\lambda$
around the critical coupling $\lambda_{\rm crit}$ by applying the theory
of higher-order ramifications in topological recursion (intensely
studied~\cite{Bouchard:2012yg}) giving $ \omega_{0,3}[\lambda_{\rm crit}]$
compared with the limit
$\lim\limits_{\lambda \to \lambda_{\rm crit}} \omega_{0,3}[\lambda]$ of the
standard result of simple ramifications:

\begin{Example}[for general $d$]
 Consider two zeroes each of degree $d$ of $R'(z)$, named
 $\beta_{\pm}$. The~poles of
 $\omega_{0,3}(u,v,z)$ at $z+u=0$ and $z+v=0$ are not affected by higher-order
 rami\-fi\-cations. We~thus only concentrate on the part
 $\mathcal{P}_z \omega_{0,3}(u,v,z)$ of $\omega_{0,3}(u,v,z)$ which collects
 the residues at
 $z=\beta_{\pm}$. Apply \cite[Definition~3.6]{Bouchard:2012an} for the
 topological recursion of higher-order ramification to obtain
\begin{gather*}
\mathcal{P}_z \omega_{0,3}(u,v,z) = \Res\displaylimits_{q\to \beta_{\pm}}
\sum_{j=1}^d {\rm d}q \,K_2\big(z,q,\hat q^j\big) \omega_{0,2}(q,u) \omega_{0,2}\big(\hat q^j,v\big)
+ \{u \leftrightarrow v\},
\\
K_2(z,q,\hat q^j)= \frac{{\rm d}z}{2R'(\hat q^j)\big(R(-q)-R(-\hat q^j)\big)} \cdot \bigg(\frac{1}{z-q}-\frac{1}{z-\hat q^j} \bigg)
\end{gather*}
with the ordinary \textit{Eynard kernel} $K_2=K$ in the language of \cite{Bouchard:2012an} (generalised kernels are not necessary for $n=3$) and the preimages $\hat z^j$ of $R(z)=R\big(\hat z^j\big)$ with $\hat z^j \to \beta_{\pm}$ $\forall j$ for $z \to \beta_{\pm}$, $z \neq \hat z^j$. This gives rise to $d$-fold ramification at each fixed point. Next, we simplify the recursion kernel by expansion around the pole at $\beta_{\pm}$.
The residue for $ \mathcal{P}_z \omega_{0,3}(u,v,z)$ gives the same term for every $d$ summands (same fixed point of all involutions):
\begin{align*}
\mathcal{P}_z \omega_{0,3}(u,v,z) = d \times \frac{\omega_{0,2}(u,\beta_+)\omega_{0,2}(v,\beta_+)}{R''(\beta_+)R'(-\beta_+)(z-\beta_+)} +[\beta_+\leftrightarrow \beta_-].
\end{align*}
This is of course the same as if one would set $\beta_i=\beta_+$, $\beta_{i+d}=\beta_-$ $\forall i$ into \ref{propOm2} afterwards being the limit $\lim\limits_{\lambda \to \lambda_{\rm crit}} \mathcal{P}_z \omega_{0,3}(u,v,z) $.
\end{Example}
This is no accident: the formula was designed to give
continuity in the parameter which causes higher-order ramification,
i.e.,
$\lim\limits_{\lambda\to \lambda_{\rm crit}} \omega_{g,n}[\lambda]=
\omega_{g,n}[\lambda_{\rm crit}]$, and as said before, additional blob contributions are not affected by higher-order ramifications. Since our spectral curve is \textit{acceptable} in the sense of \cite[Definition~8]{Bouchard:2012yg}, everything works also for our model with blobs.

\subsection{Conformal mapping of the branch cuts}

Nikolai Zhukovsky found a suitable conformal map to solve the
potential flow of certain airfoils in an easier way
\cite{Zhukovsky:1910tvy}. It transforms an infinitely thin wing into a
circular one. During the analysis of the Hermitian 1-matrix model, one
recognised that this Zhukovsky transform occurs in the spectral curve
$x(z)$ and conformally maps the domain around the branch cut into the
exterior of the unit disk \cite[Section~3.1]{Eynard:2016yaa}. We~will
take this prime example to perform a more detailed analysis for our
$x(z)=R(z)$ with $d$ branch cuts and $d+1$ sheets. Let $d=2$ from now~on. We~choose to fix
$\varepsilon_{1,2}$ and $\varrho_{1,2}$ and pull the branch cuts back
into the $z$-plane~-- ending up with three preimages/sheets. This
procedure is sketched in Figure~\ref{fig:bcut}.
\begin{figure}[h!] 
\centering
 \includegraphics[width= 0.99\textwidth]{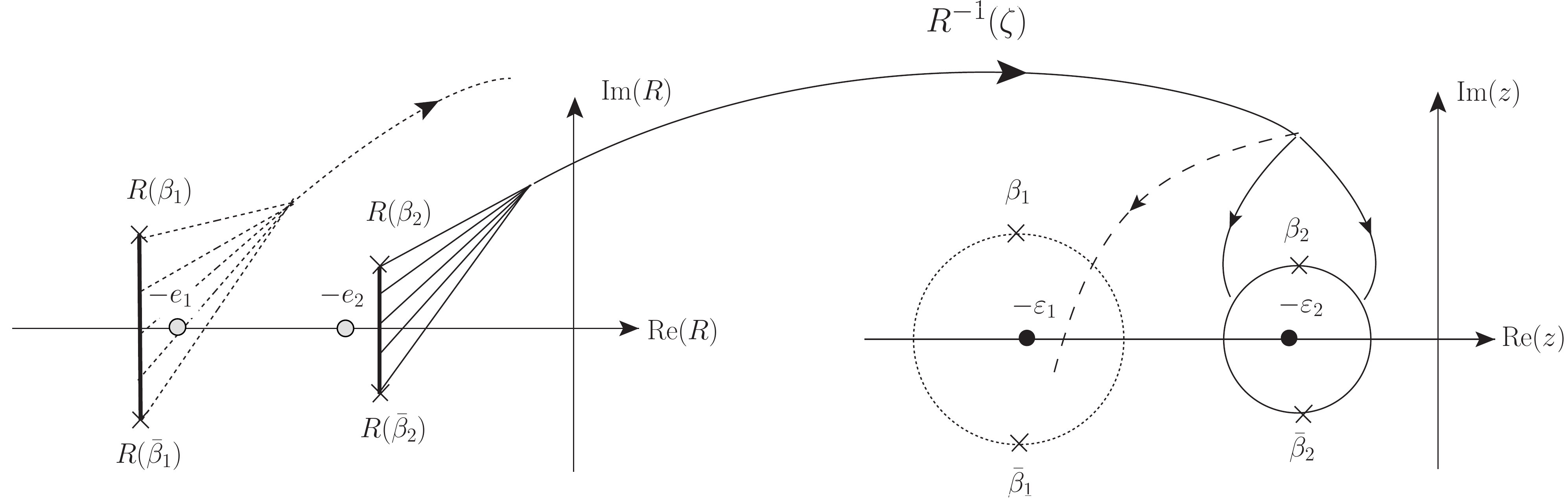}
 \caption{In the $R$-plane, we determine the branch cut to be the vertical connection
 between $R(\beta_i)$ and~$R(\bar{\beta}_i)$, with
 $\mathrm{Im}(\bar{\beta}_i)= -\mathrm{Im}(\beta_i)$. The inverse $R^{-1}$ pulls back
 the branch cut into the $z$-plane and~causes $d+1$ preimages for $d$ distinct
 values $e_k$. For small $\lambda$, small circles are generated.
 The remaining $d-1$ preimages of the cut are arcs located inside each of
 the other $d-1$ disks. They are not shown in~the~picture.
 We illustrated $d=2$.
 \label{fig:bcut}}
\end{figure}

 More formally: We map the domain
$\hat{\mathbb{C}} \setminus \{\Gamma_1 \cup \Gamma_2\}$ with
$\Gamma_i:= \big[R(\beta_i),R\big(\bar{\beta}_i\big)\big]$ as segments of~$\mathrm{i}\mathbb{R}$ into the exterior of the $\lambda$-deformed closed disks $\mathbb{D}_i$~-- the \emph{physical sheet}. In this sheet, we have a~biholomorphic map $R^{-1}\colon \hat{\mathbb{C}} \setminus \{\Gamma_1 \cup \Gamma_2\}
 \to
 \hat{\mathbb{C}} \setminus \{\mathbb{D}_1 \cup\mathbb{D}_2\}$
 sending $\infty$ to $\infty$.

Figure~\ref{fig:illu}
\begin{figure}[t!]\vspace*{1mm}
\centering
 \includegraphics[width= 0.79\textwidth]{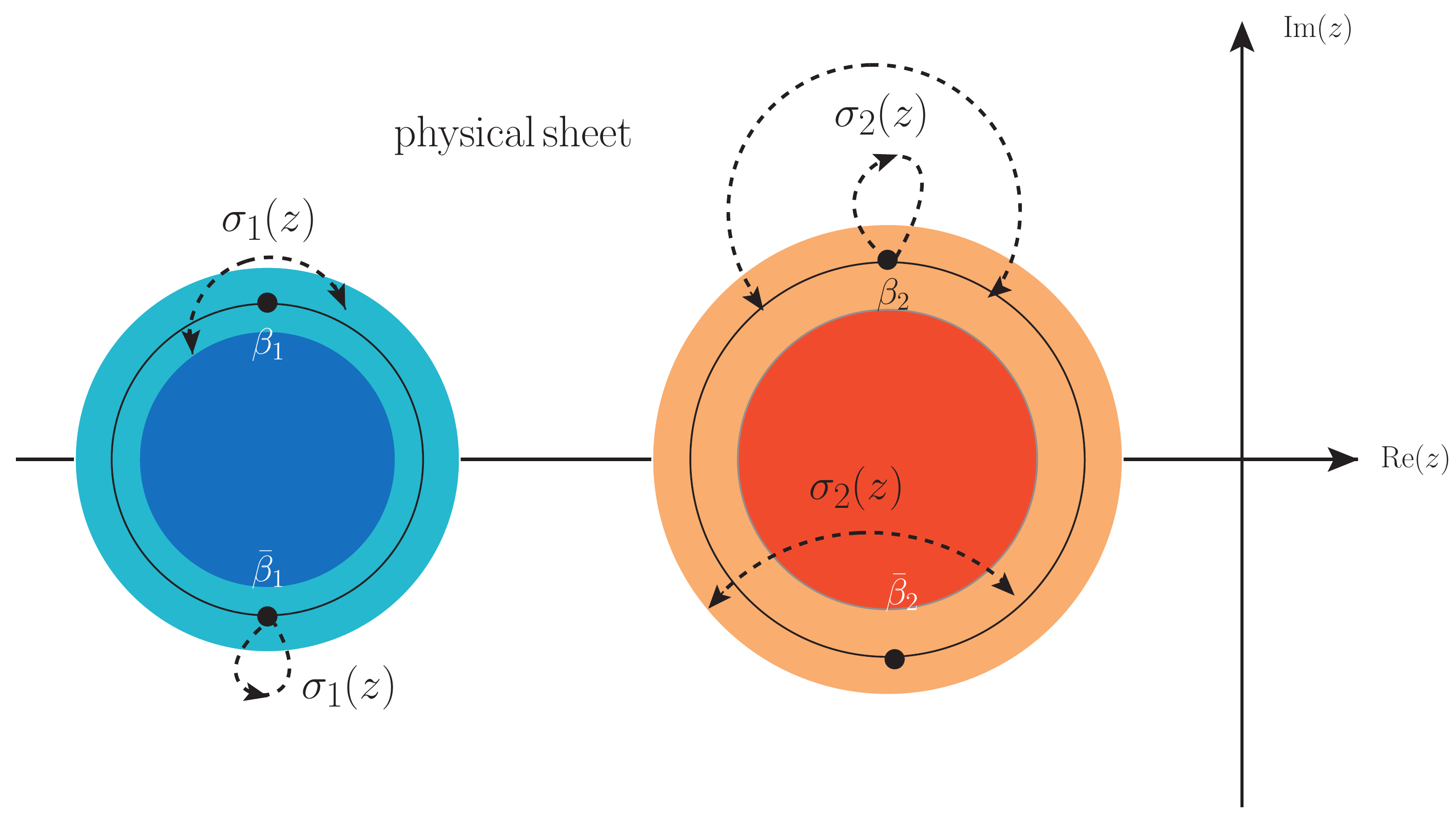}
 \caption{The preimage of
 $\hat{\mathbb{C}} \setminus \{\Gamma_1 \cup \Gamma_2\}$ under a ramified covering of degree 3
 distinguishes two (deformed) closed disks $\mathbb{D}_i$ in the $z$-plane. In a neighbourhood of
 their boundaries, the
 Galois involutions~$\sigma_{1,2}(z)$ allow to communicate with the physical sheet
 $\hat{\mathbb{C}} \setminus \{\mathbb{D}_1 \cup\mathbb{D}_2\}$. Their fixed
 points $\beta_i$, $\bar{\beta}_i$ mark north and south pole of $\mathbb{D}_i$.
 \label{fig:illu}}
 \end{figure}
illustrates the Galois involutions $\sigma_i(z)$, which are holomorphic
local involutions with fixed points $\beta_i$ and $\bar{\beta}_i$. They
fulfil $R(\sigma_i(z))=R(z)$ with $\sigma_i(z) \neq
\mathrm{id}$. These involutions (special deck transformations) are crucial to formulate
topological recursion and let the interior and exterior of the
deformed discs communicate.

 After this prelude, we continue with the numerical analysis and take
around 20 images in the $R$-plane along the branch cuts from
$R(\beta_i)$ to $R(\bar \beta_i)$ and map them with $R^{-1}$ into the
$z$-plane. A~first analysis with increasing coupling constant
$\lambda$ yields circle-like objects growing in radius and deformation
(Figure~\ref{fig:increasing}).
\begin{figure}[h!]
 \centering
 \includegraphics[width= 0.7\textwidth]{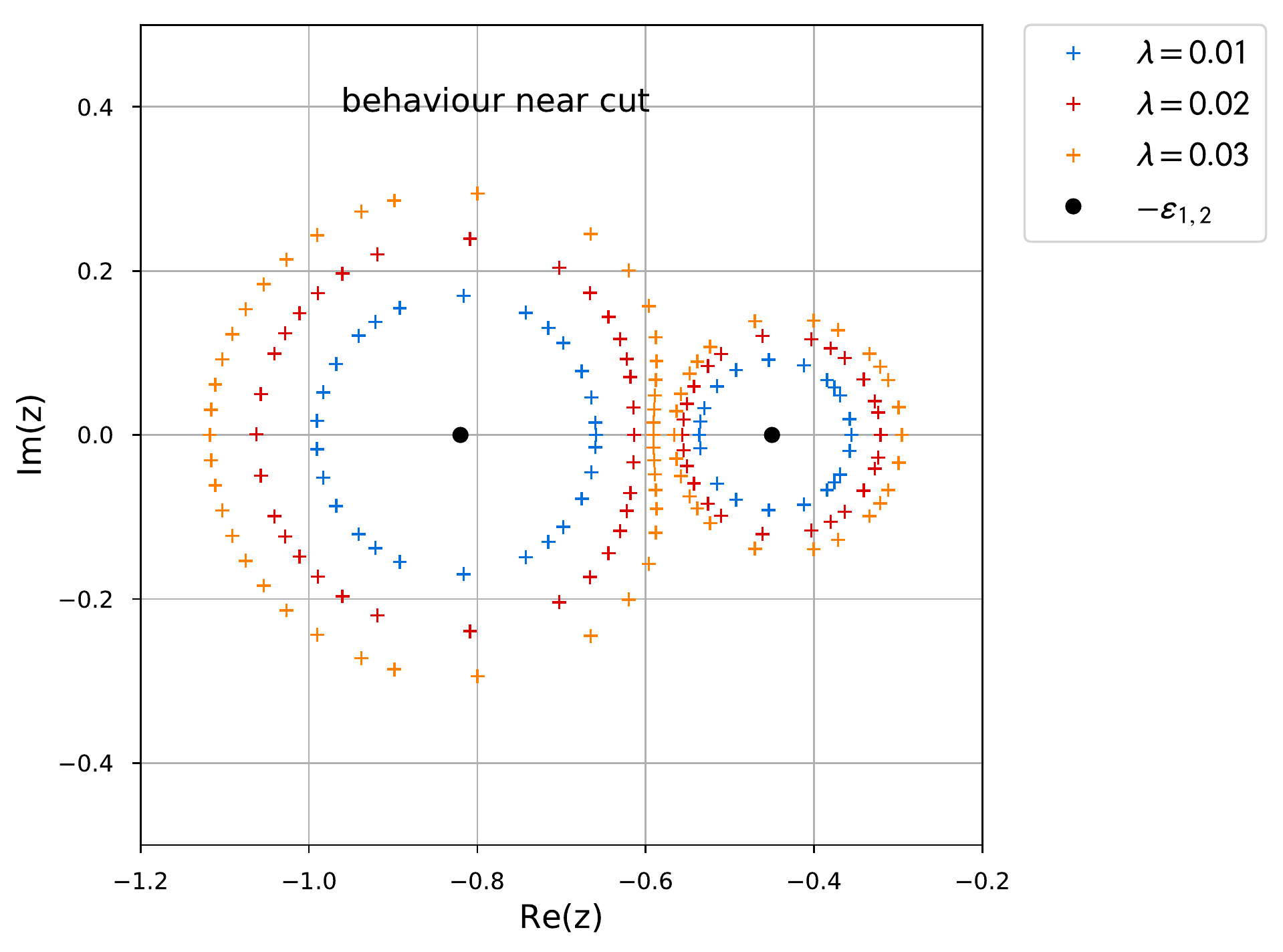}\vspace*{-3ex}
 \caption{We choose $-\varepsilon_1=-0.82$ and $-\varepsilon_2=-0.45$ and draw the preimages
 of a cut between $R(\beta_i)$ and $R\big(\overline{\beta_i}\big)$. The corresponding arcs
 from $z$ and $\sigma_i(z)$ form deformed circles. Their deformation increases with
 $\lambda$ and evolves by avoiding any intersection/collision of the two circles. A~larger
 gap between $\varepsilon_1$ and~$\varepsilon_2$ allows for stronger couplings before
 reaching a critical regime. The third preimage $\hat{z}$ (different from~$z$,~$\sigma_i(z)$
 forms an arc inside the other circle and is not given in this figure.}\label{fig:increasing}
\end{figure}

We observe that the radius of the
deformed circles is mainly determined by the multiplicity~$\varrho_k$.
The two branch cuts come closer to each other as $\lambda$ increases;
they merge at a critical value~$\lambda_{\rm crit}$. For
$\lambda>\lambda_{\rm crit}$ a stunning change of shape to an \emph{avocado
 plot} occurs, see Figure~\ref{fig:avocado}.

\begin{figure}[h!]
 \centering
 \includegraphics[width= 0.7\textwidth]{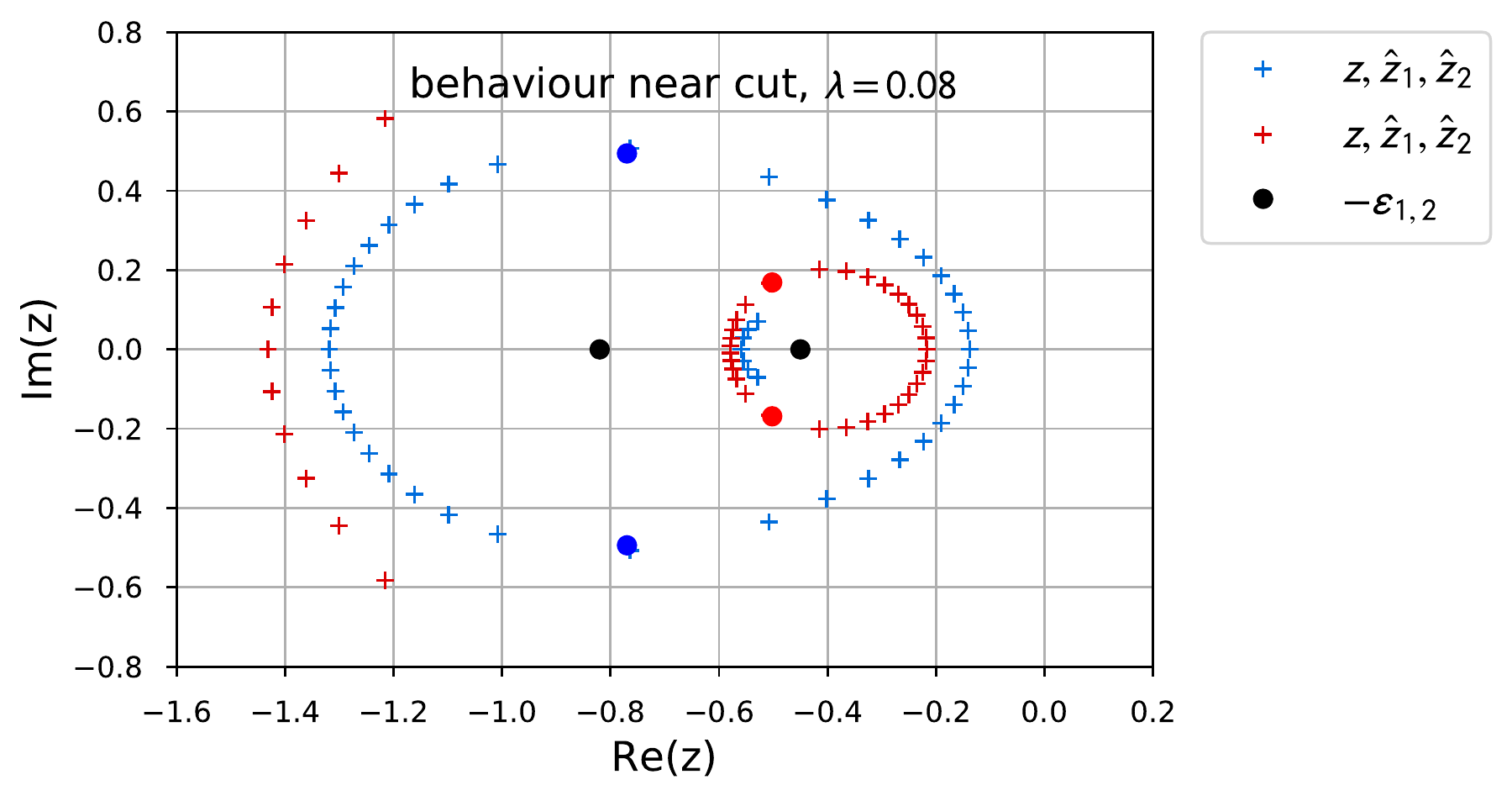}
 \caption{For $\lambda>\lambda_{\rm crit}$ a change of shape to an \textit{avocado} plot occurs. For the core, the
 local Galois involution communicates between core and flesh. The outer arc on the left
 is mapped by $R$ into regular values of the holomorphicity domain.}\label{fig:avocado}
\end{figure}

In the $z$-plane there is nothing particular at the critical value
$\lambda_{\rm crit}$. The ramification points are separate and simple (for pairwise different $\varrho_k$).
The solutions $\omega_{g,n}$ are analytic in $\lambda_{\rm crit}$ and
translate to preimages $\Omega^{(g)}_n(\zeta_1,\dots,\zeta_n)$ which for
$\zeta_i\in \mathcal{V}$ (see Figure~\ref{fig:complexification1}) are
also analytic in $\lambda_{\rm crit}$. What happens is the following. Fix
$\zeta_2,\dots,\zeta_n$ and assume $2g+n>0$. Then the function
$\zeta_1\mapsto \Omega^{(g)}_n(\zeta_1,\dots,\zeta_n)$ can be continued
to a larger domain $\tilde{\mathcal{V}} \supset \mathcal{V}$ which can
come close to $R(\beta_i)$. For $\lambda<\lambda_{\rm crit}$, any approach
$\zeta_1\to R(\beta_i)$ from inside $\tilde{\mathcal{V}}$ lets
$\Omega^{(g)}_n(\zeta_1,\dots,\zeta_n)$ approach $\infty$ for all $i=1,\dots,2d$.
For $\lambda\nearrow \lambda_{\rm crit}$ two pairs of divergent approaches come
close and eventually merge at $\lambda_{\rm crit}$. For
$\lambda>\lambda_{\rm crit}$ those $R(\beta_j)$ for which
$\beta_j$, $\overline{\beta_j}$ yield the core of the avodado become
regular values $R(\beta_j)\in \tilde{\mathcal{V}}$. This picture
generalises in obvious manner to any $d>2$, where several critical
values of $\lambda$ occur at which the discs swallow each other.
Figure~\ref{fig:d3}
shows several snapshots for~$d=3$.

\begin{figure}[h!]
 \centering
 \includegraphics[width= 0.9\textwidth]{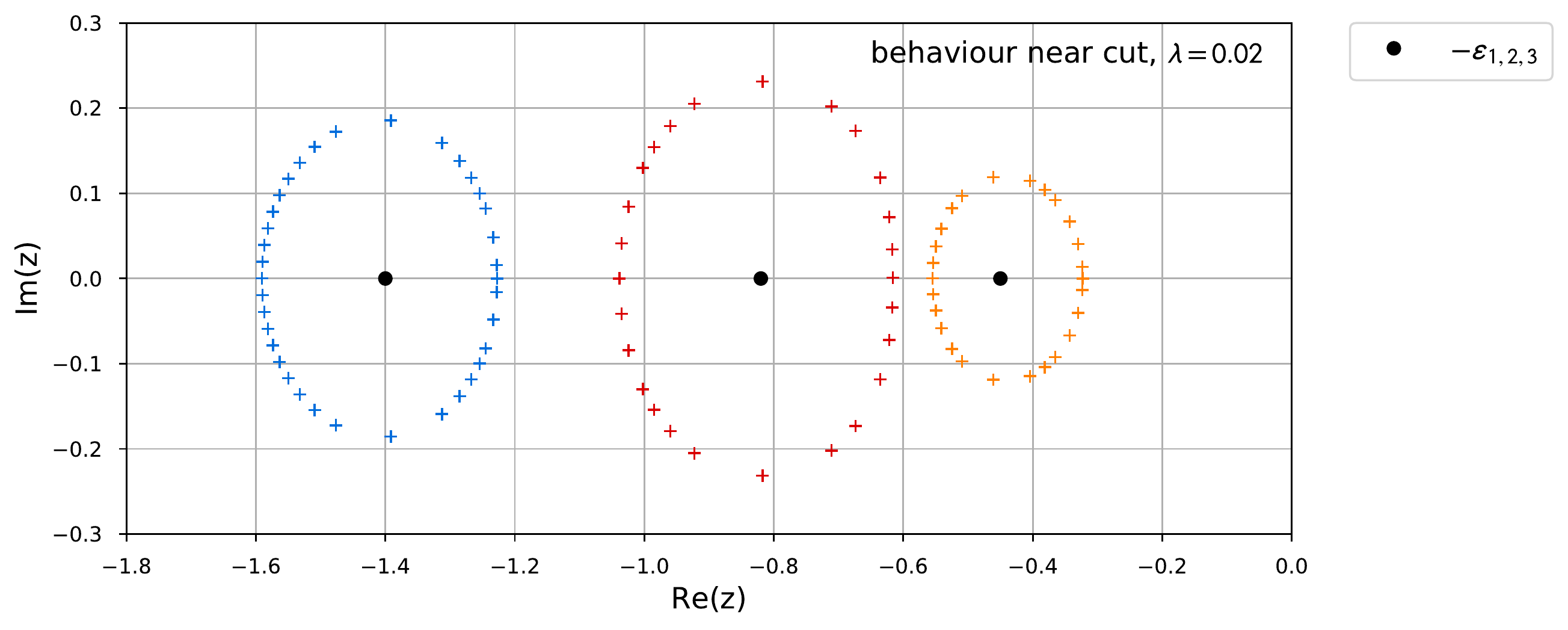}
 \includegraphics[width= 0.9\textwidth]{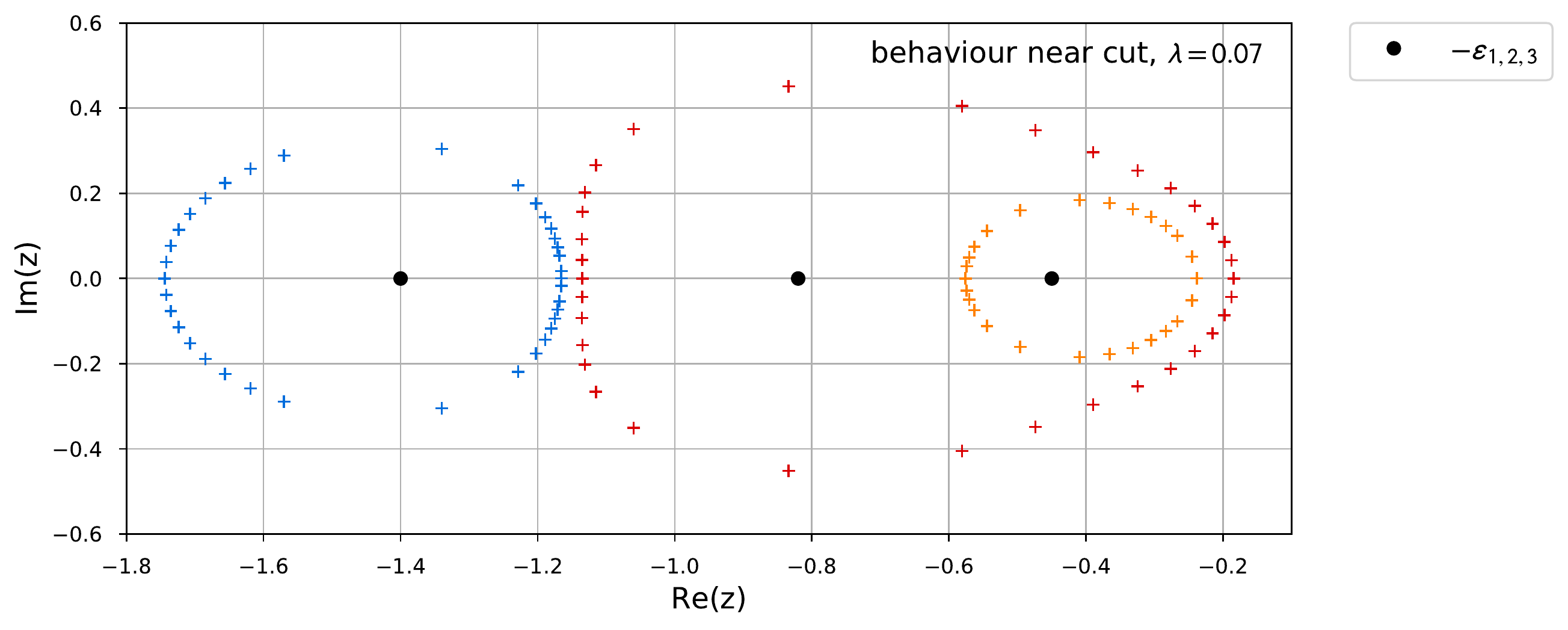}
 \includegraphics[width= 0.9\textwidth]{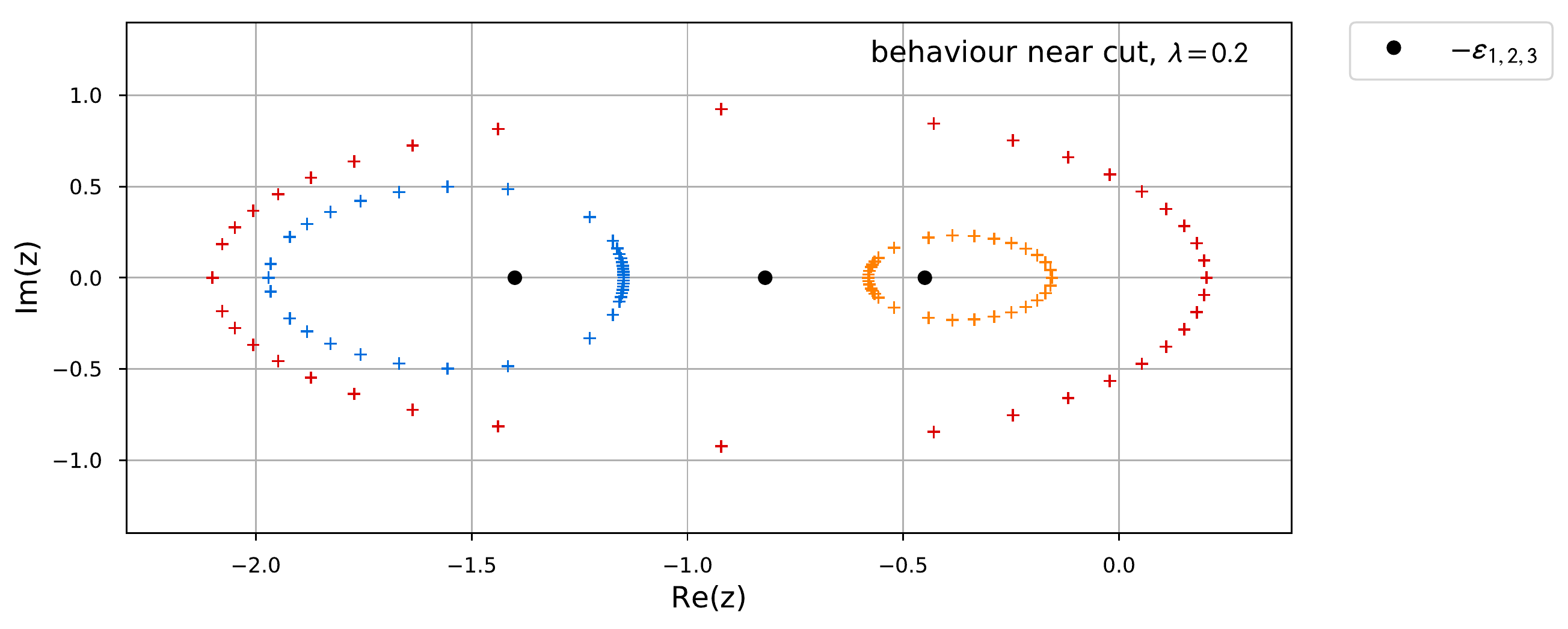}\vspace*{-1ex}
 \caption{We investigated $d=3$ with $(\varepsilon_1,\varrho_1)=(0.45,1)$, $(\varepsilon_2,\varrho_2)=(0.82,3)$
 and $(\varepsilon_3,\varrho_3)=(1.40,2)$. We~see aforementioned $d-1=2$ transitions: At $\lambda=0.02$,
 we observe a standard behaviour with three branch cuts. The biggest circle swallows the smallest
 after a certain threshold value $\lambda_1$, this process is finished at $\lambda=0.07$. In the next
 transition the second-smallest circle is swallowed at $\lambda_2$.\label{fig:d3}}
\end{figure}

There remains the interesting special case of
$\varrho_1=\varrho_2=\dots =\varrho_d$, where the above picture combines
with higher-order ramification: Which circle will swallow which? We
only mention that there is a multitude of interesting phenomena to be
discovered.

\section{Conclusion}\label{sec:conclusion}

The quartic analogue of the Kontsevich model offers exceptional
possibilities to study structures in quantum field theory. It is a
Euclidean quantum field theory defined by deformation of a Gau\ss{}ian
measure. This allows on one hand to derive Dyson--Schwinger equations
between the correlation functions, on the other hand to represent
these functions as a series in Feynman (ribbon) graphs. What makes
this model particular is the possibility to exactly solve the
Dyson--Schwinger equations in terms of algebraic or special
functions. In this paper we explored the prospects of these
achievements for the series of Feynman graphs and investigated
transitions between different singularity types when varying the
coupling constant.

After these general remarks let us be more precise about what is
achieved and what is left for the future. One of the most important
aspects of quantum field theory is renormalisation, which entails
beautiful mathematical structures \cite{Connes:1999yr, Kreimer:1997dp}.
Renormalisation is relevant for systems with infinitely many degrees
of freedom. Our model can be extended to infinitely many degrees of
freedom; the Dyson--Schwinger equations relate already renormalised
correlation functions. The~initial non-linear Dyson--Schwinger equation
has been solved implicitly \cite{Grosse:2019jnv}, but in full
generality. An explicit solution in terms of special functions
succeeded for the 2D Moyal space (where it gives the Lambert function
\cite{Panzer:2018tvy}) and for the 4D Moyal space (where it gives the
inverse of a Gau\ss{} hypergeometric function
\cite{Grosse:2019qps}). In these two cases all the renormalised
correlation functions of disk topology can be written down (thanks to \cite{DeJong}) as
integral representations. Expanding them produces the familiar
number-theoretical structures of quantum field theory such as multiple
zeta values \cite{Broadhurst:1996kc} and hyperlogarithms. We~remark
that the Kontsevich model itself \cite{Kontsevich:1992ti} can be
treated in a similar manner \cite{Grosse:2019nes}, but the expansion
gives at most logarithms.

In this paper we focus on the non-planar sector of the quartic
Kontsevich model. Although renormalisation is not needed, the limit to
infinitely many degrees of freedom is not yet understood and needs to
be studied in the future. All our results apply to a
finite-dimensional approximation by $(N\times N)$-matrices. We~proved in~\cite{Branahl:2020yru,Hock:2021tbl} that all correlation functions are
affi\-li\-ated with a family $\omega_{g,n}(z_1,\dots,z_n)$ of meromorphic
forms which can be explicitly computed by residue techniques. This
evaluation becomes increasingly complicated for large $(g,n)$, but the
results are remarkably structured and simple. We~were led in~\cite{Branahl:2020yru} to the conjecture that the $\omega_{g,n}$
follow blobbed topological recursion \cite{Borot:2015hna}, i.e., the
poles of $z_i\mapsto \omega_{g,n}(z_1,\dots,z_n)$ at ramification points
of $R$ are given by a universal formula. The function $R$ governs the
solution \cite{Grosse:2019jnv, Schurmann:2019mzu-v3} of the non-linear
Dyson--Schwinger equation.

This paper extends \cite{Branahl:2020yru} in expressing the
coefficients of the $\omega_{g,n}$ as distinguished polynomials in the
correlation functions of the quartic Kontevich model (see
Propositions~\ref{prop:Om02G} and \ref{prop:Om03G}). These distinguished
polynomials thus evaluate to expressions much simpler than a
correlation function itself (and than any of the factorially many
contributing Feynman ribbon graph, see Section~\ref{sec:combinatorics}
for their numbers). To unveil this simplicity it was necessary to
transform with the inverse of the central function~$R$ (see
Section~\ref{sec:comparison}). We~remark that the appearence of the
distinguished polynomials is in striking contrast to the Kontsevich
model \cite{Kontsevich:1992ti} in which the $(1+\dots + 1)$-point correlation functions
themselves follow topological recursion (see \cite[Chapter~6]{Eynard:2016yaa} and
\cite{Grosse:2019nes}). Moreover, the analogue of $R$
in the Kontsevich model is the function $x(z)=z^2+\text{const}$ with a~single ramification point at $z=0$. We~have shown in
Section~\ref{sec:critical} that the dependence of $R$ on the coupling
constant leads in the quartic Kontsevich model to a very rich
landscape of branch cuts with merge at critical values of the coupling
constant. Of course these phenomena are only accessible because we
have exact non-perturbative solutions.

After all we have seen that the quartic Kontsevich model shares many
features with honest quantum field theories: perturbative expansion
into Feynman graphs, non-perturbative formulation via Dyson--Schwinger
equations, renormalisation, evaluation into number-theoretical
functions. The exact solution found step by step in~\cite{Branahl:2020yru,Grosse:2019jnv, Panzer:2018tvy,
 Schurmann:2019mzu-v3} permits to identify and to
explore quantum field-theoretical structures which previously were
hidden. Of course these structures could be special to the quartic
Kontsevich model. Nonetheless we find it worthwhile to investigate
whether something similar could be present also in realistic quantum
field theories such as the standard model. Two questions deserve
particular attention:
\begin{itemize}\itemsep=0pt
\item Is it possible to trace a part of the complexity in QFT back to a change of variables via
 the complicated inverse of a relatively simple function $R$?

\item Can one collect combinations of the $R^{-1}$-transformed correlation functions
 to much simpler functions of topological significance?
\end{itemize}

\subsection*{Acknowledgements}

It is a pleasure to dedicate this paper to Dirk Kreimer who supported
this research project in a~substantial way: The groundwork \cite{Panzer:2018tvy} has been
laid during the Les Houches 2018 summer school ``Structures in local quantum field theories''
organised by Spencer Bloch and Dirk Kreimer.
AH and RW would like to thank Karen Yeats and Erik Panzer for invitation to present
our results at the IH\'ES remote conference ``Algebraic Structures in Perturbative Quantum Field Theory''
in honour of Dirk Kreimer's 60th birthday.
Our work was supported\footnote{Funded by
 the Deutsche Forschungsgemeinschaft (DFG, German Research
 Foundation)~-- Project-ID 427320536 - SFB 1442, as well as under
 Germany's Excellence Strategy EXC 2044 390685587, Mathematics
 M\"unster: Dynamics~-- Geometry~-- Structure.} by the Cluster of
Excellence \emph{Mathematics M\"unster} and the CRC 1442 \emph{Geometry:
 Deformations and Rigidity}. AH is supported through the Walter Benjamin fellowship.\footnote{Funded by
 the Deutsche Forschungsgemeinschaft (DFG, German Research
 Foundation)~-- Project-ID 465029630.}

\pdfbookmark[1]{References}{ref}
\LastPageEnding

\end{document}